\pdfoutput=1
\documentclass[a4paper,UKenglish,cleveref, autoref, thm-restate]{lipics-v2021}
\nolinenumbers

\hideLIPIcs

\usepackage[cmtip,all]{xy} 
\usepackage{amsmath}

\makeatletter
\renewcommand*\env@matrix[1][*\c@MaxMatrixCols c]{%
  \hskip -\arraycolsep
  \let\@ifnextchar\new@ifnextchar
  \array{#1}}
\makeatother

\usepackage{mathtools}
\usepackage{physics}    
\usepackage{amssymb}
\usepackage{amsfonts}
\usepackage{mathdots}   
\usepackage{mathrsfs}   
\usepackage{nicefrac}   
\usepackage{stmaryrd}
\usepackage{tikz-cd}     
\usepackage{xparse}     
\usepackage{nicefrac}

\usepackage{amsfonts}
\usepackage{mathtools}

\usepackage[inline]{enumitem}

\usepackage{bold-extra}

\usepackage{xcolor}

\usepackage{xspace}

\usepackage{tikz}
\usetikzlibrary{decorations.markings}
\usepackage{tikzit}
\usepackage{float}
\usepackage{circuitikz}
\usetikzlibrary{shadows,fadings}
\ctikzset{bipoles/length=.8cm}

\input{styles/symplectic-ZX.tikzdefs}
\tikzstyle{graphlabel}=[font={\scriptsize\boldmath}, inner sep=1mm, outer sep=-1.8mm, scale=0.8]
\tikzstyle{graphlines}=[dashed,color=gray]
\tikzstyle{wn}=[font={\scriptsize\boldmath}, inner sep=1mm, outer sep=-1.8mm, scale=\spiderscaling, tikzit shape=circle, draw=black, fill=black!01, tikzit fill=white, tikzit draw=black, shape=circle, tikzit category=GLA]
\tikzstyle{bn}=[font={\scriptsize\boldmath}, inner sep=1mm, outer sep=-1.8mm, scale=\spiderscaling, tikzit shape=circle, draw=black, fill={rgb,255: red,100; green,100; blue,100}, tikzit draw=black, shape=circle, tikzit category=GLA]
\tikzstyle{gn}=[style=bn, tikzit fill=gray]
\tikzstyle{rn}=[style=wn, tikzit fill=white]
\tikzstyle{grn}=[style=gwn, tikzit fill=white]
\tikzstyle{ggn}=[style=gbn, tikzit fill={zx_grey}]
\tikzstyle{gwn}=[shading=whiteballshading, line width=1pt, inner sep=1mm, outer sep=-1.8mm, scale=\spiderscaling, tikzit shape=circle, draw=black, fill=white, tikzit fill=white, tikzit draw=black, shape=circle, tikzit category=GLA]
\tikzstyle{gbn}=[shading=blackballshading, font={\scriptsize\boldmath}, line width=1pt, inner sep=1mm, outer sep=-1.8mm, scale=\spiderscaling, tikzit shape=circle, draw=black, fill={rgb,255: red,100; green,100; blue,100}, tikzit draw=black, shape=circle, tikzit category=GLA]
\tikzstyle{gvoidshape}=[shading=hadballshading, line width=1pt, inner sep=1mm, outer sep=-1.8mm, scale=\spiderscaling, tikzit shape=circle, draw=black, fill={zx_grey}, tikzit fill=white, tikzit draw=black, shape=circle, tikzit category=GLA]
\tikzstyle{voidshape}=[font={\scriptsize\boldmath},inner sep=1mm, outer sep=-1.8mm, scale=\spiderscaling, tikzit shape=circle, draw=black, fill={zx_grey}, tikzit draw=black, shape=circle, tikzit category=GLA]
\tikzstyle{void}=[style=voidtemp, tikzit shape=circle, tikzit fill=green, tikzit draw=black, tikzit category=GLA]
\tikzstyle{gvoid}=[style=gvoidtemp, tikzit shape=circle, tikzit fill=green, tikzit draw=black, tikzit category=GLA]
\tikzstyle{had}=[fill={zx_grey}, draw=black, shape=rectangle, tikzit category=ZX, tikzit draw=black, minimum size=5pt, inner sep=1.5pt, scale=\boxscaling, font={\scriptsize\boldmath}]
\tikzstyle{ghad}=[shading=hadballshading, fill={zx_grey}, draw=black, shape=rectangle, tikzit category=ZX, tikzit draw=black, thick, minimum size=5pt, inner sep=1.5pt,scale=\boxscaling, font={\scriptsize\boldmath}]
\tikzstyle{wphase}=[
  node on layer= frontlayer,
  shade, shading=axis, left color=white!100, right color=black!5.5, shading angle=-45,
  drop shadow={blend mode=color burn, left color=black, opacity=.015, shading angle=-45,right color=black!75, shadow xshift=-0.40pt, shadow yshift=-0.45pt, shadow scale=1.020}, 
  drop shadow={blend mode=color burn, left color=black, opacity=.015, shading angle=-45,right color=black!75, shadow xshift=-0.40pt, shadow yshift=-0.45pt, shadow scale=1.015}, 
  drop shadow={blend mode=color burn, left color=black, opacity=.015, shading angle=-45,right color=black!75, shadow xshift=-0.40pt, shadow yshift=-0.45pt, shadow scale=1.010}, 
  rectangle, rounded corners, ultra thin, draw opacity=1, draw=black!50,  inner sep=2pt, font={\scriptsize\boldmath}, label distance=1mm, text opacity=1,minimum height=.3cm,minimum width=.3cm, fill=white, fill opacity=.70, 
  tikzit category=ZX, tikzit fill=white, tikzit draw=white, 
  scale=\labelscaling
]
\tikzstyle{bphase}=[node on layer= frontlayer, 
  shade, shading=axis,fill=black, left color=black!75, right color=black!100, shading angle=-45,
  drop shadow={blend mode=color burn, left color=black, opacity=.03, shading angle=-45,right color=black!75, shadow xshift=-0.40pt, shadow yshift=-0.45pt, shadow scale=1.020}, 
  drop shadow={blend mode=color burn, left color=black, opacity=.03, shading angle=-45,right color=black!75, shadow xshift=-0.40pt, shadow yshift=-0.45pt, shadow scale=1.015}, 
  drop shadow={blend mode=color burn, left color=black, opacity=.03, shading angle=-45,right color=black!75, shadow xshift=-0.40pt, shadow yshift=-0.45pt, shadow scale=1.010}, 
  rectangle, rounded corners, ultra thin, draw opacity=1, fill opacity=.3, draw=black!70, inner sep=2pt, font={\scriptsize\boldmath}, label distance=1mm, text opacity=1,minimum height=.3cm,minimum width=.3cm, text=white, 
  tikzit category=ZX, tikzit fill=black, tikzit draw=white, 
  scale=\labelscaling
]
\tikzstyle{mphase}=[
  node on layer= frontlayer,
  shade, shading=axis,fill=gray, left color=white!100, right color=gray!45, shading angle=-45,
  drop shadow={blend mode=color burn, left color=black, opacity=.025, shading angle=-45,right color=black!75, shadow xshift=-0.40pt, shadow yshift=-0.45pt, shadow scale=1.020}, 
  drop shadow={blend mode=color burn, left color=black, opacity=.025, shading angle=-45,right color=black!75, shadow xshift=-0.40pt, shadow yshift=-0.45pt, shadow scale=1.015}, 
  drop shadow={blend mode=color burn, left color=black, opacity=.025, shading angle=-45,right color=black!75, shadow xshift=-0.40pt, shadow yshift=-0.45pt, shadow scale=1.010}, 
  rectangle, rounded corners, ultra thin, draw opacity=1, draw=black!60, inner sep=2pt, font={\scriptsize\boldmath}, label distance=1mm, fill opacity=.35, text opacity=1,minimum height=.3cm,minimum width=.3cm, 
  tikzit category=ZX, tikzit fill=gray, tikzit draw=white, 
  scale=\labelscaling
]
\tikzstyle{gphase}=[style=bphase,tikzit category=ZX, tikzit fill=black, tikzit draw=white]
\tikzstyle{rphase}=[style=wphase,tikzit category=ZX, tikzit fill=white, tikzit draw=white]
\tikzstyle{scalar}=[mphase]
\tikzstyle{lmat}=[shape=signal, signal to=west, signal from=east, fill={zx_grey}, draw=black, minimum height=6pt, inner sep=1pt, font={\scriptsize\boldmath}, tikzit fill=gray, tikzit category=GLA, anchor=center, outer sep=-.1cm, signal pointer angle=\arrowangle,scale=\arrowscaling]
\tikzstyle{rmat}=[shape=signal, signal to=east, signal from=west, fill={zx_grey}, draw=black, minimum height=6pt, inner sep=1pt, font={\scriptsize\boldmath}, tikzit fill=gray, tikzit category=GLA, anchor=center, outer sep=-.1cm, signal pointer angle=\arrowangle,scale=\arrowscaling]
\tikzstyle{dmat}=[shape=signal, signal to=east, signal from=west, fill={zx_grey}, draw=black, minimum height=6pt, inner sep=1pt, font={\scriptsize\boldmath}, tikzit fill=gray, tikzit category=GLA, rotate=270, anchor=center, outer sep=-.1cm, signal pointer angle=\arrowangle,scale=\arrowscaling]
\tikzstyle{umat}=[shape=signal, signal to=east, signal from=west, fill={zx_grey}, draw=black, minimum height=6pt, inner sep=1pt, font={\scriptsize\boldmath}, tikzit fill=gray, tikzit category=GLA, rotate=90, anchor=center, outer sep=-.1cm, signal pointer angle=\arrowangle,scale=\arrowscaling]
\tikzstyle{lmatt}=[shading=matballshading,shape=signal, signal to=west, signal from=east, fill={zx_grey}, draw=black, minimum height=6pt, inner sep=1pt, font={\scriptsize\boldmath}, tikzit fill=gray, tikzit category=GLA, anchor=center, outer sep=-.1cm, thick, signal pointer angle=\arrowangle,scale=\arrowscaling]
\tikzstyle{rmatt}=[shading=matballshading,shape=signal, signal to=east, signal from=west, fill={zx_grey}, draw=black, minimum height=6pt, inner sep=1pt, font={\scriptsize\boldmath}, tikzit fill=gray, tikzit category=GLA, anchor=center, outer sep=-.1cm, thick, signal pointer angle=\arrowangle,scale=\arrowscaling]
\tikzstyle{dmatt}=[shading=matballshading,shape=signal, signal to=east, signal from=west, fill={zx_grey}, draw=black, minimum height=6pt, inner sep=1pt, font={\scriptsize\boldmath}, tikzit fill=gray, tikzit category=GLA, rotate=270, anchor=center, outer sep=-.1cm, thick, signal pointer angle=\arrowangle,scale=\arrowscaling]
\tikzstyle{umatt}=[shading=matballshading, outer color={zx_grey_thick}, inner color={zx_grey},shape=signal, signal to=east, signal from=west, fill={zx_grey}, draw=black, minimum height=6pt, inner sep=1pt, font={\scriptsize\boldmath}, tikzit fill=gray, tikzit category=GLA, rotate=90, anchor=center, outer sep=-.1cm, thick, signal pointer angle=\arrowangle,scale=\arrowscaling]
\tikzstyle{graph_vertex}=[fill=black, draw=black, shape=circle, tikzit category=mbqc, minimum size=2.4mm, inner sep=.8mm]
\tikzstyle{graph_weight}=[fill=white, draw=none, shape=rectangle, tikzit category=mbqc, inner sep=2pt, scale=.8]
\tikzstyle{graph_state}=[fill=yellow, draw=none, shape=rectangle, tikzit category=ZX, rounded corners=1.3mm, minimum height=1.7cm, minimum width=1.3cm, opacity=.7, text opacity=1]
\tikzstyle{box}=[fill=white, draw=black, shape=rectangle, inner sep=2.5pt]
\tikzstyle{wirelable}=[node on layer= labeltextlayer, font={\scriptsize},scale=.9,text=dark_grey, fill=none, inner sep=1pt]
\tikzstyle{tightwirelable}=[node on layer= labeltextlayer, text=dark_grey, fill=white, inner sep=0pt]
\tikzstyle{gather}=[scale=.8,shading=gatherballshading, outer color={zx_grey_thick}, inner color={zx_grey}, fill={zx_grey}, draw=black, tikzit category=scal, rounded corners=0.8mm, regular polygon, regular polygon sides=3, shape border rotate=-90, inner sep=1.6pt, anchor=center, outer sep=-.1cm, line width=0.75 pt]
\tikzstyle{divide}=[scale=.8,shading=divideballshading, outer color={zx_grey_thick}, inner color={zx_grey}, regular polygon, regular polygon sides=3, shape border rotate=90, draw=black, fill={zx_grey}, inner sep=1.6pt, tikzit category=scal, rounded corners=0.8mm, anchor=center, outer sep=-.1cm,line width=0.75 pt]
\tikzstyle{plus}=[inner sep=2.5pt, draw, circle, path picture={ \draw[black](path picture bounding box.east) -- (path picture bounding box.west) (path picture bounding box.south) -- (path picture bounding box.north);}, fill=white]
\tikzstyle{dot}=[thick, fill=black, circle, scale=1, inner sep=.05cm]
\tikzstyle{origin}=[thick, fill=black, circle, scale=.5, inner sep=.05cm]
\tikzstyle{dash_edge}=[-, dashed]
\tikzstyle{hadamard_edge}=[-, dashed, dash pattern=on 2pt off 1.5pt, thick, draw=blue]
\tikzstyle{brace edge}=[-, tikzit draw=blue, decorate, decoration={brace,amplitude=1mm,raise=-1mm}]
\tikzstyle{ultra thin}=[-, line width=0.03 pt]
\tikzstyle{thin}=[-, line width=0.5 pt]
\tikzstyle{thick}=[-, line width=.8pt, tikzit draw=red]
\tikzstyle{very thick}=[-, line width=1.05pt, tikzit draw=red]
\tikzstyle{multiplexer}=[-, fill={rgb,255: red,179; green,179; blue,179}]
\tikzstyle{gmultiplexer}=[-, line width=1pt, tikzit draw=red, fill={rgb,255: red,179; green,179; blue,179}]
\tikzstyle{background}=[-, fill={rgb,255: red,220; green,220; blue,220}, draw=none, tikzit draw={rgb,255: red,128; green,128; blue,128}, on layer= backlayer]
\tikzstyle{white}=[-, fill=white, draw=none, tikzit fill=white]
\tikzstyle{blue_line}=[-, draw=blue]
\tikzstyle{bbox}=[-, fill=white]
\tikzstyle{shadedpolygon}=[rounded corners=0.2mm, opacity=1, fill=gray, fill opacity=0.5]
\tikzstyle{arrow}=[->]

\input{styles/styles-pbs.tikzdefs}

\tikzstyle{diamant}=[diamond, fill=couleurdefond, draw=black]
\tikzstyle{newe}=[rectangle, fill={gray!15}, draw=black, tikzit shape=rectangle, inner sep=0.2em]
\tikzstyle{cercle}=[circle, fill=couleurdefond, draw=black]
\tikzstyle{scercle}=[circle, fill=couleurdefond, draw=black, tikzit fill=white, inner sep=0.1em]
\tikzstyle{cartouche}=[rounded rectangle, fill=couleurdefond, draw=black]
\tikzstyle{neg}=[rounded rectangle, fill=couleurdefond, draw=black, execute at end node={$\neg$}]
\tikzstyle{sneg}=[rounded rectangle, fill=couleurdefond, draw=black, execute at end node={$\neg$}, scale=0.8]
\tikzstyle{negserie}=[rounded rectangle, fill=couleurdefond, draw=black, execute at end node={\footnotesize$\star\star$}]
\tikzstyle{diagrammevide}=[rectangle, fill=couleurdefond, draw=black, inner sep=1.25em, borddiagrammevide, tikzit shape=rectangle]
\tikzstyle{mdiagrammevide}=[rectangle, fill=couleurdefond, draw=black, inner sep=0.75em, sborddiagrammevide, tikzit shape=rectangle]
\tikzstyle{msdiagrammevide}=[rectangle, fill=couleurdefond, draw=black, inner sep=0.7em, msborddiagrammevide, tikzit shape=rectangle]
\tikzstyle{sdiagrammevide}=[rectangle, fill=couleurdefond, draw=black, inner sep=0.5em, sborddiagrammevide, tikzit shape=rectangle]
\tikzstyle{xsdiagrammevide}=[rectangle, fill=couleurdefond, draw=black, inner sep=0.4em, xsborddiagrammevide, tikzit shape=rectangle]
\tikzstyle{bs}=[shape=beam, fill=couleurdefond, draw, inner sep=0.25em, thick, tikzit fill=white]
\tikzstyle{sbs}=[shape=beam, fill=couleurdefond, draw, inner sep=0.2em, thick, tikzit fill=white]
\tikzstyle{npbs}=[shape=beam, horizontal fill={{npbsmoitiebasse}{npbsmoitiehaute}}, draw, inner sep=0.25em, thick, tikzit fill={rgb,255: red,128; green,128; blue,128}]
\tikzstyle{npbsalenvers}=[shape=beam, horizontal fill={{npbsmoitiehaute}{npbsmoitiebasse}}, draw, inner sep=0.25em, thick, tikzit fill={rgb,255: red,128; green,128; blue,128}]
\tikzstyle{snpbs}=[shape=beam, horizontal fill={{npbsmoitiebasse}{npbsmoitiehaute}}, draw, inner sep=0.2em, thick, tikzit fill={rgb,255: red,128; green,128; blue,128}]
\tikzstyle{snpbsalenvers}=[shape=beam, horizontal fill={{npbsmoitiehaute}{npbsmoitiebasse}}, draw, inner sep=0.2em, thick, tikzit fill={rgb,255: red,128; green,128; blue,128}]
\tikzstyle{cnot}=[shape=circle, draw, path picture={ 
\draw[black](path picture bounding box.north) -- (path picture bounding box.south) (path picture bounding box.west) -- (path picture bounding box.east);
}]
\tikzstyle{thickcnot}=[shape=circle, draw, thick, path picture={ 
\draw[thick,black](path picture bounding box.north) -- (path picture bounding box.south) (path picture bounding box.west) -- (path picture bounding box.east);
}]
\tikzstyle{boite22}=[fill=white, draw=black, shape=rectangle, minimum height=1cm, minimum width=0.5cm]
\tikzstyle{boite15}=[fill=white, draw=black, shape=rectangle, minimum height=0.7cm, minimum width=0.5cm]
\tikzstyle{boite2}=[fill=white, draw=black, shape=rectangle, minimum height=0cm, minimum width=0cm]
\tikzstyle{snegpotentiel}=[fill=couleurdefond, draw=black, shape=rounded rectangle, inner sep=0.25em, tikzit fill={rgb,255: red,191; green,191; blue,191}, execute at end node={\footnotesize$\star$}]
\tikzstyle{negpotentiel}=[fill=couleurdefond, draw=black, shape=rounded rectangle, tikzit fill={rgb,255: red,191; green,191; blue,191}, execute at end node={$\star$}]
\tikzstyle{token}=[fill=black, draw=black, shape=circle, inner sep=0.1em]
\tikzstyle{whitetoken}=[fill=white, draw=black, shape=circle, inner sep=0.1em]
\tikzstyle{boitePBS}=[fill=white, draw=gray, thick, shape=rectangle, rounded corners=3pt, minimum height=0.6cm, inner sep=0.1em, minimum width=0.5cm]
\tikzstyle{boitePBS2}=[fill=white, draw=gray, thick, shape=rectangle, rounded corners=3pt, minimum height=0.55cm, inner sep=0.1em, minimum width=0.5cm]
\tikzstyle{sgene}=[fill={gray!30}, draw=black, shape=rounded rectangle, rounded rectangle east arc=0pt, minimum height=0.5cm, inner sep=0em, minimum width=0cm, scale=0.8]
\tikzstyle{sdetector}=[fill={gray!30}, draw=black, shape=rounded rectangle, rounded rectangle west arc=0pt, minimum height=0.5cm, inner sep=0em, minimum width=0cm, scale=0.8]
\tikzstyle{xsgene}=[fill={gray!30}, draw=black, shape=rounded rectangle, rounded rectangle east arc=0pt, minimum height=0.5cm, inner sep=0em, minimum width=0cm, scale=0.67]
\tikzstyle{xsdetector}=[fill={gray!30}, draw=black, shape=rounded rectangle, rounded rectangle west arc=0pt, minimum height=0.5cm, inner sep=0em, minimum width=0cm, scale=0.67]
\tikzstyle{PolRot}=[fill=white, draw=black, shape=rectangle, minimum height=0.5cm, inner sep=0.1em, minimum width=0.1cm]
\tikzstyle{PolRotrouge}=[fill={red!3}, draw=red, shape=rectangle, minimum height=0.5cm, inner sep=0.1em, minimum width=0.1cm, tikzit fill=red, execute at begin node={\textcolor{red}\bgroup}, execute at end node={\egroup}]
\tikzstyle{PhS}=[fill=white, draw=black, shape=rectangle, minimum height=0.5cm, inner sep=0.1em, minimum width=0.1cm]
\tikzstyle{gene}=[fill=white, draw=black, shape=rounded rectangle, rounded rectangle east arc=0pt, minimum height=0.5cm, inner sep=0em, minimum width=0cm]
\tikzstyle{detector}=[fill=white, draw=black, shape=rounded rectangle, rounded rectangle west arc=0pt, minimum height=0.5cm, inner sep=0em, minimum width=0cm]
\tikzstyle{cartoucherouge}=[rounded rectangle, fill={red!55!white}, draw=black, tikzit fill=red]
\tikzstyle{cartouchebleu}=[rounded rectangle, fill={blue!33!white}, draw=black, tikzit fill=blue]
\tikzstyle{diamantrouge}=[diamond, fill={rgb,255: red,255; green,115; blue,115}, draw=black]
\tikzstyle{diamantbleu}=[diamond, fill={rgb,255: red,171; green,171; blue,255}, draw=black]

\tikzstyle{new}=[-]
\tikzstyle{tirets}=[-, draw=black, dashed]
\tikzstyle{noire}=[-, draw=black]
\tikzstyle{ep}=[-, draw=black]
\tikzstyle{longdashed}=[-, dash pattern=on 5pt off 5pt]
\tikzstyle{pointilles}=[-, draw=black, dotted]
\tikzstyle{grise}=[-, draw={rgb,255: red,191; green,191; blue,191}]
\tikzstyle{rouge}=[-, draw=red]
\tikzstyle{bleue}=[-, draw=bleu, tikzit draw=blue]
\tikzstyle{verte}=[-, draw={rgb,255: red,0; green,230; blue,0}]
\tikzstyle{borddiagrammevide}=[-, dash pattern=on 0.5em off 0.5em on 0.5em off 0.5em on 0.5em off 0em]
\tikzstyle{msborddiagrammevide}=[-, dash pattern=on 0.28em off 0.28em on 0.28em off 0.28em on 0.28em off 0em]
\tikzstyle{sborddiagrammevide}=[-, dash pattern=on 0.2em off 0.2em on 0.2em off 0.2em on 0.2em off 0em]
\tikzstyle{xsborddiagrammevide}=[-, dash pattern=on 0.1em off 0.1em on 0.15em off 0.1em on 0.1em off 0em]
\tikzstyle{mediumdash}=[-, dash pattern=on 2pt off 2pt]
\tikzstyle{rougefonce}=[-, draw={red!50!black}, tikzit draw={rgb,255: red,136; green,0; blue,0}]
\tikzstyle{background}=[-, fill={rgb,255: red,229; green,229; blue,229}, draw=none, tikzit draw={rgb,255: red,128; green,128; blue,128}]

\newcommand{\Aff}{%
  \mathsf{Aff}
}

\newcommand{\Lin}{%
  \mathsf{Lin}
}
\newcommand{\Lag}{%
  \mathsf{Lag}
}
\newcommand{\Co}{%
  \mathsf{Co}
}
\newcommand{\Isot}{%
  \mathsf{Isot}
}

\newcommand{\zx}{%
  \mathsf{GSA}
}
\newcommand{\gla}{%
  \mathsf{GLA}
}
\newcommand{\gaa}{%
  \mathsf{GAA}
}
\newcommand{\stab}{%
  \mathsf{Stab}
}
\newcommand{\GGA}{%
  \mathsf{GGA}
}
\newcommand{\GQGA}{%
  \GSA[\C]^+
}
\newcommand{\Gauss}{\mathsf{Gauss}}
\newcommand{\QGauss}{\mathsf{QGauss}}

\newcommand{\GaussRel}{\mathsf{GaussRel}}
\newcommand{\QGaussRel}{\ALR[\C]^+}

\newcommand{\interp}[1]{%
  \left\llbracket #1 \right\rrbracket
}
\newcommand{\trans}{%
  \mathsf{T}
}
\DeclareRobustCommand{\disc}{%
  {\scalebox{.5}{\begin{tikzpicture}
	\begin{pgfonlayer}{nodelayer}
		\node [style=none] (0) at (0.25, 0) {};
		\node [style=none] (1) at (-0.25, 0) {};
		\node [style=none] (5) at (0.525, 0.05) {};
		\node [style=none] (6) at (0.525, -0.075) {};
		\node [style=none] (7) at (0.4, 0.175) {};
		\node [style=none] (8) at (0.4, -0.175) {};
		\node [style=none] (9) at (0.25, 0.325) {};
		\node [style=none] (10) at (0.25, -0.325) {};
	\end{pgfonlayer}
	\begin{pgfonlayer}{edgelayer}
		\draw (1.center) to (0.center);
		\draw (7.center) to (8.center);
		\draw (6.center) to (5.center);
		\draw (9.center) to (10.center);
	\end{pgfonlayer}
\end{tikzpicture}
}}
}

\newcommand{\N}{%
  \mathbb{N}
}
\newcommand{\Z}{%
  \mathbb{Z}
}
\newcommand{\Zp}{%
  {\mathbb{F}_p}
}
\newcommand{\K}{%
  \mathbb{K}
}

\newcommand{\R}{%
  \mathbb{R}
}
\newcommand{\C}{%
  \mathbb{C}
}

\newcommand{\LOv}{%
  \mathsf{LOv}
}

\NewDocumentCommand{\Rel}{O{X}}{%
  \mathsf{Rel}_{#1}
}
\NewDocumentCommand{\Symp}{O{\K}}{%
  \mathsf{Symp}_{#1}
}
\NewDocumentCommand{\ASymp}{O{\K}}{%
  {\Aff}\Symp[#1]
}
\NewDocumentCommand{\AR}{O{\K}}{%
  {\Aff}\Rel[#1]
}
\NewDocumentCommand{\lR}{O{\K}}{%
  {\Lin}\Rel[#1]
}
\NewDocumentCommand{\LR}{O{\K}}{%
  {\Lag}\Rel[#1]
}
\NewDocumentCommand{\IR}{O{\K}}{%
  {\Isot}\Rel[#1]
}
\NewDocumentCommand{\CR}{O{\K}}{%
  {\Co}\IR[#1]
}
\NewDocumentCommand{\ALR}{O{\K}}{%
  {\Aff}\LR[#1]
}
\NewDocumentCommand{\AIR}{O{\K}}{%
  {\Aff}\IR[#1]
}
\NewDocumentCommand{\ACR}{O{\K}}{%
  {\Aff}\CR[#1]
}
\NewDocumentCommand{\ZX}{O{\K}}{%
  \zx_{#1}
}
\NewDocumentCommand{\GSA}{O{\K}}{%
  \zx_{#1}
}
\NewDocumentCommand{\ZXdisc}{O{\K}}{%
  \ZX[#1]^\disc
}
\NewDocumentCommand{\GLA}{O{\K}}{%
  \gla_{#1}  
}
\NewDocumentCommand{\GAA}{O{\K}}{%
  \gaa_{#1}
}
\NewDocumentCommand{\Stab}{O{p}}{%
  \stab_{#1}
}

\newcommand{\abbralagr}{%
  \mathsf{ALR}
}

\NewDocumentCommand{\Matrices}{ O{m} O{n} O{\K} }{%
  \operatorname{M}_{#1,#2}(#3)
}
\NewDocumentCommand{\Sym}{O{n} O{\K} }{%
  \operatorname{Sym}_{#1}(#2)
}

\NewDocumentCommand{\Unit}{O{n} O{\C} }{%
  \operatorname{U}_{#1}(#2)
}

\NewDocumentCommand{\Orth}{O{n} O{\R} }{%
  \operatorname{O}_{#1}(#2)
}

\NewDocumentCommand{\SpOrth}{O{n} O{\R} }{%
  \operatorname{SO}_{#1}(#2)
}

\NewDocumentCommand{\Sp}{O{n} O{\R} }{%
  \operatorname{Sp}_{#1}(#2)
}

\NewDocumentCommand{\ASp}{O{n} O{\R} }{%
  \operatorname{AffSp}_{#1}(#2)
}

\NewDocumentCommand{\Gl}{O{n} O{\R} }{%
  \operatorname{Gl}_{#1}(#2)
}

\newlength\oversetwidth
\newlength\underwidth

\let\oldoverline\overline
\renewcommand{\overline}[1]{\mkern 1.5mu\oldoverline{\mkern-1.5mu#1\mkern-1.5mu}\mkern 1.5mu}

\renewcommand{\phi}{%
  \varphi
}
\renewcommand{\epsilon}{%
  \varepsilon
}

\DeclareMathOperator{\diag}{diag}

\renewcommand{\leq}{%
  \leqslant
}
\renewcommand{\geq}{%
  \geqslant
}

\renewcommand{\overrightarrow}{%
  \vec
}

\newcommand{\vdotss}{ 	\tikz[baseline, every node/.style={inner sep=0}]{ 	\node at (0,0){.}; 	\node at (0,4pt){.}; 	\node at (0,8pt){.}; 	} 	}
\newcommand{\vdotsn}{{\ \tikz[baseline, every node/.style={inner sep=0}]{\node at (0,.23){$\vdots$}; \node at (.35,0){$\scriptstyle n$}}\!}}
\newcommand{\vdotsm}{{\ \tikz[baseline, every node/.style={inner sep=0}]{\node at (0,.23){$\vdots$}; \node at (.35,0){$\scriptstyle m$}}\!}}

\title{Complete equational theories for classical and quantum Gaussian relations}

\author{Robert I. Booth}{University of Edinburgh, United Kingdom \and University of Bristol, United Kingdom}{robert.booth@ed.ac.uk}{https://orcid.org/0000-0002-1146-3380}{}
\author{Titouan Carette}{LIX, CNRS, École polytechnique, Institut Polytechnique de Paris, 91120, Palaiseau, France}{carette(at)lix.polytechnique.fr}{https://orcid.org/0000-0002-1618-4081}{}
\author{Cole Comfort}{Universit\'e de Lorraine, CNRS, Inria, LORIA, F 54000 Nancy, France}{}{}{}

\authorrunning{R.\,I. Booth et al}

\Copyright{Robert I. Booth, Titouan Carette, and Cole Comfort} 

\ccsdesc[500]{Theory of computation~Models of computation}
\ccsdesc[300]{Theory of computation~Rewrite systems}

\keywords{Graphical algebra, symplectic geometry, string diagrams, category theory, Markov categories, probabilistic computation, categorical quantum mechanics, quantum information, quantum computing,  quantum mechanics, quantum optics, graph theory} 

\category{} 

\relatedversion{} 

\acknowledgements{CC is supported by the project NEASQC funded from the
European Union’s Horizon 2020 research and innovation programme (grant
agreement No 951821).}

\EventEditors{John Q. Open and Joan R. Access}
\EventNoEds{2}
\EventLongTitle{42nd Conference on Very Important Topics (CVIT 2016)}
\EventShortTitle{CVIT 2016}
\EventAcronym{CVIT}
\EventYear{2016}
\EventDate{December 24--27, 2016}
\EventLocation{Little Whinging, United Kingdom}
\EventLogo{}
\SeriesVolume{42}
\ArticleNo{23}

\begin{document}

\maketitle

\begin{abstract}
  We give generators and relations for the hypergraph props of Gaussian
  relations \cite{gaussrel} and positive affine Lagrangian relations. The
  former extends Gaussian probabilistic processes by completely-uninformative
  priors, and the latter extends Gaussian quantum mechanics with
  infinitely-squeezed states. These presentations are given by adding a
  generator to the presentation of real affine relations and of real affine Lagrangian
  relations which freely codiscards effects,  as well as certain rotations.

  The presentation of positive affine Lagrangian relations provides a rigorous justification for many common yet
  informal calculations in the quantum physics literature involving
  infinite-squeezing. Our presentation naturally extends Menicucci
  et al.'s graph-theoretic representation of Gaussian quantum
  states \cite{menicucci_graphical_2011} with a representation for
  Gaussian transformations. Using this graphical calculus, we also give a graphical proof of Braunstein and Kimble's continuous-variable quantum teleportation protocol \cite{braunstein_teleportation_1998}.  We also interpret the \textsf{LOv}-calculus
  \cite{clement_lov-calculus_2022}, a diagrammatic calculus for reasoning
  about passive linear-optical quantum circuits in our graphical calculus. Moreover, we show how our presentation allows
  for additional optical operations such as active squeezing.
\end{abstract}

%


\section{Introduction}
\label{sec:intro}
The research programme of categorical quantum mechanics seeks to axiomatize
quantum mechanics and quantum computing using the theory of monoidal
categories.  More specifically, the ZX-calculus community has sought out to
find diagrammatic presentations for different fragments of quantum theory
using string diagrams \cite{zx,zxcompletea,zxcompleteb,backens_zx-calculus_2014,qutrit,zw,zh,zxa}.
This graphical quantum-mechanical paradigm has found a 
wide-ranging applications in quantum computing \cite{Differentiation,East2022,lattice,comfort_algebra_2023,Huang2023,duncan_graph-theoretic_2020,Kissinger2020,Backens2021,Toumi2021}, as well as in
teaching quantum theory \cite{pqp, Heunen2019-ox}.

Up until relatively recently, however, the fragments of quantum mechanics of
study have been almost exclusively finite dimensional. This is partly due to
the inconvenient fact that the category of (infinite-dimensional) Hilbert spaces is
not compact closed, a highly desirable feature for rewriting. This is at odds
with much of quantum physics, where many systems have continuous
degrees of freedom which can only be modelled in infinite dimensions.
Famously, the canonical commutation relations only hold in infinite dimensions.

Different approaches have been taken to get around this obstacle and attempt to
diagrammatically axiomatize infinite-dimensional quantum systems. 

One general
approach to tackling this obstacle has been to change the semantics. For
example, Gogioso and Genovese suggest embedding Hilbert spaces into a larger
compact closed category using nonstandard analysis \cite{Gogioso2017}.  On the
other hand Cockett et al. have sought to completely  abandon the Hilbert space semantics,
shifting to non-degenerate models of multiplicative linear logic, and exploiting
more refined notion of duals in the setting of *-autonomous categories
\cite{Cockett2021}. However, neither of these lines of inquiry have provided
presentations in the same spirit of the ZX-calculus; instead, they have led
to categorical reconstructions, and exotic toy models of quantum mechanics.

Another approach has been to just put up with the lack of
duals \cite{abramsky_h*-algebras_2010}. This line of thinking has
produced graphical presentations for different fragments of quantum
mechanics, including for coherent control \cite{Coherentcontrol}
and passive linear quantum optics \cite{clement_lov-calculus_2022,
heurtel_complete_2024}. However, the lack of duals makes it harder to
incorporate existing quantum circuit optimization techniques which exploit
duals and hypergraph structure. For example the quantum circuit optimisation
package \texttt{PyZx} represents quantum circuits as undirected open graphs
\cite{Kissinger2020}, and the cutting-edge ZX-calculus quantum circuit
optimization techniques often make use of these graph-like representations
\cite{duncan_graph-theoretic_2020,Kissinger2020,Backens2021}.

However, these two approaches are not mutually exclusive. For example De Felice et al.  \cite{deFelice2023} mix both approaches: working in one semantics without duals, and taking finite truncations into another semantics in a way similar to Gogioso and Genovese, albeit without nonstandard analysis.

In this paper, we restrict our attention in two main ways. Firstly,
we seek an axiomatization only of the \emph{Gaussian} fragment of
infinite-dimensional quantum mechanics. Gaussian quantum mechanics
(GQM) is an important fragment for many candidate physical systems
for implementing quantum computation \cite{lloyd_quantum_1999,
braunstein_quantum_2005, ortiz-gutierrez_continuous_2017}, including
quantum optics \cite{fabre_modes_2020, yokoyama_ultra-large-scale_2013,
yoshikawa_generation_2016, asavanant_time-domain_2019,
bartolucci_fusion-based_2023, de_gliniasty_spin-optical_2024}. Formally, it
bears a strong resemblance to the stabiliser fragment of finite-dimensional
quantum theory, although many of the primitives of stabiliser theory do not
generalize well to infinite-dimensional Hilbert spaces.

Secondly, instead of looking for a Hilbert space semantics, we give
our semantics in the category \(\ALR[\C]\) of \emph{complex affine Lagrangian
relations}. This allows us to express all of Gaussian quantum mechanics,
as well the nonphysical operations known as ``infinite squeezing,'' which
includes hypergraph structure, compact structure as well as a notion of
Dirac deltas (eg. \cite[page~18]{Von_Neumann2018-nl}).  Our relational interpretation of Dirac deltas makes them behave algebraically how physicists (often informally) manipulate them; however we do not give a distributional semantics in the sense of Schwartz \cite{schwartz1957theorie}.
 To be more specific about our semantics, we construct a subcategory of \(\ALR[\C]\)
of relations verifying a positivity condition which exactly captures
these aforementioned primitives. This change of semantic category is partly motivated
by the fact that the stabiliser fragments of odd prime dimensional qudit quantum systems are
already known to be projectively equivalent to categories of affine Lagrangian relations
over finite fields \cite{alr}. Therefore the equational theory of affine Lagrangian relations of Booth et al. \cite{gsa} restricts to the stabiliser fragment of the ZX-calculus  \cite{booth_complete_2022,
poor_qupit_2023}. Furthermore, Lagrangian relations have been extensively
studied in the context of geometric quantization: \(\ALR[\C]\) is Weinstein's
affine symplectic category \cite{weinstein_symplectic_2009}, we discuss this connection further in subsection~\ref{subsec:gsaapplications}.

Notably, although they do not give semantics in terms of Lagrangian relations, Menicucci et al. use symplectic-geometric methods in order to represent Gaussian states using graphs \cite{menicucci_graphical_2011}.  Although they do not give generators and relations for a monoidal category, they do provide graph transformation rules.  Therefore, our work can be seen as a relational/compositional extension of theirs.

%

%

Although our original motivation was find a relational semantics, as well as a graphical presentation for Gaussian quantum mechanics; surprisingly, by imposing a single additional constraint on our subspaces, we also obtain a complete presentation for the category of Gaussian relations, which was originally formulated by Stein and Samuelson \cite{gaussrel}.  Our semantic formulation of this category makes it into a concrete category of relations: where composition is given merely by relational composition.  This presentation extends that of affine relations by Bonchi et al. \cite{gaa} (at least over \(\R\)) 
 which itself extends the presentation of linear relations by Bonchi et al. \cite{gla}.  Similar, yet disjoint presentations exist for piecewise linear relations \cite{dpla}, and polyhedral relations \cite{dpa}.  Such presentations for categories of relations with algebraic structure have been used as a syntax for signal flow diagrams in control theory \cite{Bonchi2019,control,Bonchi2017,Bonchi2014, coya,erbele}, systems with bounded resources \cite{dpa}, and electrical circuits \cite{gaa,Boisseau2022,network,passive,dpla,coya}.
 Therefore, it is likely that our presentation for positive affine Lagrangian relations will be useful outside of quantum mechanics.

{\bf Related work.}

Independently of this work, Stein et al. \cite{stein_graphical_2024} give
the same presentation of Gaussian relations. However, their semantics are
completely different: their semantics are non-linear varieties whereas we
regard Gaussian relations as a subcategory of affine relations. The precise
connection between these interpretations still needs to be elucidated. Their
characterisation does not extend to the quantum case.

{\bf Structure of paper}

\noindent
In section~\ref{sec:strings}, we review string diagrams for symmetric monoidal categories, and related notions. In section~\ref{sec:gsa}, we review graphical symplectic algebra: commencing with a brief review of symplectic linear algebra in subsection~\ref{ssec:symplectic}, and the graphical language for affine Lagrangian relations in subsection~\ref{ssec:gsa}. In subsection~\ref{subsec:gsaapplications}, we begin to discuss the quantum nature of affine Lagrangian relations, which motivates our more refined semantics.

Section~\ref{sec:gauss} is devoted to giving a complete presentation of Gaussian relations.  We first review Gaussian probability theory in subsection~\ref{ssec:gaussprob}.  Then in subsection~\ref{ssec:gaussrel}, we give a novel formulation of Gaussian relations in terms of a subcategory of complex affine Lagrangian relations, followed by a graphical presentation in subsection~\ref{ssec:gga}.

The final section~\ref{sec:sguass} is devoted to giving a complete presentation for positive affine Lagrangian relations.  We begin with subsection~\ref{ssec:gaussqm} with a review of Gaussian quantum mechanics. In subsection~\ref{ssec:positive}, we characterize positive Lagrangian relations more concretely.  In subsection~\ref{ssec:QGaussRel}, we give a presentation for this category.  In subsection~\ref{ssec:QGaussRel}, we connect our work to Menicucci et al.'s graphical representation of Gaussian states 
\cite{menicucci_graphical_2011}.
In subsection~\ref{ssec:teleportation}, we give a graphical proof of Braunstein and Kimble's continuous variable quantum teleportation protocol \cite{braunstein_teleportation_1998}.
  Similarly, in subsection~\ref{ssec:lov} we show how the graphical language for passive linear quantum optics
known as the \(\LOv\) calculus \cite{clement_lov-calculus_2022} can be interpreted in our language.

\section{Graphical languages}
\label{sec:strings}
In this section we review the necessary background on string diagrams and graphical languages needed to understand this article.  A more in depth introduction to this topic can be found in the article of Booth et al. \cite[section~2]{gsa}.

\subsection{String diagrams}
%
A {\bf prop} is a strict symmetric monoidal category such that the monoid of objects under the tensor product is the natural numbers under addition.  For this reason, by convention we denote the tensor product on objects by \(+\), and on maps by \(\oplus\).
We draw a map of type \(n\to m\) as a box with \(n\) incoming wires and \(m\) outgoing wires.
The composition of \(g\circ f\) of arrows \(f\) and \(g\) is drawn by ``plugging wires'' and the monoidal product \(f\oplus g\) as ``vertical juxtaposition''; we can also exchange wires with the symmetry map \(\sigma_{1,1}: 2 \to 2\) drawn as two wires crossing (satisfying the topologically evident equations):
\begin{equation}
  \input{figures/strings/composition.tikz}
  \quad
  \input{figures/strings/monoidal_product.tikz}
  \quad
  \input{figures/strings/symmetry_equations.tikz}
\end{equation}

Arrows of type \(0\to n\) are {\bf states}; \(n\to 0\) are {\bf effects}; \(0\to 0\) are {\bf scalars}.



\subsection{Duality and flexsymmetry}
\begin{definition}
  A {\bf compact prop} is a prop in which a single wire can be bent into cups \(\eta_1:0\to 2\) and caps \(\epsilon_1:2\to 0\), satisfying the topologically evident equations:
  \begin{equation*}
    \eta_1\eqqcolon\input{figures/strings/cup.tikz}
    \qand
    \epsilon_1\eqqcolon\input{figures/strings/cap.tikz}
    \quad
    \textit{such that}
    \quad
    \input{figures/strings/string_pulling.tikz}
  \end{equation*}
\end{definition}
One advantage of working with compact props is that every map \(f:n\to m\), can be uniquely represented by a state called its \emph{\textbf{name}} given by bending the inputs into outputs \(\lfloor f \rfloor:0\to n+m\). Sometimes the order of wires in the name of a map doesn't matter:
\begin{definition}
 \label{def:flexsym}
  An arrow \(f : n \to m\) in a compact prop is \textbf{flexsymmetric} if for any permutation \(\varsigma\) of \(m+n\), \(\varsigma \circ \lfloor f \rfloor=\lfloor f \rfloor\).
\end{definition}
In some cases, compact props admit presentations where all generators are flexsymmetric. This means that the string diagrams consisting of such generators can be treated as undirected graphs.  We use this extensively throughout this paper. 

%
%
%

\subsection{Scalable notation}

In order to reason about props \(\mathsf{C}\), it will often be useful to
work in their strictification \([\mathsf{C}]\).  When representing maps in
\([\mathsf{C}]\), each wire is labelled by a natural number \(n\), indicating
that the wire represents a bundle of \(n\) atomic wires. This prop is given by
the generators and equations of \(\mathsf{C}\) in addition to new generators
which divide and gather the wires:

  \begin{equation}
    \input{figures/strings/tensor_dividers.tikz}
  \end{equation}

In other words, in this setting, we are able to reason about composite systems in \(\mathsf{C}\). As a matter of convention, we will typeset generic wires as thicker, to indicate when there may be some \(n\in \N \) atomic wires bundled together.  If we know that there is only a single atomic wire within such a bundle, we will draw it thinly; moreover, if we know that the bundle is empty, we will omit the wire entirely.
This will allow us to make inductive arguments using only pictures; therefore, given any such \(\mathsf{C}\), getting a hold of \([\mathsf{C}]\) will be very useful.  We will call the ways of thinking about \(\mathsf{C}\) using \([\mathsf{C}]\) \textit{\textbf{scalable notations}}.

\section{Graphical symplectic algebra}
\label{sec:gsa}

In this section, we review graphical symplectic algebra of Booth et al. \cite[section~3]{gsa}:

\subsection{Symplectic linear algebra}
\label{ssec:symplectic}
First, we review symplectic linear algebra. \(\K\) will denote an arbitrary
field, although in the remainder of the article and unless otherwise stated,
we will take \(\K\) to be either the field of real, or complex numbers. The
group of units of \(\K\) is denoted by \(\K^*\).

\begin{definition}
  A \textbf{symplectic \(\K\)-vector space} is a \(\K\)-vector space \(V\) endowed
  with a bilinear map \(\omega_V: V \oplus V \to \K\) called the  \textbf{symplectic form} which is:
  \begin{itemize}
    \item \textbf{Alternating:} For all \(\vec v \in V\), \(\omega_V(\vec v,\vec v)=0\).
    \item \textbf{Nondegenerate:} Given some \(\vec v \in V\): if for all \(\vec u\in V\) \(\omega_V(\vec v,\vec u)=0\), then \(\vec v=\vec 0\).
  \end{itemize}
  
  A \textbf{symplectic \(\K\)-linear map}
  between symplectic spaces \((V,\omega_V) \to (W, \omega_W)\) is a \(\K\)-linear
  map \(S : V \to W\) that preserves the symplectic form, i.e. \(\omega_W(S\vec v,S\vec u)
  = \omega_V(\vec v,\vec u)\) for all \(\vec v, \vec u\in V\). A
  \textbf{symplectomorphism} is a symplectic isomorphism.
  Similarly, a \textbf{symplectic \(\K\)-affine map} \(S,\vec a\) is a \(\K\)-affine map whose linear component \(S\) is symplectic; and an \textbf{affine symplectomorphism} is an  symplectic \(\K\)-affine isomorphism.
\end{definition}

\begin{example}
The vector space \(\K^{2n}\) has a symplectic form
\(\displaystyle\omega_n\!\left(  \begin{bmatrix}\vec z \\ \vec x\end{bmatrix}, \begin{bmatrix}\vec{z'} \\ \vec{x'}\end{bmatrix} \right)\coloneqq\sum_{j=0}^{n-1} z_jx_j'-x_jz_j' \)
\end{example}

As a matter of notation, denote the skew-symmetric \(2n\times 2n\) matrix which represents this form by \(\Omega_n\).
Note that  \(\omega_{1}\) measures the (oriented) area of parallelograms in \(\K^2\).
\begin{theorem}[Linear Darboux theorem]
  \label{thm:darboux}
  Every finite-dimensional symplectic vector space \((V,\omega_V)\) is symplectomorphic to the symplectic vector space \((\K^{2n} , \omega_n)\).
\end{theorem}

A linear subspace \(S\) of a symplectic vector space \((V,\omega)\)  
is \emph{\textbf{Lagrangian}}, if \(2\dim (S)=\dim(V)\) and for all \(\vec v,\vec w \in S\), \(\omega(\vec v,\vec w)=0\).
In particular, every Lagrangian subspace \(S \subseteq (\K^{2n}, \omega_n)\) is the kernel of a \(n\times 2n\) matrix \(\ker [Z|X] = S\), called a \emph{\textbf{generator matrix}}, so that for all rows \(\vec v, \vec w\), \(\omega(\vec v,\vec w)=0\).
%
%
%
%
%
%
Recall the following generalisation of a linear subspace:
\begin{definition}
Given \(n\in\N\), an \textbf{affine subspace} of \(\K^n\) is a subset
\(S+\vec a \coloneqq \{ \vec v \in \K^n \ | \ S\vec v= \vec a  \} \subseteq \K^n\)
for some vector \(\vec a \in \K^m \) and matrix \(S:n\to m\). 
An affine subspace \(S+\vec a\) of  a symplectic vector space \((V,\omega)\) is Lagrangian when either \(S\) is Lagrangian or empty.
An \textbf{affine Lagrangian relation} \(m\to n\)
is an affine Lagrangian subspace of
\(\K^{2m} \oplus \K^{2n}\) with respect to the  symplectic form given by:
\[
  \omega_{m,n}
  \left(
    \begin{bmatrix}
      \vec z_0 \\ \vec x_0 \\ \vec z_1 \\ \vec x_1
    \end{bmatrix}
    ,
    \begin{bmatrix}
      \vec z_2 \\ \vec x_2 \\ \vec z_3 \\ \vec x_3
    \end{bmatrix}
  \right)
  \coloneqq
    \omega_n
  \left(
    \begin{bmatrix}
      \vec z_1 \\ \vec x_1
    \end{bmatrix}
    ,
    \begin{bmatrix}
      \vec z_3 \\ \vec x_3
    \end{bmatrix}
  \right)
  -
  \omega_m
  \left(
    \begin{bmatrix}
      \vec z_0 \\ \vec x_0
    \end{bmatrix}
    ,
    \begin{bmatrix}
      \vec z_2 \\ \vec x_2
    \end{bmatrix}
  \right)
\]
\end{definition}

We must introduce this twisted symplectic form \(\omega_{m,n}\), as opposed to \(\omega_{m+n}\), so that affine Lagrangian relations can be composed relationally:
\begin{definition}
  The compact prop of affine Lagrangian relations \(\ALR\)  has:
  \begin{itemize}
    \item \textbf{Arrows} affine Lagrangian relations.
    \item \textbf{Composition} of \(R:n\to m\) and \(S:m\to k\) is given by:
    \[ S\circ R \coloneqq  \{ (\vec a,\vec c) \in \K^{2n}\oplus \K^{2k} \ | \ \exists \vec b \in \K^{2m}: (\vec a,\vec b) \in R, (\vec b,\vec c) \in S  \}:n\to k \]
    \item \textbf{Identity} on \(n\), given by the diagonal relation \(\{(\vec v,\vec v) \in  \K^{2n}\oplus \K^{2n} \}:n\to n \)
    \item \textbf{Monoidal structure}, given by the object \(0\) and the direct sum:
      \[ S\oplus T \coloneqq \left\{\left(\begin{bmatrix}\vec {z_S} \\ \vec{z_T} \\ \vec{x_S} \\ \vec{x_T} \end{bmatrix},\begin{bmatrix} \vec {z_S'} \\ \vec{z_T'} \\ \vec{x_S'} \\ \vec{x_T'} \end{bmatrix}\right)  \,\middle|\,   \left(\begin{bmatrix} \vec {z_S}\\ \vec {x_S} \end{bmatrix} , \begin{bmatrix} \vec {z_S'}\\ \vec {x_S'} \end{bmatrix} \right)\in S,  \left(\begin{bmatrix} \vec {z_T}\\ \vec {x_T} \end{bmatrix} , \begin{bmatrix} \vec {z_T'}\\ \vec {x_T'} \end{bmatrix} \right) \in T \right\} \]
    \item \textbf{Symmetric structure} is given by the relation: 
      \[\sigma_{1,1}\coloneqq \left\{\left(\begin{bmatrix}z_0 \\ z_1 \\ x_0 \\ x_1 \end{bmatrix},\begin{bmatrix} z_1 \\ z_0 \\ x_1 \\ x_0 \end{bmatrix}\right)  \,\middle|\,  z_0,z_1,x_0,x_1\in\K \right\}:2\to 2 \]
    \item \textbf{The compact structure} is given by:\\
    \hspace*{-.5cm} \mbox{\(
    \eta_1 \coloneqq \left\{\left(\bullet, \begin{bmatrix} z \\ z \\ x \\ -x \end{bmatrix}\right) \,\middle|\,  z,x\in \K\right\}: 0\to 2 \quad\quad
    \epsilon_1 \coloneqq \left\{\left(\begin{bmatrix} z \\ z \\ x \\ -x \end{bmatrix}, \bullet\right) \,\middle|\,  z,x\in \K\right\}: 2\to 0
  \)}
  \end{itemize}
\end{definition}





Note that every affine symplectic map \(f\) 
induces an affine Lagrangian relation via its graph \(\Gamma(f)\).
This is a faithful functor so \(\ALR\) strictly generalises (affine) symplectic maps.
\subsection{Graphical presentation}
\label{ssec:gsa}
In this section, we review the presentation of affine Lagrangian relations of Booth et al \cite{gsa}:
\begin{theorem}
  \label{thm:GSA}
  Given a field \(\K\), the compact prop \(\ALR\)  has a flexsymmetric presentation \(\ZX\) (read \textbf{graphical symplectic algebra}), generated by, for all \(a,b\in \K\) and \(n,m\in \N\):
  \begin{align*}
  &\interp{\input{figures/AffLagRel/b-spider.tikz}}_\abbralagr^\zx
    \coloneqq \left\{ \left(
    \begin{bmatrix}
      \vec z \\ x \vspace*{-.2cm}\\ \vdotsm \\ x
    \end{bmatrix} ,
    \begin{bmatrix}
      \vec {z'} \\ x \vspace*{-.2cm} \\ \vdotsn \\ x
    \end{bmatrix}
    \right) \,\middle|\, \begin{aligned} &\vec z\in\K^m, \vec{z'} ; \in \K^n, x \in \K \qq{such that} \\ &\sum_{j=0}^{m-1} z_j - \sum_{k=0}^{n-1} z_k' + bx = a \end{aligned} \right\} \\
  &\interp{\input{figures/AffLagRel/w-spider.tikz}}_\abbralagr^\zx
    \coloneqq
    \left\{ \left(
      \begin{bmatrix}
      z \vspace*{-.2cm} \\ \vdotsm \\ z \\  \vec x
      \end{bmatrix} ,
      \begin{bmatrix}
      -z \vspace*{-.2cm} \\ \vdotsn \\ -z \\  \vec {x'}
      \end{bmatrix}
    \right) \,\middle|\, \begin{aligned} & \vec x \in \K^m, \vec{x'}\in \K^n, z\in \K \qq{such that} \\ &\sum_{j=0}^{m-1} x_j + \sum_{k=0}^{n-1} x_k' - bz = a \end{aligned} \right\}
    \end{align*}
  Modulo flexsymmetry of the generators and for all \(a,b,c,d\in\K\) and \(z\in\K^*\), the equations:
  
  \noindent
  \[
    \input{figures/AffLagRel/axioms_zx.tikz}
  \]
  
\end{theorem}

Call the white and grey nodes \textit{\textbf{spiders}}. The Abelian group \(\K^2\) from which they take their labels is called their \textit{\textbf{phase-group}}. 
  Call the first component of \(\K^2\cong \K\oplus\K\) the \textit{\textbf{affine phase}}, and the second component the
  \textit{\textbf{symplectic phase}}.  Spiders which are unlabelled are called \textit{\textbf{phase-free}}.
  The distinguished boxes with labels \(1\) and \(-1\) are respectively the \textit{\textbf{symplectic Fourier transform}} and \textit{\textbf{inverse symplectic Fourier transform}}. The derived arrow-shaped generators are called \textit{\textbf{squeezing maps}}.

A prop equipped with the structure satisfying the axioms encapsulated by the
phase-free grey spiders here is called a \textit{\textbf{hypergraph-prop}}
\cite{kissinger_finite_2015}.  Conceptually, the grey spider should be regarded
as way to formally wire things together.

It is useful to have some geometric intuition for the generators of \(\GSA\). The affine phases of \(1\to 1\) grey and white spiders shift the affine plane in orthogonal directions; whereas, the symplectic phases shear the affine plane.  The symplectic Fourier transform rotates the affine plane by \(\pi/2\).  A squeezing map by \(a\) multiplies one axis in the plane by \(a\) and divides the other axis by \(a\). Spiders with more legs are more complicated to visualise..
\begin{remark}\label{rem:affrel}
  We can restrict the semantics to the prop \(\AR\) of affine relations whose morphisms are merely affine subspaces, by restricting the syntax to the following generators, for all \(a\in\K\):\quad\quad
  \(
    \input{figures/AffLagRel/gaa_generators.tikz}
  \)
  
There is a presentation, which we call \(\GAA\), whose axioms we have included in appendix~\ref{app:gaa}.  Here, the squeezing map corresponds to ordinary scalar multiplication.
\end{remark}

Let \(\Sym\) denote the set of symmetric \(n \times n\) matrices over \(\K\) and \(\Matrices\) the set  of \(m \times n\) matrices over \(\K\).  We define scalable versions of the generators of \(\ZX\), where now multiple wires are bundled together:

\begin{definition}
  \label{def:scalable_spider}
  Phase free scalable spiders are defined by induction on the number of wires.  The base case is trivial, and the inductive step is given by:
  \begin{equation}
    \label{eq:GAAthickspiders}
    \input{figures/AffLagRel/scaledppspidef.tikz}
  \end{equation}
  This allow us to define block matrices by induction as well, with the inductive step:
  \begin{equation}
  \label{eq:arrow_decomposition}
  \input{figures/AffLagRel/blockmatex.tikz}
  \end{equation}
  When restricted to \(\GAA \cong \AR\) these block matrices correspond to honest matrices, regarded as affine relations.
  Back in \(\GSA\), we also have scalable versions of the Fourier transform, as well and the relational converse of matrices:
  \begin{equation}
    \label{eq:scalable_box}
    \input{figures/AffLagRel/scalable/box-def.tikz}
    \quad\quad
    \input{figures/AffLagRel/scalable/matrix_converse.tikz}
  \end{equation}
  Now we define \(\K^{k}\times \Sym[k]\)-phased scalable spiders.
  For the inductive step, take \(n,m\in \N\), \(a,b\in \K\), \(\vec v, \vec w \in \K^k\) and \(A\in\Sym[k]\), we define  \([k+1]\)-coloured phased spiders as:
  \begin{equation}
    \label{eq:thick_spider_def}
    \input{figures/AffLagRel/scalable/b-spider-def.tikz}
    \quad\quad
    \input{figures/AffLagRel/scalable/w-spider-def.tikz}
  \end{equation}
\end{definition}

A \([k]\)-coloured \(n\to m\) spider parametrizes an undirected open graph with edges coloured by \(\K\), vertices coloured by \(\K^2\), with \(n\) / \(m\) distinguished inputs / outputs:

\begin{example}
  \label{ex:scalable_graph_example}
  On the left hand side we have a ``graph state'' with \(n = 0\), \(m=1\) and \(k=3\).  On the right hand side, we have a spider with  \(m = 3\),  \(n = 2\) and \(k =
  3\):
  \begin{equation}
    \input{figures/AffLagRel/scalable/b-spider-explicit-simple.tikz}
    \quad\quad\quad
    \input{figures/AffLagRel/scalable/b-spider-explicit.tikz}
  \end{equation}
\end{example}

Theorem~\ref{thm:GSA} follows from the fact that nonempty states \(0\to n\) in \(\ZX\) can be uniquely represented by a more general notion of graph where not all vertices are outputs:
\begin{proposition}
  \label{rem:scalable_reduced_AP-form}
  Every nonempty state \(0\to n\) in \(\ZX\) is reducible to a unique  \textbf{reduced-AP form}, given by a 5-tuple \((L,\Sigma, \vec x, \vec \mu, \varsigma)\),  where \(m \leqslant n \in \N\), \(\vec{x} \in \K^m\), \(\vec{\mu} \in \K^{n-m}\), \(L \in \Matrices[m][n-m]\) and \(\Sigma \in \Sym[n-m]\), where \(\varsigma\in \Matrices[n][n]\) is a permutation:
  \begin{equation*}
    \label{eq:scalable-reduced-AP-form}
    \input{figures/AffLagRel/scalable/reduced-AP-form.tikz}
  \end{equation*}
\end{proposition}


%

It is also useful to think of this normal form concretely in terms of generator matrices:
\begin{remark}
  The normal form of proposition~\ref{rem:scalable_reduced_AP-form} corresponds to reducing a nonempty affine Lagrangian relation to a unique generator matrix:
  \[
    \ker\left(
      \begin{bmatrix}[cc|cc]
        I_m & L  & \Sigma     &  0\\
        0     & 0  & -L^\trans 	& I_{n-m}
      \end{bmatrix}
      \begin{bmatrix}
        \varsigma & 0\\
        0               & \varsigma
      \end{bmatrix}
    \right)
    +
    \begin{bmatrix}
      \vec x \\ \vec \mu
    \end{bmatrix}
    \begin{bmatrix}
        \varsigma & 0\\
        0               & \varsigma
     \end{bmatrix}
  \]
  
  The affine part of the subspace is irrelevant for  much of the calculations; therefore, we will often just refer to the pair  \(([I,L]\varsigma,\Sigma)\) as the \emph{\textbf{AP-invariant}} of the Lagrangian subspace.
\end{remark}

%
%
%
%
%
%
%
\subsection{Weil representation into Hilbert spaces}
\label{subsec:gsaapplications}

In this subsection, we begin to discuss how \(\ALR\) gives a semantics for fragments of quantum mechanics.  First, we recall the state space for finite dimensional quantum systems:
\begin{definition}
  Given a finite set \(X\), there is an associated \(|X|\)-dimensional Hilbert space of \textbf{square summable functions} on \(X\):
  \begin{equation*}
    \ell^2(X)\coloneqq \left\{ f:X\to \C \ \middle| \ \sum_{x \in X} | f(x) |^2 <\infty \right\}
    \quad\text{with inner product}\quad
    \langle f, g \rangle \coloneqq \sum_{x\in X} \overline{f}(x)g(x) 
  \end{equation*}
  The Dirac deltas for elements of \(X\) determine a distinguished orthonormal basis \(\{|x\rangle \}_{x\in X}\) on on \(\ell^2(X)\).
  Given some \(1 \leq d\in \N\), we identify the Hilbert space \(\ell^2(\Z/d\Z)\) with the state space of a \textbf{qudit}, and \(\ell^2((\Z/d\Z)^n)\cong (\ell^2(\Z/d\Z))^{\otimes n}\) with the state space of \(n\)-qudits.  
\end{definition}

Informally, qudits are d-dimensional quantum analogues of classical bits: the basis elements (and their normalized linear combinations) are the possible values of the qudit.

Building on the work of Gross \cite{Gross2006}, Comfort and Kissinger showed that for odd prime \(p\), there is a projective equivalence between \(\ALR[\Zp]\) and ``{\em odd-prime qudit stabiliser quantum circuits}'' (where there is an isomorphism of crings \(\Zp\cong \Z/p\Z \)) \cite{alr}. 
Symplectic vector spaces \((\mathbb{F}_p^{2n},\omega_n)\)  are identified with Hilbert spaces \(\ell^2(\mathbb{F}_p^n)\);  and nonempty affine Lagrangian subspaces \((S+\vec a)\subseteq (\mathbb{F}_p^{2n}, \omega_{n})\) with \textit{\textbf{stabiliser states}} \(\ell^2(\mathbb{F}_p^0)\to\ell^2(\mathbb{F}_p^n)\) given by the projector:

\begin{equation}\label{eq:discreteweil}
  \dfrac{1}{p^n}\sum_{{[\vec z}^\trans\ {\vec x}^\trans]^\trans \in S} \bigotimes_{j=0}^{p^n-1} \exp(2\pi i a_i/p) \mathcal{Z}^{z_j}\mathcal{X}^{x_j}
  \quad
  \text{where}
  \quad
  \mathcal{Z}|j\rangle \coloneqq \exp(2\pi i j/p ) |j\rangle
  \quad
  \mathcal{X}|j\rangle \coloneqq |j+1\rangle
\end{equation}

Because finite dimensional Hilbert spaces, and in particular, qudit stabiliser circuits are compact, this determines the representation completely. There is a natural generalization of qudits to infinite dimensions, used for example, in quantum optics:
\begin{definition}
  The Hilbert space of \textbf{square  integrable functions} on \(\R^n\) is:
  \begin{equation*}
    L^2(\R^n)\coloneqq \left\{ f:\R^n\to \C  \ \middle| \ \int_{\R^n} | f(\vec v) |^2\, d\vec v <\infty \right\}
    \quad\text{with}\quad
    \langle f, g \rangle \coloneqq \int_{\R^n} \overline{f}(\vec v)g(\vec v) \, d\vec v
  \end{equation*}

  We identify the Hilbert space \(L^2(\R)\) with the state space of a \textbf{qumode}, and \(L^2(\R^n)\cong (L^2(\R))^{\otimes n}\) with the state space of \(n\)-qumodes. 
\end{definition}
%
%

One might hope to obtain a similar projective representation of \(\ALR[\R]\) into infinite dimensional Hilbert spaces.
Indeed, affine symplectomorphisms on \((\R^{2n},\omega_n)\)
are known to give a projective representation of the group of ``Gaussian unitaries'' on the  Hilbert
space \(L^2(\R^n)\) (see subsection~\ref{ssec:gaussqm}, definition~\ref{def:gaussunitary}).
However, attempting to lift the representation of states taken in the discrete setting to qumodes reveals that the analogous projector to that in equation~\eqref{eq:discreteweil} is not bounded, and therefore not a state in the category of Hilbert spaces.
Even worse, the category infinite dimensional Hilbert spaces is not even compact, so it does not suffice to only concern ourselves with states.


Nevertheless, this technical difficulty has not stopped physicists from widely
employing such ``unphysical'' states in their calculations, where they are
known as \emph{infinitely-squeezed} states \cite{menicucci_universal_2006,
gu_quantum_2009}. Because \(\ALR[\R]\) has a simple presentation in terms
of the hypergraph prop \(\GSA[\R]\), this provides  both a rigorous syntax
and semantics for infinitely squeezed states and Gaussian operators: so that
\(\GSA[\R]\) can be regarded as the ``ZX-calculus for infinitely squeezed
quantum circuits with Gaussian operators.''  For example, we can think of the
diagram \(\input{figures/AffLagRel/dirac-delta.tikz}\) as a formal Dirac delta
centered at \(r\in \R\).

The hypergraph structure of \(\ALR[\R]\) has already been exploited to perform calculations on affinely constrained, {\em classical} mechanical systems.  For example, in the context of electrical circuits, the hypergraph structure is interpreted as idealised-resistance-free wiring \cite{network,passive}.
Quantising this classical mechanical interpretation of \(\ALR[\R]\) is a notoriously difficult problem in geometric quantisation.
Much research has been done on the geometric quantization of classical mechanics although, oftentimes they work in the nonlinear setting of Lagrangian submanifolds between symplectic manifolds
\cite{guillemin_problems_1979, weinstein_symplectic_2009, benenti_linear_1983, lawruk_special_1975, tulczyjew_category_1984, dold_symplectic_1982, benenti_category_1982, urbanski_structure_1985, kijowski_symplectic_1979}.

There \emph{is} however a natural continuous-variable analogue of the
finite-dimensional stabiliser circuits, called \emph{\textbf{Gaussian quantum
mechanics}} (GQM) \cite{weedbrook_gaussian_2012}, where the states of GQM are genuine quantum states
(ie. normalized, bounded linear maps out of \(\C\)). In this setting, the measurement statistics are described by
Gaussian probability distributions, and the reversible transformations are given by Gaussian unitaries. However, as we have just discussed, GQM is not adequately axiomatised
by \(\ALR[\R]\).  The main goal of this article to obtain a presentation of a  hypergraph
prop that preserves the convenient features of \(\ALR[\R]\) while being
expressive enough to capture GQM. In order to transform infinitely squeezed states to Gaussian states, we need to add some notion of Gaussian blur to them, which will have the effect of making them ``less sharp,'' and thus continuous.

%
%
%
%
%
%

\section{Classical Gaussians}
\label{sec:gauss}
In this section we first review how Gaussian transformations can be composed. Next, we give a novel concrete description of the prop of Gaussian relations of Stein and Samuelson \cite{gaussrel} as a sub-prop of \(\ALR[\C]\).  Finally, we give a presentation for Gaussian relations by extending that of \(\GAA[\R]\).
\subsection{Gaussian distributions and transformations}
\label{ssec:gaussprob}
First, we review multivariate Gaussian distributions.  Given any \(n\in\N\) and \(\Sigma \in \Sym[n][\K]\), denote that \(\Sigma\) is respectively positive semi-definite and positive definite by \(0 \preceq \Sigma\) and \(0\prec \Sigma\).
\begin{definition}
  Given any \(n\in\N\), an \(n\)-variable \textbf{Gaussian distribution}  \(\mathcal{N}(\Sigma,\vec \mu)\)  is a probability distribution on \(\R^n\) determined by some \(0\preceq \Sigma\in\Sym[n][n]\), called the \textbf{covariance matrix} and a vector \(\vec \mu \in \R^n\), called the \textbf{mean}.  The characteristic function of  \(\mathcal{N}(\Sigma,\vec \mu)\) is given by
  \(
    \vec u \mapsto \exp(i{\vec u}^\trans\vec \mu-\frac{1}{2}{\vec u}^\trans \Sigma \vec u)
  \).  
   Moreover, when \(0\prec \Sigma\), then  \(\mathcal{N}(\Sigma,\vec \mu)\) has a density function given by 
  \(
    \vec u \mapsto {\exp(\frac{-1}{2}(\vec u-\vec \mu)^\trans \Sigma^{-1} (\vec u-\vec \mu))}/{\sqrt{(2\pi)^n \det (\Sigma)}}
  \).
\end{definition}
%

Gaussian distributions can be pushed forward along affine transformations to form a prop, following Fritz \cite[\S 6]{Fritz2020}:
\begin{definition}
  The prop, \(\Gauss\), of Gaussian transformations has:
  \begin{itemize}
    \item \textbf{Morphisms}: given by triples \( (A\in \Matrices[m][n],0\preceq\Sigma\in\Sym[m][\R],\vec \mu\in \R^m):n \to m\);
    \item \textbf{Composition}: \((B,\Delta,\vec \nu) \circ (A,\Sigma,\vec \mu) \coloneqq (BA, \Delta + A^\trans  \Sigma A, \vec\nu + A \vec\mu)\) and \(1_n\coloneqq (1_n,0,0)\);
    \item The rest of the prop structure is given pointwise by the direct sum.
  \end{itemize}
\end{definition}

Affine transformations embed in Gaussian transformations via \((T,\vec a)\mapsto (T, 0, \vec a)\). Thus, is natural to ask if Gaussian distributions can also be composed relationally. To find the right category we need to get a hold of the subobjects that will arise by relational composition.  The following notion was first introduced by Willems \cite{Willems}, and later coined by Stein and Samuelson \cite{gaussrel} to this end:
\begin{definition}
An \textbf{extended Gaussian distribution} on \(\R^n\) consists of a linear subspace \(D\subseteq \R^n\) in addition to a Gaussian distribution \(\mathcal{N}(\Sigma,\vec \mu)\) on the quotient \(\R^n/D \cong \R^{n-\dim(D)}\).
\end{definition}
\subsection{Gaussian relations}
\label{ssec:gaussrel}
In this subsection, we refine the morphisms of complex affine Lagrangian relations in order to us to compose extended Gaussian distributions, relationally:
\begin{definition}
  A  complex Lagrangian relation \(S:n\to m\) is:
  \begin{itemize}
    \item \textbf{positive} when for all  \(\vec v\in S\), then \(i\omega_{n,m}\!\left(\overline{\vec v}, \vec v\right) \geq 0\);
    \item \textbf{quasi-real} when it is positive, and for all \([\vec v_I^\trans \ \vec v_O^\trans]^\trans, [\vec w_I^\trans\ \vec w_O^\trans]^\trans \in S\): \[\chi_m(\vec v_O, \vec w_O)-\chi_n(\vec v_I, \vec w_I)=0\quad\text{where}\quad \chi_n ([\vec z_0^\trans\ \vec x_0^\trans]^\trans,[\vec z_1^\trans\ \vec x_1^\trans]^\trans)\coloneqq \overline{\vec{z}_{0}}{}^\trans\vec{x}_1+{\overline{\vec{x}_0}}{}^\trans\vec {z}_1\]
  \end{itemize}
   A nonempty affine Lagrangian relation  \((S+\vec a)\) is positive (resp. quasi-real) when both  \(\vec a\) is real and \(S\) is positive (resp. quasi-real), where empty affine Lagrangian relations are quasi-real.
\end{definition}

Both of these notions induce more refined compact props:
\begin{lemma}\label{lem:closure}
  Both positive and quasi-real complex affine Lagrangian relations are closed under composition and tensor product, and the prop structure is positive and quasi-real.
\end{lemma}

There are different equivalent definitions of positive Lagrangian subspaces:
\begin{proposition}\label{prop:positivechar}

	Let $\mathcal{L}$ be a Lagrangian subspace of $ (\mathbb{C}^{2n},\omega_n) $ with AP-invariant $(E,\phi)$, where $r\coloneqq\rank(\Im(\phi))$ then the following conditions are equivalent:
	\begin{enumerate}
		\item $E$ has real coefficients and $\Im(\phi)\succeq 0 $;
		\item   $\mathcal{L}$ is real symplectomorphic to \(\ker \begin{bmatrix}[c|c] I_n & i\diag(I_r, \vec 0)\end{bmatrix}\).
		
		That is, there exist a real symplectomorphism \(S\) such that 
		\(\interp{\input{SympGauss/iNformnnew.tikz}}_\abbralagr^\zx=\mathcal L\);
		\item  $\mathcal L$ is positive.
	\end{enumerate}
\end{proposition}

Similarly, there are equivalent ways to define quasi-real Lagrangian subspaces:
\begin{proposition}\label{prop:quasirealchar}

	Let $\mathcal{L}$ be a Lagrangian subspace of $ (\mathbb{C}^{2n},\omega_n) $ with AP-invariant $(E,\phi)$, where $r\coloneqq\rank(\Im(\phi))$, then the following conditions are equivalent:
	\begin{enumerate}
		\item $\Re(\phi)= 0 $;
		\item 
		\(\mathcal L\) is real diagonal symplectomorphic to
		\(\ker{\begin{bmatrix}[cc|cc]
			I_r& 0  & i \diag(I_r, \vec 0) & 0 \\ 0 & 0 & 0 & I_{n-r-m}
		\end{bmatrix}}\) for \(m\geq n-r\).
		
		That is, there exists an invertible matrix \(S\) such that 
		\(\interp{\input{Gauss/iNform.tikz}}_\abbralagr^\zx=\mathcal L\);
		\item  $\mathcal L$ is quasi-real.
	\end{enumerate}
\end{proposition}

With propositions~\ref{prop:quasirealchar} and~\ref{prop:positivechar} combined, we have enough information to completely characterize the reduced AP-form  of quasi-real affine Lagrangian subspaces:

\begin{proposition}\label{prop:exgaussapform}
  Nonempty, quasi-real affine Lagrangian subspaces of \((\C^{2n},\omega_n)\) have reduced AP-form given by 5-tuples \((L,i\Sigma, 0, \vec \mu, \varsigma)\), where \(L\in \Matrices[n][n-m][\R], \vec \mu \in \R^{n-m}\) and \(0\preceq\Sigma\in \Sym[n-m][\R]\).  This induces an extended Gaussian distribution \(\mathcal{N}(\Sigma,\vec \mu)\) on \(\R^n/(\ker([I,L]\varsigma))\), so that if \(m=0\), this is a Gaussian distribution on \(\R^n\).

\end{proposition}
%
%
%
%

Annoyingly, due to an incompatibility of conventions because we chose our normal form to have black spiders as vertices, this restricts to affine relations with respect to the {\em colour-swapped}, interpretation of remark~\ref{rem:affrel}, where an antipode must be also be inserted.
\begin{definition}\label{def:gaussrel}
Define the compact prop of  \textbf{Gaussian relations}, \(\GaussRel\),  to be the sub-prop of \(\ALR[\C]\) whose morphisms are quasi-real.
\end{definition}

\(\Gauss\) embeds in \(\GaussRel\), just as affine symplectomorphisms embed in \(\ALR\).
As previously mentioned, the prop of Gaussian relations was first formulated by Stein and Samuelson \cite{gaussrel}, although it is constructed using structured cospans, where composition is given by conditionals instead of directly by relational composition.
Surprisingly, a detour through symplectic geometry makes this into a concrete category of relations! 
%
%
%
\subsection{Graphical presentation}
\label{ssec:gga}

Just as in the previous subsection, we restricted \(\ALR[\C]\) to obtain \(\GaussRel\); here we take the opposite approach and extend the presentation  \(\GAA[\R] \cong \AR[\R]\hookrightarrow \ALR[\C]\) with shearing  by \(i\), \(\input{Gauss/shearing.tikz}\), interpreted as the Gaussian distribution \(\mathcal{N}(1,0)\) on \(\R\), to present \(\GaussRel\):
\begin{definition}
  Let \(\GGA\) be the compact prop given by adding adding a single generator  \(\input{Gauss/void.tikz}:0\to 1\) to \(\GAA[\R]\), called the \textbf{vacuum} such that the scalable vacuum codiscards rotations and effects; ie. so that for all \(n\in \N\), \(\vec{x} \in \R^n\) and \(\theta \in [0,2\pi)\):
   \[\input{Gauss/coerase.tikz}\]
\end{definition}

\begin{theorem}\label{thm:gauss_completeness}
	There is an isomorphism of compact props \(\GGA\cong\GaussRel\).
\end{theorem}
Note that our \emph{proof} of completeness follows from the analogous result for positive Lagrangian relations (theorem~\ref{thm:sympgauss_completeness}) in the following section.  However, this is purely pedagogical, as we need to introduce Gaussian probability theory before Gaussian quantum mechanics.

\section{Symplectic Gaussians}
\label{sec:sguass}

Recall from subsection~\ref{subsec:gsaapplications} that affine Lagrangian subspaces of \((\R^{2n}, \omega_n) \) cannot be canonically represented by bounded linear maps \(L^2(\R^0)\to L^2(\R^n)\) because they would correspond to infinitely-squeezed states.  It turns out that by dropping the requirement from the previous section that affine Lagrangian relations are quasi-real, we obtain a semantics which can accommodate for both such infinitely-squeezed state and ``Gaussian'' quantum states.

\subsection{Gaussian quantum mechanics}
\label{ssec:gaussqm}
In this section, we review Gaussian quantum mechanics, and their Wigner representation.  See the review article of Weedbrook et al. for a detailed reference \cite{weedbrook_gaussian_2012}.  Brask provides an abridged, yet helpful reference \cite{Brask}. 
Recall from subsection~\ref{subsec:gsaapplications}, that the continuous variable analogue of a qudit, the qumode has state space given by the Hilbert space of square integrable complex functions on \(\R\), \(L^2(\R)\); where the state space of  \(n\)-qumodes is identified with the Hilbert space of square integrable complex functions on \(\R^n\), \(L^2(\R^n)\): 

\begin{definition}
	An \(n\)-qumode \textbf{Gaussian state} 
	\(\varphi\in L^2 (\R^n)\) has the form:
	\[
    \varphi\left(\vec x\right)=\exp(i\alpha)
	  \exp(i\vec{s}^\trans \vec x) \sqrt[4]{{\det(\Im(\Phi))}/{\pi^n }}\exp({i}(\vec x -
	  \vec t\,)^\trans \Phi(\vec x-\vec t\,)/2)
	\]
  where\ \(\alpha \in [0,2\pi)\), \(\vec s,\vec t \in \R^n\), and \(\Phi\in \Sym[n][\C]\) with \(\Im(\Phi)\succ 0\). We call such a matrix \(\Phi\)
	a \textbf{phase matrix}, and the vector \(\begin{bmatrix} {\vec s}\ {}^\trans &
	{\vec t}\ {}^\trans \end{bmatrix}{}^\trans \in \R^{2n}\) a \textbf{displacement}. Together, they
	characterise the Gaussian state up to the ``global phase'' \(\exp(i\alpha)\).
There is an important Gaussian state  on \(L^2(\R)\) called the \textbf{vacuum} with trivial displacement and phase matrix \(i\).
\end{definition}

Importantly, unlike Dirac deltas, Gaussian states are continuous and normalized, so that they are quantum states.  In other words \(\phi\) can be regarded as a wave-function.

The connection between Gaussian quantum states and the Gaussian probability
distributions which we discussed in section~\ref{sec:gauss} is revealed by considering
the symplectic vector space \((\R^{2n},\omega_n)\) as the \emph{\textbf{phase space}} for the
quantum system on \(L^2(\R^n)\), on which Gaussian states correspond to particular Gaussian distributions.

\begin{definition}
  The \textbf{Wigner transform} of \(\varphi \in L^2 (\R^n)\) is a \emph{real-valued} \(W\in L^2 (\R^{2n})\):
  \[
    W\left(\begin{bmatrix}\vec q \\ \vec p \end{bmatrix}\right)\coloneqq \frac{1}{\pi^n }\int_{\R^n } \overline{\varphi}\left(\vec q+\vec \xi\right)\varphi\left(\vec q -\vec \xi\right)\exp\left(2i\vec{\xi}^\trans \vec p\right)\, d\vec \xi
    \quad\quad
    \text{for all } \vec p, \vec q\in\R^n 
  \]
\end{definition}
%

\begin{proposition}\label{prop:gauss_wigner}
	The Wigner transform of an \(n\)-qumode Gaussian state with phase
	matrix \(\Phi\) and displacement \(\vec \mu\) is the density function
	of the Gaussian distribution \({\mathcal N}(\Sigma,\vec \mu)\) on
	\(\R^{2n}\) where:
	\[
	  \Sigma\coloneqq \begin{bmatrix}
	  \Im (\Phi) + \Re(\Phi)\Im(\Phi)^{-1} \Re(\Phi)  & -\Re(\Phi)\Im(\Phi)^{-1} \\ -\Im(\Phi)^{-1} \Re(\Phi) & \Im(\Phi)^{-1}
	  \end{bmatrix}
	\]
	Moreover, this is a bijection up to global phase, as conversely,
	Gaussian distributions \({\mathcal N}\left(\Delta,\vec \mu\right)\)
	on \(\R^{2n}\), with
  \[
	  \Delta\coloneqq
	  \begin{bmatrix}
	  	A & B \\ B^\trans & C 
	  \end{bmatrix}
  \]
  for \(A,C \in \Sym[n][\R]\), \(B \in \Matrices[n][n][\R]\) such that \(\det(\Delta)=1\) and \(\Delta + i \Omega_n \succeq 0\) induce Gaussian states
  with phase matrix \(\Phi\coloneqq -BC^{-1} + i C^{-1}\).
\end{proposition}

\begin{definition}
\label{def:gaussunitary}
The \(n\)-qumode \textbf{Gaussian unitaries} are the unitaries on \(L^2(\R^n)\) whose action on Gaussian states by postcomposition is closed. 
\end{definition}

To characterize Gaussian unitaries up to global phase, it suffices to know how their action on Gaussian states effects their Wigner representation:
\begin{lemma}
\label{lem:weil}
Any Gaussian unitary \(U\) on \(L^2(\R^n)\)
can be represented, up to global phase, by an affine symplectomorphism \((\vec x, S)\) on \((\R^{2n},\omega_n)\) 
such that the action of \(U\) on the Wigner representation of  Gaussian
states is described by 
\(\Sigma \mapsto S \Sigma S^\trans\) and \(\vec \mu  \mapsto \vec \mu + \vec x\).
\end{lemma}

In other words, since the pair \((\Sigma,\vec \mu)\) does not account for
global phases, Gaussian unitaries correspond to a projective representation of
the real affine symplectic group \cite{weedbrook_gaussian_2012}. Although
this representation \emph{can} be deprojectivised to a genuine
unitary representation, called the \emph{metaplectic} representation
\cite{de_gosson_symplectic_2011}, our graphical calculi do not account for
global scalars in the first place. As a result, we content ourselves with
the projective representation described by lemma~\eqref{lem:weil}.  Finally,
there is also a good notion of Gaussian quantum effect:
\begin{definition}
Take a Gaussian state \(\varphi\)   represented by the Gaussian distribution \({\mathcal N}(\Sigma, \vec \mu)\) on  \(\R^{2n} \oplus \R^{2m}\), and another \(\psi\) represented by \({\mathcal N}(\Delta, \vec \nu)\) on \(\R^{2n}\) where:
\begin{equation*}
  \Sigma\coloneqq
  \begin{bmatrix}
      \Sigma_{n,n}        & \Sigma_{n,m} \\
      \Sigma_{n,m}^\trans & \Sigma_{m,m}
    \end{bmatrix}
    \quad\quad
    \vec \mu
    \coloneqq
     \begin{bmatrix}
      \vec \mu_n \\
      \vec \mu_m
    \end{bmatrix}
\end{equation*}
The \textbf{projection} of \(\varphi\) onto \(\psi\) is the Gaussian state represented by the Gaussian distribution induced by generalised Schur complement:
\begin{equation}
  \label{eq:schur_complement}
  {\mathcal N}\left(
    \Sigma_{m,m} - \Sigma_{n,m} (\Sigma_{n,n} + \Delta)^{+} \Sigma_{n,m}^\trans,
  \vec \mu_m -  \Sigma_{n,m} (\Sigma_{n,n} + \Delta)^{+} (\vec \mu_n - \vec \nu)
  \right)
\end{equation}
Where  \((\Sigma_{n,n} + \Delta)^+\) denotes the Moore-Penrose pseudo-inverse of \(\Sigma_{n,n} + \Delta\) (this coincides with the inverse when the matrix is invertible).
\end{definition}
These projections correspond precisely to the corresponding projections
of quantum states in Hilbert space, lifted to the covariance matrix
representation.
\begin{definition}
  The prop of quantum Gaussian transformations, \(\QGauss\), has:
  \begin{itemize}
    \item \textbf{morphisms} generated by \(n\)-qumode Gaussians states with type \(0 \to n\),  \(n\)-qumode Gaussian unitaries with type \(n \to n\), and projections onto  \(n\)-qumode Gaussian states of type \(n \to 0\) for all \(n\in\N\), all up to global phase;
    \item \textbf{composition} given by the action in lemma~\ref{lem:weil} and
    equation~\eqref{eq:schur_complement}; 
    \item the rest of the prop structure is given pointwise by the direct sum.
  \end{itemize}
\end{definition}

It is easy to see how there is an embedding  \(\Gauss \to \QGauss\),  by doubling.

\subsection{Positive Lagrangian relations}
\label{ssec:positive}
Recall the characterization of Gaussian relations of definition~\ref{def:gaussrel} in terms of quasi-real complex affine Lagrangian relations.  In this section, we relax the requirement that these relations need to be quasi-real, yet retain that they must be positive:

\begin{definition}
  Let  \(\QGaussRel\) denote the sub-prop of \(\ALR[\C]\) of positive affine Lagrangian relations.
\end{definition}

From proposition~\ref{prop:positivechar}, it follows that:
\begin{proposition}\label{prop:AP_positive}
  Nonempty positive affine Lagrangian subspaces of \((\C^{2n},\omega_n)\) have reduced AP-form given by 5-tuples \((L,\Delta+i\Sigma, \vec \nu, \vec \mu, \varsigma)\), where \(L\in \Matrices[n][n-m][\R], \vec \mu \in \R^{n-m}\), \(\vec \nu \in \R^m\) and \(\Delta,\Sigma\in \Sym[n-m][\R]\) with \(0\preceq\Sigma\). 
\end{proposition}

By regarding the Wigner representation of \(n\)-qumode Gaussian states as positive, affine Lagrangian subspaces of \((\C^{2n},\omega_n)\), effects as subspaces of \((\C^{2n},-\omega_n)\), and taking the complex graphs of the affine symplectomorphisms given by Gaussian unitaries, it is clear that:
\begin{proposition}
	There is an embedding of props
	\(\QGauss \to \QGaussRel\).
\end{proposition}
%
\subsection{Graphical presentation}
\label{ssec:QGaussRel}
Just as in the previous subsection, we restricted \(\ALR[\C]\) to obtain \(\QGaussRel\); here we take the opposite approach and extend  the presentation \(\GSA[\R]\cong\ALR[\R]\hookrightarrow\ALR[\C]\) with shearing  by \(i\), \(\input{Gauss/shearing.tikz}\), interpreted as the vacuum, to give a presentation for \(\QGaussRel\): 
\begin{definition}
  The compact prop \(\GQGA\) is presented by adding adding a single generator
  \(\input{Gauss/void.tikz}:0\to 1\) to \(\GSA[\R]\) modulo the equations
  codiscarding symplectic rotations and effects. Equationally, we impose that
  for all \(n\in\N\), \(a,b \in \R\),  \(\theta \in (-\pi,\pi)\), and \(\vartheta \in [0,2\pi)
  \):
  \[\input{SympGauss/coerase.tikz}\]
\end{definition}
\begin{theorem}\label{thm:sympgauss_completeness}
  There is an isomorphism of compact props \(\GQGA\cong \QGaussRel\).
\end{theorem}

This elucidates the connection of positive Lagrangian relations to Gaussian relations:
\begin{corollary}
There is a symmetric monoidal embedding \([-]_\GGA^{\GQGA}:\GQGA\to\GGA\) acting on objects by \([n]_\GGA^{\GQGA}\coloneqq 2n\) and on generators by:
  \[
    \interp{\input{SympGauss/cqgadouble_wspider_lhs.tikz}}_\GGA^{\GQGA}\hspace*{-.25cm}
    \coloneqq
    \input{SympGauss/cqgadouble_wspider_rhs.tikz}
    \quad
    \interp{\input{SympGauss/cqgadouble_bspider_lhs.tikz}}_\GGA^{\GQGA}\hspace*{-.25cm}
    \coloneqq
    \input{SympGauss/cqgadouble_bspider_rhs.tikz}
    \quad
    \interp{\input{SympGauss/cqgadouble_void_lhs.tikz}}_\GGA^{\GQGA}\hspace*{-.25cm}
    \coloneqq
    \input{SympGauss/cqgadouble_void_lhs.tikz}
  \]
\end{corollary}

So that we see that positive affine Lagrangian relations can be regarded as extended Gaussian distributions on the phase space.

Putting things together we have the following commutative diagram of props, which wraps around on the left and right:
\[
\xymatrixrowsep{3mm}\xymatrixcolsep{6mm}
\xymatrix{
 \cdots  &  \GAA[\R] \ar[d]_{\cong} \ar@{>->}[r]   \ar@{>-}[l]                         &  \GGA \ar@{>->}[r] \ar[d]_{\cong}       & \GQGA \ar@{>->}[r] \ar[d]_{\cong}  & \GSA[\C] \ar[d]_{\cong}     &                             \GSA[\R]  \ar@{>->}[l] \ar[d]_{\cong} \ar@{>->}@/_1pc/[ll]   & \ar@{->}[l] \cdots\\
 \cdots  &   \AR[\R]  \ar@{>->}[r]   \ar@{>-}[l]                                                  &  \GaussRel \ar@{>->}[r]                       & \QGaussRel \ar@{>->}[r]                   & \ALR[\C]                              &                            \ALR[\R]  \ar@{>->}[l]                         \ar@{>->}@/^1pc/[ll]   & \ar@{->}[l] \cdots\\
            &   {\sf AffMat}_{\R} \ar@{>->}[r]  \ar@{>->}[u]                                   &   \Gauss \ar@{>->}[r] \ar@{>->}[u]    & \QGauss \ar@{>->}[u]                      &                                           &                                                                                             &
}\]


\subsection{Picturing Gaussian quantum states}
There is a close link between the states of \(\GQGA\) and the
graph-theoretic representation of Gaussian states of Menicucci et al.
\cite{menicucci_graphical_2011}. This mirrors the relationship between quantum
graph states \cite{van_den_nest_graphical_2004, zhou_quantum_2003} and states
in the stabiliser ZX-calculus for both even \cite{backens_zx-calculus_2014}
and odd-prime-dimensional quantum systems \cite{booth_complete_2022}.

Recall from \cite{menicucci_graphical_2011} that any Gaussian state can be
uniquely associated to a matrix \(U + iV\) where \(U,V \in \Sym[n][\R]\)
and \(V\) is positive-definite. This matrix corresponds exactly to the phase
matrix of the associated Gaussian state. As a result, it is represented
in \(\GSA[\C]\) by \input{figures/SympGauss/Menicucci.tikz}. Therefore,
there is a very natural embedding of the graph representation of
\cite{menicucci_graphical_2011} into \(\GQGA\hookrightarrow\GSA[\C]\) given
by the following algorithm:
\begin{enumerate}
  \item To each vertex of the graph, associate a grey spider. If that vertex
  has a self-edge with weight \(w\), give the corresponding grey spider the
  label \((0,w)\).
  \item To each edge between two different vertices with weight \(w\),
  associate an edge between the corresponding spiders mediated by an
  \(a\)-labelled box.
  \item Connect each spider in the graph to a (unique) output.
  \item  Using \(\TextFusionB\), for each vertex unfuse the complex symplectic phase \(ai\)  of vertices into one-legged spiders, and then use \(\TextColour\), to turn this spider into a vacuum state connected to the original vertex with edge-weight \(1/\sqrt{a}\).
\end{enumerate}
This amounts to identifying the phase matrices of the graph representation
and the \(\GQGA\) diagram.  For example:
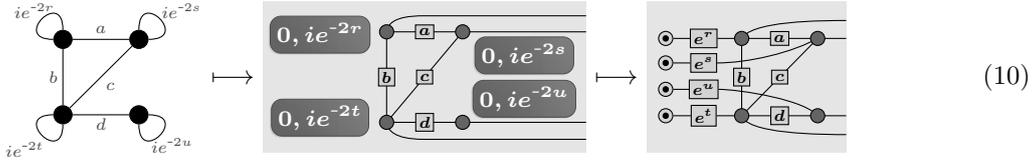
\begin{equation}
  \label{eq:Menicucci_example}
  \input{figures/SympGauss/Menicucci_graph.tikz}
  \longmapsto \input{figures/SympGauss/Menicucci_diagram.tikz}
  \longmapsto \input{figures/SympGauss/Menicucci_diagram_voids.tikz}
\end{equation}

In other words, we have extended the diagrammatic representation of Menicucci et al.
with a representation of Gaussian unitaries and effects as diagrams rather than
through their action on Gaussian states, along with an equational theory for
manipulating these diagrams. By completeness of \(\GQGA\), the action
of Gaussian unitaries on Gaussian states is derivable. We also have a natural
representation of translations, although the reduced-AP form of translated
Gaussian states is not as simple as equation~\eqref{eq:Menicucci_example}
since it can include internal vertices which carry part of the translation.

\subsection{Picturing quantum protocols}
\label{ssec:teleportation}
Infinitely-squeezed eigenstates of the position and momentum displacement operators \(\hat p\) and \(\hat q\) are represented  by  \(\ket{p : \hat{p}} \mapsto \input{figures/SympGauss/momentum_eigenstate.tikz}
\) and \(\ket{q: \hat{q}} \mapsto \input{figures/SympGauss/position_eigenstate.tikz}
\). The vacuum
state is represented by \(\ket{\phi} \mapsto \input{figures/SympGauss/vacuum_state.tikz}\).
We can also translate a universal gate set for Gaussian unitaries into our
language; consisting of the \textbf{displaced shear} of position, the \textbf{Fourier
transform} and \textbf{weighted CNOT} \cite{gu_quantum_2009} with \(a,b\in\R\):
\begin{equation}
  \exp(i(a\hat{q} + b\hat{q}^2)) \mapsto \input{figures/SympGauss/shear.tikz}
  \qquad
  \exp(i\frac{\pi}{2}(\hat{q}^2 + \hat{p}^2)) \mapsto \input{figures/SympGauss/fourier.tikz}
  \qquad
  \exp(ia(\hat{q} \otimes \hat{p})) \mapsto \input{figures/SympGauss/cnot.tikz}
\end{equation}

We can now give a straightforward graphical proof of Braunstein and Kimble's continuous variable quantum teleportation protocol \cite{braunstein_teleportation_1998} (to be compared with the graphical proof in finite dimensions eg. \cite[\S5.4]{wetering} or \cite[\S9.2.7]{pqp}), where Alice and Bob share a Gaussian Bell state with covariance of position $0<\varepsilon \in \R$.  Alice records the homodyne measurement outcome \((a,b)\in\R^2\) in the Bell basis, and sends it to Bob, who performs the phase correction \({\hat p}^{- b} {\hat q}^{- a} \).  This is stated graphically and simplified  as follows:
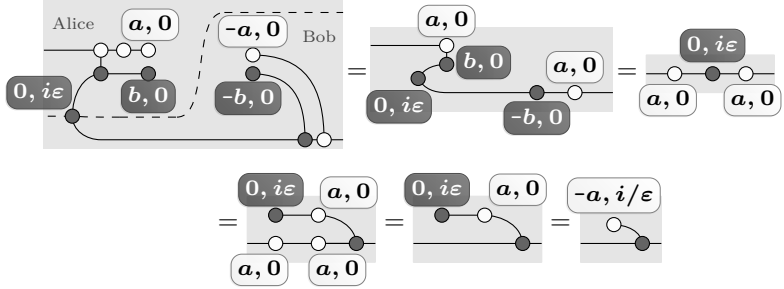
\[
\input{figures/SympGauss/teleportation_rotated.tikz}
\]
The result is a quantum channel with an error; however, in the infinitely-squeezed limit of \(\varepsilon=0\)  Alice teleports a perfect channel to Bob: $\input{figures/SympGauss/teleportation_squeezed.tikz}$. In the literature, the infinitely-squeezed Bell state is often haphazardly represented by the non-convergent integral \(\int_{\R}  | p : \hat p \rangle \otimes | p : \hat p \rangle \, dp\). However, we reiterate that this expression \emph{can} be represented in our calculus by an unlabeled 2-legged grey node; and it is interpreted \emph{soundly} in our semantics by replacing  the integral over position eigenstates \(\ket{q : \hat q}\) with the  \emph{relational composition} in the phase-space.

Note that in this derivation we could have used the vacuum generator instead
of the \((0,i)\) phase, just like in equation~\eqref{eq:Menicucci_example}.

\subsection{Picturing quantum optics}
\label{ssec:lov}
\(\GQGA\) allows us to reason diagrammatically about Gaussian quantum
mechanics. Perhaps the most important application of Gaussian quantum
mechanics is for quantum optics.  The \(\LOv\)-calculus is a complete
diagrammatic language for \emph{passive linear} quantum optics
\cite{clement_lov-calculus_2022}. It is a prop defined on the following set
of generators, for all \(\theta, \varphi\in\R\);

  the \textit{\textbf{phase shifter}},  \textit{\textbf{wave plate}}, \textit{\textbf{beamsplitter}}, \textit{\textbf{polarising beamsplitter}}, \textit{\textbf{vacuum state}} and \textit{\textbf{effect}}:
  \begin{equation*}
    \input{figures/SympGauss/lov/convtp-phase-shift.tikz}:1 \to 1
    \quad \input{figures/SympGauss/lov/pol-rot.tikz}:1 \to 1
    \quad \input{figures/SympGauss/lov/beamsplitter.tikz}:2 \to 2
    \quad \input{figures/SympGauss/lov/bs.tikz}:2 \to 2
    \quad\input{figures/SympGauss/lov/gene-0.tikz}:0 \to 1 \qquad \input{figures/SympGauss/lov/detector-0.tikz}:1 \to 0
  \end{equation*}

There is a symmetric monoidal embedding \(\LOv \to \GQGA\), given by \(\interp{n}_{\GQGA}^{\LOv} = 2n\) and:

\noindent
\(
\begin{array}{rlcrlcrl}
  \label{eq:optics}
  \interp{\input{figures/SympGauss/lov/pol-rot.tikz}}_{\GQGA}^{\LOv}\!\!
    =&\input{figures/SympGauss/lov/rotation_ZX.tikz} &&
  \interp{\input{figures/SympGauss/lov/convtp-phase-shift.tikz}}_{\GQGA}^{\LOv}\!\!
    =&\input{figures/SympGauss/lov/phase_ZX.tikz} &&
  \interp{\input{figures/SympGauss/lov/detector-0.tikz}}^\LOv_{\GQGA}\!\!
    =&\input{figures/SympGauss/lov/vacuum_effect.tikz}\\ \vspace*{-.2cm}\\
  \interp{\input{figures/SympGauss/lov/bs.tikz}}_{\GQGA}^{\LOv}\!\!
    =&\input{figures/SympGauss/lov/polarising_bs.tikz}&&
  \interp{\input{figures/SympGauss/lov/beamsplitter.tikz}}_{\GQGA}^{\LOv}\!\!
    =&\input{figures/SympGauss/lov/beamsplitter_ZX.tikz}&&
    \interp{\input{figures/SympGauss/lov/gene-0.tikz}}^\LOv_{\GQGA}\!\!
    =&\input{figures/SympGauss/lov/vacuum_state.tikz} 
\end{array}
\)

The doubling of the objects of \(\LOv\) accounts for the two polarisation
modes of light (typically the \emph{vertical} and \emph{horizontal}
modes).  By working in \(\GQGA\), we already have the ability to model
additional quantum-optical components and Gaussian unitaries that are
not present in \(\mathsf{LOv}\), given by the real affine symplectic group
\cite{weedbrook_gaussian_2012}. For instance, active squeezing of the vertical
polarisation mode, and shearing between the position and momentum operators,
are respectively given for \(a \in \R\) by:
\begin{equation}
  \input{figures/SympGauss/lov/squeezing.tikz}
  \quad\quad\quad
  \input{figures/SympGauss/lov/shear.tikz}
\end{equation}
Furthermore \(\GQGA\) is compact closed, as opposed to the usual infinite
dimensional Hilbert space semantics,  so by contrast we have much more freedom
for rewriting in this setting.

\bibliography{styles/bibliography}

\begin{thebibliography}{10}

\bibitem{kijowski_symplectic_1979}
A symplectic framework for field theories.
\newblock URL: \url{http://link.springer.com/10.1007/3-540-09538-1}, \href
  {https://doi.org/10.1007/3-540-09538-1} {\path{doi:10.1007/3-540-09538-1}}.

\bibitem{abramsky_h*-algebras_2010}
Samson Abramsky and Chris Heunen.
\newblock H*-algebras and nonunital frobenius algebras: first steps in
  infinite-dimensional categorical quantum mechanics.
\newblock URL: \url{http://arxiv.org/abs/1011.6123}, \href
  {https://arxiv.org/abs/1011.6123} {\path{arXiv:1011.6123}}.

\bibitem{asavanant_time-domain_2019}
Warit Asavanant, Yu~Shiozawa, Shota Yokoyama, Baramee Charoensombutamon, Hiroki
  Emura, Rafael~N. Alexander, Shuntaro Takeda, Jun-ichi Yoshikawa, Nicolas~C.
  Menicucci, Hidehiro Yonezawa, and Akira Furusawa.
\newblock Time-{Domain} {Multiplexed} 2-{Dimensional} {Cluster} {State}:
  {Universal} {Quantum} {Computing} {Platform}.
\newblock {\em Science}, 366(6463):373--376, October 2019.
\newblock arXiv: 1903.03918.
\newblock URL: \url{http://arxiv.org/abs/1903.03918}, \href
  {https://doi.org/10.1126/science.aay2645}
  {\path{doi:10.1126/science.aay2645}}.

\bibitem{backens_zx-calculus_2014}
Miriam Backens.
\newblock The {ZX}-calculus is complete for stabilizer quantum mechanics.
\newblock 16(9):093021.
\newblock URL: \url{http://arxiv.org/abs/1307.7025}, \href
  {https://arxiv.org/abs/1307.7025} {\path{arXiv:1307.7025}}, \href
  {https://doi.org/10.1088/1367-2630/16/9/093021}
  {\path{doi:10.1088/1367-2630/16/9/093021}}.

\bibitem{zh}
Miriam Backens and Aleks Kissinger.
\newblock {ZH}: A complete graphical calculus for quantum computations
  involving classical non-linearity.
\newblock {\em Electronic Proceedings in Theoretical Computer Science},
  287:23--42, January 2019.
\newblock \href {https://doi.org/10.4204/eptcs.287.2}
  {\path{doi:10.4204/eptcs.287.2}}.

\bibitem{Backens2021}
Miriam Backens, Hector Miller-Bakewell, Giovanni de~Felice, Leo Lobski, and
  John van~de Wetering.
\newblock There and back again: A circuit extraction tale.
\newblock {\em Quantum}, 5:421, March 2021.
\newblock URL: \url{http://dx.doi.org/10.22331/q-2021-03-25-421}, \href
  {https://doi.org/10.22331/q-2021-03-25-421}
  {\path{doi:10.22331/q-2021-03-25-421}}.

\bibitem{network}
John~C. Baez, Brandon Coya, and Franciscus Rebro.
\newblock Props in network theory.
\newblock {\em Theory and Applications of Categories}, 33(25):727--783, 2018.
\newblock URL: \url{http://www.tac.mta.ca/tac/volumes/33/25/33-25.pdf}.

\bibitem{control}
John~C. Baez and Jason Erbele.
\newblock Categories in control.
\newblock {\em Theory and Applications of Categories}, 30(24):836--881, 2015.
\newblock URL: \url{http://www.tac.mta.ca/tac/volumes/30/24/30-24.pdf}.

\bibitem{passive}
John~C. Baez and Brendan Fong.
\newblock A compositional framework for passive linear networks.
\newblock {\em Theory and Applications of Categories}, 33(38):pp 1158--1222,
  2018.
\newblock URL: \url{http://www.tac.mta.ca/tac/volumes/33/38/33-38.pdf}.

\bibitem{bartolucci_fusion-based_2023}
Sara Bartolucci, Patrick Birchall, Hector Bombín, Hugo Cable, Chris Dawson,
  Mercedes Gimeno-Segovia, Eric Johnston, Konrad Kieling, Naomi Nickerson,
  Mihir Pant, Fernando Pastawski, Terry Rudolph, and Chris Sparrow.
\newblock Fusion-based quantum computation.
\newblock {\em Nature Communications}, 14(1):912, February 2023.
\newblock Number: 1 Publisher: Nature Publishing Group.
\newblock URL: \url{https://www.nature.com/articles/s41467-023-36493-1}, \href
  {https://doi.org/10.1038/s41467-023-36493-1}
  {\path{doi:10.1038/s41467-023-36493-1}}.

\bibitem{benenti_category_1982}
Sergio Benenti.
\newblock The category of symplectic reductions.
\newblock In {\em Proceedings of the meeting "Geometry and Physics"}.

\bibitem{benenti_linear_1983}
Sergio Benenti.
\newblock Linear symplectic relations.
\newblock In Albert Crumeyrolle and J.~Griffone, editors, {\em Symplectic
  Geometry}. Pitman.

\bibitem{dpla}
Guillaume Boisseau and Robin Piedeleu.
\newblock Graphical piecewise-linear algebra.
\newblock In {\em Lecture Notes in Computer Science}, pages 101--119. Springer
  International Publishing, 2022.
\newblock URL: \url{https://arxiv.org/pdf/2111.03956.pdf}, \href
  {https://doi.org/10.1007/978-3-030-99253-8_6}
  {\path{doi:10.1007/978-3-030-99253-8_6}}.

\bibitem{Boisseau2022}
Guillaume Boisseau and Pawe{\l} Soboci{\'n}ski.
\newblock String diagrammatic electrical circuit theory.
\newblock {\em Electronic Proceedings in Theoretical Computer Science},
  372:178–191, November 2022.
\newblock \href {https://doi.org/10.4204/eptcs.372.13}
  {\path{doi:10.4204/eptcs.372.13}}.

\bibitem{dpa}
Filippo Bonchi, Alessandro~Di Giorgio, and Pawe{\l} Soboci{\'n}ski.
\newblock Diagrammatic polyhedral algebra.
\newblock In Mikolaj Bojanczyk and Chandra Chekuri, editors, {\em 41st {IARCS}
  Annual Conference on Foundations of Software Technology and Theoretical
  Computer Science, {FSTTCS} 2021, December 15-17, 2021, Virtual Conference},
  volume 213 of {\em LIPIcs}, pages 40:1--40:18. Schloss Dagstuhl -
  Leibniz-Zentrum f{\"{u}}r Informatik, 2021.
\newblock \href {https://doi.org/10.4230/LIPIcs.FSTTCS.2021.40}
  {\path{doi:10.4230/LIPIcs.FSTTCS.2021.40}}.

\bibitem{Bonchi2019}
Filippo Bonchi, Joshua Holland, Robin Piedeleu, Paweł Sobociński, and Fabio
  Zanasi.
\newblock Diagrammatic algebra: from linear to concurrent systems.
\newblock {\em Proceedings of the ACM on Programming Languages},
  3(POPL):1–28, January 2019.
\newblock URL: \url{https://www.ioc.ee/~pawel/papers/popl19.pdf}, \href
  {https://doi.org/10.1145/3290338} {\path{doi:10.1145/3290338}}.

\bibitem{gaa}
Filippo Bonchi, Robin Piedeleu, Pawe{\l} Soboci{\'n}ski, and Fabio Zanasi.
\newblock Graphical affine algebra.
\newblock In {\em 2019 34th Annual ACM/IEEE Symposium on Logic in Computer
  Science (LICS)}, pages 1--12. IEEE, 2019.
\newblock URL: \url{http://www.zanasi.com/fabio/files/paperLICS19.pdf}, \href
  {https://doi.org/10.1109/LICS.2019.8785877}
  {\path{doi:10.1109/LICS.2019.8785877}}.

\bibitem{gla}
Filippo Bonchi, Pawe{\l} Soboci{\'n}ski, and Fabio Zanasi.
\newblock Interacting {H}opf algebras.
\newblock {\em Journal of Pure and Applied Algebra}, 221(1):144--184, 2017.
\newblock URL: \url{https://arxiv.org/pdf/1403.7048.pdf}, \href
  {https://doi.org/10.1016/j.jpaa.2016.06.002}
  {\path{doi:10.1016/j.jpaa.2016.06.002}}.

\bibitem{Bonchi2014}
Filippo Bonchi, Paweł Sobociński, and Fabio Zanasi.
\newblock {\em A Categorical Semantics of Signal Flow Graphs}, page 435–450.
\newblock Springer Berlin Heidelberg, 2014.
\newblock URL: \url{https://hal.science/hal-02134182/file/sfg.pdf}, \href
  {https://doi.org/10.1007/978-3-662-44584-6_30}
  {\path{doi:10.1007/978-3-662-44584-6_30}}.

\bibitem{Bonchi2017}
Filippo Bonchi, Paweł Sobociński, and Fabio Zanasi.
\newblock The calculus of signal flow diagrams i: Linear relations on streams.
\newblock {\em Information and Computation}, 252:2–29, February 2017.
\newblock URL: \url{https://www.ioc.ee/~pawel/papers/sfg1.pdf}, \href
  {https://doi.org/10.1016/j.ic.2016.03.002}
  {\path{doi:10.1016/j.ic.2016.03.002}}.

\bibitem{booth_complete_2022}
Robert~I. Booth and Titouan Carette.
\newblock Complete {ZX}-calculi for the stabiliser fragment in odd prime
  dimensions.
\newblock In Stefan Szeider, Robert Ganian, and Alexandra Silva, editors, {\em
  47th International Symposium on Mathematical Foundations of Computer Science
  ({MFCS} 2022)}, volume 241 of {\em Leibniz International Proceedings in
  Informatics ({LIPIcs})}, pages 24:1--24:15. Schloss Dagstuhl –
  Leibniz-Zentrum für Informatik.
\newblock {ISSN}: 1868-8969.
\newblock URL: \url{https://drops.dagstuhl.de/opus/volltexte/2022/16822}, \href
  {https://doi.org/10.4230/LIPIcs.MFCS.2022.24}
  {\path{doi:10.4230/LIPIcs.MFCS.2022.24}}.

\bibitem{gsa}
Robert~I. Booth, Titouan Carette, and Cole Comfort.
\newblock Graphical symplectic algebra, 2024.
\newblock URL: \url{https://arxiv.org/abs/2401.07914}, \href
  {https://arxiv.org/abs/arXiv:2401.07914} {\path{arXiv:arXiv:2401.07914}}.

\bibitem{Coherentcontrol}
Cyril Branciard, Alexandre Clément, Mehdi Mhalla, and Simon Perdrix.
\newblock Coherent control and distinguishability of quantum channels via
  pbs-diagrams.
\newblock Schloss Dagstuhl – Leibniz-Zentrum f\"{u}r Informatik, 2021.
\newblock URL:
  \url{https://drops.dagstuhl.de/entities/document/10.4230/LIPIcs.MFCS.2021.22},
  \href {https://doi.org/10.4230/LIPICS.MFCS.2021.22}
  {\path{doi:10.4230/LIPICS.MFCS.2021.22}}.

\bibitem{Brask}
Jonatan~Bohr Brask.
\newblock Gaussian states and operations -- a quick reference, 2021.
\newblock URL: \url{https://arxiv.org/abs/2102.05748}, \href
  {https://arxiv.org/abs/arXiv:2102.05748} {\path{arXiv:arXiv:2102.05748}}.

\bibitem{braunstein_teleportation_1998}
Samuel~L Braunstein and H~J Kimble.
\newblock Teleportation of continuous quantum variables.
\newblock 80(4):4.
\newblock \href {https://doi.org/10.1103/PhysRevLett.80.869}
  {\path{doi:10.1103/PhysRevLett.80.869}}.

\bibitem{braunstein_quantum_2005}
Samuel~L. Braunstein and Peter van Loock.
\newblock Quantum information with continuous variables.
\newblock {\em Reviews of Modern Physics}, 77(2):513--577, June 2005.
\newblock arXiv: quant-ph/0410100.
\newblock URL: \url{http://arxiv.org/abs/quant-ph/0410100}, \href
  {https://doi.org/10.1103/RevModPhys.77.513}
  {\path{doi:10.1103/RevModPhys.77.513}}.

\bibitem{clement_lov-calculus_2022}
Alexandre Clément, Nicolas Heurtel, Shane Mansfield, Simon Perdrix, and
  Benoît Valiron.
\newblock {LOv}-calculus: A graphical language for linear optical quantum
  circuits.
\newblock In Stefan Szeider, Robert Ganian, and Alexandra Silva, editors, {\em
  47th International Symposium on Mathematical Foundations of Computer Science
  ({MFCS} 2022)}, volume 241 of {\em Leibniz International Proceedings in
  Informatics ({LIPIcs})}, pages 35:1--35:16. Schloss Dagstuhl –
  Leibniz-Zentrum für Informatik.
\newblock {ISSN}: 1868-8969.
\newblock URL: \url{https://drops.dagstuhl.de/opus/volltexte/2022/16833}, \href
  {https://doi.org/10.4230/LIPIcs.MFCS.2022.35}
  {\path{doi:10.4230/LIPIcs.MFCS.2022.35}}.

\bibitem{Cockett2021}
Robin Cockett, Cole Comfort, and Priyaa Srinivasan.
\newblock Dagger linear logic for categorical quantum mechanics.
\newblock {\em Logical Methods in Computer Science}, Volume 17, Issue 4,
  November 2021.
\newblock URL: \url{http://dx.doi.org/10.46298/LMCS-17(4:8)2021}, \href
  {https://doi.org/10.46298/lmcs-17(4:8)2021}
  {\path{doi:10.46298/lmcs-17(4:8)2021}}.

\bibitem{zx}
Bob Coecke and Ross Duncan.
\newblock Interacting quantum observables: categorical algebra and
  diagrammatics.
\newblock {\em New Journal of Physics}, 13(4):043016, April 2011.
\newblock \href {https://doi.org/10.1088/1367-2630/13/4/043016}
  {\path{doi:10.1088/1367-2630/13/4/043016}}.

\bibitem{pqp}
Bob Coecke and Aleks Kissinger.
\newblock {\em Picturing Quantum Processes: A First Course on Quantum Theory
  and Diagrammatic Reasoning}, page 28–31.
\newblock Springer International Publishing, 2018.
\newblock URL: \url{http://dx.doi.org/10.1007/978-3-319-91376-6_6}, \href
  {https://doi.org/10.1007/978-3-319-91376-6_6}
  {\path{doi:10.1007/978-3-319-91376-6_6}}.

\bibitem{comfort_algebra_2023}
Cole Comfort.
\newblock The algebra for stabilizer codes.
\newblock \href {https://arxiv.org/abs/2304.10584 [quant-ph]}
  {\path{arXiv:2304.10584 [quant-ph]}}, \href
  {https://doi.org/10.48550/arXiv.2304.10584}
  {\path{doi:10.48550/arXiv.2304.10584}}.

\bibitem{zxa}
Cole Comfort.
\newblock The {ZX\&}-calculus: A complete graphical calculus for classical
  circuits using spiders.
\newblock In Beno\^it Valiron, Shane Mansfield, Pablo Arrighi, and Prakash
  Panangaden, editors, {\em {\rm Proceedings 17th International Conference on}
  Quantum Physics and Logic, {\rm Paris, France, June 2 - 6, 2020}}, volume 340
  of {\em Electronic Proceedings in Theoretical Computer Science}, pages
  60--90. Open Publishing Association, 2021.
\newblock \href {https://doi.org/10.4204/EPTCS.340.4}
  {\path{doi:10.4204/EPTCS.340.4}}.

\bibitem{alr}
Cole Comfort and Aleks Kissinger.
\newblock A graphical calculus for lagrangian relations.
\newblock {\em Electronic Proceedings in Theoretical Computer Science},
  372:338–351, November 2022.
\newblock URL: \url{http://dx.doi.org/10.4204/EPTCS.372.24}, \href
  {https://doi.org/10.4204/eptcs.372.24} {\path{doi:10.4204/eptcs.372.24}}.

\bibitem{coya}
Brandon Coya.
\newblock {\em Circuits, bond graphs, and signal-flow diagrams: A categorical
  perspective}.
\newblock PhD thesis, University of California Riverside, 2018.
\newblock URL: \url{https://arxiv.org/pdf/1805.08290.pdf}.

\bibitem{lattice}
Niel de~Beaudrap and Dominic Horsman.
\newblock The zx calculus is a language for surface code lattice surgery.
\newblock {\em Quantum}, 4:218, January 2020.
\newblock URL: \url{http://dx.doi.org/10.22331/q-2020-01-09-218}, \href
  {https://doi.org/10.22331/q-2020-01-09-218}
  {\path{doi:10.22331/q-2020-01-09-218}}.

\bibitem{deFelice2023}
Giovanni de~Felice, Razin~A. Shaikh, Boldizsár Poór, Lia Yeh, Quanlong Wang,
  and Bob Coecke.
\newblock Light-matter interaction in the zxw calculus.
\newblock {\em Electronic Proceedings in Theoretical Computer Science},
  384:20–46, August 2023.
\newblock \href {https://doi.org/10.4204/eptcs.384.2}
  {\path{doi:10.4204/eptcs.384.2}}.

\bibitem{de_gliniasty_spin-optical_2024}
Grégoire de~Gliniasty, Paul Hilaire, Pierre-Emmanuel Emeriau, Stephen~C. Wein,
  Alexia Salavrakos, and Shane Mansfield.
\newblock A {Spin}-{Optical} {Quantum} {Computing} {Architecture}, January
  2024.
\newblock arXiv:2311.05605 [quant-ph].
\newblock URL: \url{http://arxiv.org/abs/2311.05605}, \href
  {https://doi.org/10.48550/arXiv.2311.05605}
  {\path{doi:10.48550/arXiv.2311.05605}}.

\bibitem{de2006symplectic}
Maurice~A De~Gosson.
\newblock {\em Symplectic geometry and quantum mechanics}, volume 166.
\newblock Springer Science \& Business Media, 2006.

\bibitem{de_gosson_symplectic_2011}
Maurice~A. de~Gosson.
\newblock {\em Symplectic {Methods} in {Harmonic} {Analysis} and in
  {Mathematical} {Physics}}.
\newblock Pseudo-{Differential} {Operators}. Springer, Basel, 2011.

\bibitem{duncan_graph-theoretic_2020}
Ross Duncan, Aleks Kissinger, Simon Perdrix, and John van~de Wetering.
\newblock Graph-theoretic simplification of quantum circuits with the
  {ZX}-calculus.
\newblock 4:279.
\newblock \href {https://doi.org/10.22331/q-2020-06-04-279}
  {\path{doi:10.22331/q-2020-06-04-279}}.

\bibitem{East2022}
Richard~D.P. East, John van~de Wetering, Nicholas Chancellor, and Adolfo~G.
  Grushin.
\newblock Aklt-states as zx-diagrams: Diagrammatic reasoning for quantum
  states.
\newblock {\em PRX Quantum}, 3(1), January 2022.
\newblock URL: \url{http://dx.doi.org/10.1103/PRXQuantum.3.010302}, \href
  {https://doi.org/10.1103/prxquantum.3.010302}
  {\path{doi:10.1103/prxquantum.3.010302}}.

\bibitem{erbele}
Jason~Michael Erbele.
\newblock {\em Categories in control: applied {PROP}s}.
\newblock PhD thesis, University of California Riverside, 2016.
\newblock URL: \url{https://arxiv.org/pdf/1611.07591.pdf}.

\bibitem{fabre_modes_2020}
Claude Fabre and Nicolas Treps.
\newblock Modes and states in quantum optics.
\newblock 92(3):035005.
\newblock URL: \url{http://arxiv.org/abs/1912.09321}, \href
  {https://arxiv.org/abs/1912.09321} {\path{arXiv:1912.09321}}, \href
  {https://doi.org/10.1103/RevModPhys.92.035005}
  {\path{doi:10.1103/RevModPhys.92.035005}}.

\bibitem{Fritz2020}
Tobias Fritz.
\newblock A synthetic approach to markov kernels, conditional independence and
  theorems on sufficient statistics.
\newblock {\em Advances in Mathematics}, 370:107239, August 2020.
\newblock URL: \url{http://dx.doi.org/10.1016/j.aim.2020.107239}, \href
  {https://doi.org/10.1016/j.aim.2020.107239}
  {\path{doi:10.1016/j.aim.2020.107239}}.

\bibitem{Gogioso2017}
Stefano Gogioso and Fabrizio Genovese.
\newblock Infinite-dimensional categorical quantum mechanics.
\newblock {\em Electronic Proceedings in Theoretical Computer Science},
  236:51–69, January 2017.
\newblock URL: \url{http://dx.doi.org/10.4204/EPTCS.236.4}, \href
  {https://doi.org/10.4204/eptcs.236.4} {\path{doi:10.4204/eptcs.236.4}}.

\bibitem{Gross2006}
D.~Gross.
\newblock Hudson’s theorem for finite-dimensional quantum systems.
\newblock {\em Journal of Mathematical Physics}, 47(12), December 2006.
\newblock URL: \url{http://dx.doi.org/10.1063/1.2393152}, \href
  {https://doi.org/10.1063/1.2393152} {\path{doi:10.1063/1.2393152}}.

\bibitem{gu_quantum_2009}
Mile Gu, Christian Weedbrook, Nicolas~C. Menicucci, Timothy~C. Ralph, and Peter
  van Loock.
\newblock Quantum {Computing} with {Continuous}-{Variable} {Clusters}.
\newblock {\em Physical Review A}, 79(6), June 2009.
\newblock arXiv: 0903.3233.
\newblock URL: \url{http://arxiv.org/abs/0903.3233}, \href
  {https://doi.org/10.1103/PhysRevA.79.062318}
  {\path{doi:10.1103/PhysRevA.79.062318}}.

\bibitem{guillemin_problems_1979}
Victor Guillemin and Shlomo Sternberg.
\newblock Some problems in integral geometry and some related problems in
  micro-local analysis.
\newblock 101(4):915.
\newblock \href {https://doi.org/10.2307/2373923} {\path{doi:10.2307/2373923}}.

\bibitem{zw}
Amar Hadzihasanovic.
\newblock A diagrammatic axiomatisation for qubit entanglement.
\newblock In {\em 2015 30th Annual ACM/IEEE Symposium on Logic in Computer
  Science (LICS)}, pages 573--584, Los Alamitos, CA, USA, jul 2015. IEEE
  Computer Society.
\newblock URL: \url{https://arxiv.org/pdf/1501.07082.pdf}, \href
  {https://doi.org/10.1109/LICS.2015.59} {\path{doi:10.1109/LICS.2015.59}}.

\bibitem{Heunen2019-ox}
Chris Heunen and Jamie Vicary.
\newblock {\em Categories for quantum theory}.
\newblock Oxford Graduate Texts in Mathematics. Oxford University Press,
  London, England, November 2019.

\bibitem{heurtel_complete_2024}
Nicolas Heurtel.
\newblock A complete graphical language for linear optical circuits with
  finite-photon-number sources and detectors.
\newblock URL: \url{http://arxiv.org/abs/2402.17693}, \href
  {https://arxiv.org/abs/2402.17693 [quant-ph]} {\path{arXiv:2402.17693
  [quant-ph]}}, \href {https://doi.org/10.48550/arXiv.2402.17693}
  {\path{doi:10.48550/arXiv.2402.17693}}.

\bibitem{Huang2023}
Jiaxin Huang, Sarah~Meng Li, Lia Yeh, Aleks Kissinger, Michele Mosca, and
  Michael Vasmer.
\newblock Graphical css code transformation using zx calculus.
\newblock {\em Electronic Proceedings in Theoretical Computer Science},
  384:1–19, August 2023.
\newblock URL: \url{http://dx.doi.org/10.4204/EPTCS.384.1}, \href
  {https://doi.org/10.4204/eptcs.384.1} {\path{doi:10.4204/eptcs.384.1}}.

\bibitem{Differentiation}
Emmanuel Jeandel, Simon Perdrix, and Margarita Veshchezerova.
\newblock {Addition and Differentiation of ZX-Diagrams}.
\newblock In Amy~P. Felty, editor, {\em 7th International Conference on Formal
  Structures for Computation and Deduction (FSCD 2022)}, volume 228 of {\em
  Leibniz International Proceedings in Informatics (LIPIcs)}, pages
  13:1--13:19, Dagstuhl, Germany, 2022. Schloss Dagstuhl -- Leibniz-Zentrum
  f{\"u}r Informatik.
\newblock URL:
  \url{https://drops.dagstuhl.de/entities/document/10.4230/LIPIcs.FSCD.2022.13},
  \href {https://doi.org/10.4230/LIPIcs.FSCD.2022.13}
  {\path{doi:10.4230/LIPIcs.FSCD.2022.13}}.

\bibitem{zxcompletea}
Emmanuel Jeandel, Simon Perdrix, and Renaud Vilmart.
\newblock {Completeness of the {ZX}-Calculus}.
\newblock {\em {Logical Methods in Computer Science}}, {Volume 16, Issue 2},
  June 2020.
\newblock \href {https://doi.org/10.23638/LMCS-16(2:11)2020}
  {\path{doi:10.23638/LMCS-16(2:11)2020}}.

\bibitem{kissinger_finite_2015}
Aleks Kissinger.
\newblock Finite matrices are complete for (dagger-)hypergraph categories,
  August 2015.
\newblock arXiv:1406.5942 [math].
\newblock URL: \url{http://arxiv.org/abs/1406.5942}, \href
  {https://doi.org/10.48550/arXiv.1406.5942}
  {\path{doi:10.48550/arXiv.1406.5942}}.

\bibitem{Kissinger2020}
Aleks Kissinger and John van~de Wetering.
\newblock Reducing the number of non-clifford gates in quantum circuits.
\newblock {\em Physical Review A}, 102(2), August 2020.
\newblock URL: \url{https://arxiv.org/pdf/1903.10477.pdf}, \href
  {https://doi.org/10.1103/physreva.102.022406}
  {\path{doi:10.1103/physreva.102.022406}}.

\bibitem{lawruk_special_1975}
B~Lawruk, J~Śniatycki, and W.M Tulczyjew.
\newblock Special symplectic spaces.
\newblock 17(2):477--497.
\newblock URL:
  \url{https://linkinghub.elsevier.com/retrieve/pii/0022039675900571}, \href
  {https://doi.org/10.1016/0022-0396(75)90057-1}
  {\path{doi:10.1016/0022-0396(75)90057-1}}.

\bibitem{lloyd_quantum_1999}
Seth Lloyd and Samuel~L. Braunstein.
\newblock Quantum computation over continuous variables.
\newblock {\em Physical Review Letters}, 82(8):1784--1787, February 1999.
\newblock arXiv: quant-ph/9810082.
\newblock URL: \url{http://arxiv.org/abs/quant-ph/9810082}, \href
  {https://doi.org/10.1103/PhysRevLett.82.1784}
  {\path{doi:10.1103/PhysRevLett.82.1784}}.

\bibitem{menicucci_graphical_2011}
Nicolas~C. Menicucci, Steven~T. Flammia, and Peter van Loock.
\newblock Graphical calculus for {Gaussian} pure states.
\newblock {\em Physical Review A}, 83(4), April 2011.
\newblock arXiv: 1007.0725.
\newblock URL: \url{http://arxiv.org/abs/1007.0725}, \href
  {https://doi.org/10.1103/PhysRevA.83.042335}
  {\path{doi:10.1103/PhysRevA.83.042335}}.

\bibitem{menicucci_universal_2006}
Nicolas~C. Menicucci, Peter van Loock, Mile Gu, Christian Weedbrook, Timothy~C.
  Ralph, and Michael~A. Nielsen.
\newblock Universal {Quantum} {Computation} with {Continuous}-{Variable}
  {Cluster} {States}.
\newblock {\em Physical Review Letters}, 97(11), September 2006.
\newblock arXiv: quant-ph/0605198.
\newblock URL: \url{http://arxiv.org/abs/quant-ph/0605198}, \href
  {https://doi.org/10.1103/PhysRevLett.97.110501}
  {\path{doi:10.1103/PhysRevLett.97.110501}}.

\bibitem{zxcompleteb}
Kang~Feng Ng and Quanlong Wang.
\newblock A universal completion of the {ZX}-calculus, 2017.
\newblock URL: \url{https://arxiv.org/pdf/1706.09877.pdf}.

\bibitem{ortiz-gutierrez_continuous_2017}
Luis Ortiz-Gutiérrez, Bruna Gabrielly, Luis~F. Muñoz, Kainã~T. Pereira,
  Jefferson~G. Filgueiras, and Alessandro~S. Villar.
\newblock Continuous variables quantum computation over the vibrational modes
  of a single trapped ion.
\newblock {\em Optics Communications}, 397:166--174, August 2017.
\newblock arXiv:1603.00065 [quant-ph].
\newblock URL: \url{http://arxiv.org/abs/1603.00065}, \href
  {https://doi.org/10.1016/j.optcom.2017.04.011}
  {\path{doi:10.1016/j.optcom.2017.04.011}}.

\bibitem{poor_qupit_2023}
Boldizsár Poór, Robert~I. Booth, Titouan Carette, John Van De~Wetering, and
  Lia Yeh.
\newblock The qupit stabiliser {ZX}-travaganza: Simplified axioms, normal forms
  and graph-theoretic simplification.
\newblock In {\em Electronic Proceedings in Theoretical Computer Science},
  volume 384, pages 220--264.
\newblock \href {https://doi.org/10.4204/EPTCS.384.13}
  {\path{doi:10.4204/EPTCS.384.13}}.

\bibitem{schwartz1957theorie}
Laurent Schwartz.
\newblock Th{\'e}orie des distributions {\`a} valeurs vectorielles. i.
\newblock In {\em Annales de l'institut Fourier}, volume~7, pages 1--141, 1957.

\bibitem{gaussrel}
Dario Stein and Richard Samuelson.
\newblock A category for unifying gaussian probability and nondeterminism.
\newblock Schloss Dagstuhl – Leibniz-Zentrum f\"{u}r Informatik, 2023.
\newblock URL:
  \url{https://drops.dagstuhl.de/entities/document/10.4230/LIPIcs.CALCO.2023.13},
  \href {https://doi.org/10.4230/LIPICS.CALCO.2023.13}
  {\path{doi:10.4230/LIPICS.CALCO.2023.13}}.

\bibitem{stein_graphical_2024}
Dario Stein, Fabio Zanasi, Richard Samuelson, and Robin Piedeleu.
\newblock Graphical quadratic algebra.
\newblock URL: \url{http://arxiv.org/abs/2403.02284}, \href
  {https://arxiv.org/abs/2403.02284 [cs, math]} {\path{arXiv:2403.02284 [cs,
  math]}}.

\bibitem{Toumi2021}
Alexis Toumi, Richie Yeung, and Giovanni de~Felice.
\newblock Diagrammatic differentiation for quantum machine learning.
\newblock {\em Electronic Proceedings in Theoretical Computer Science},
  343:132–144, September 2021.
\newblock URL: \url{http://dx.doi.org/10.4204/EPTCS.343.7}, \href
  {https://doi.org/10.4204/eptcs.343.7} {\path{doi:10.4204/eptcs.343.7}}.

\bibitem{tulczyjew_category_1984}
W.~M. Tulczyjew and S.~Zakrzewski.
\newblock The category of {F}resnel kernels.
\newblock 1(3):79--120.
\newblock URL:
  \url{https://www.sciencedirect.com/science/article/pii/0393044084900214},
  \href {https://doi.org/10.1016/0393-0440(84)90021-4}
  {\path{doi:10.1016/0393-0440(84)90021-4}}.

\bibitem{urbanski_structure_1985}
Pawel Urbanski.
\newblock The structure of positive linear symplectic relations.
\newblock 119.
\newblock URL: \url{https://www.fuw.edu.pl/~urbanski/Skany/StrPosRel.pdf}.

\bibitem{wetering}
John van~de Wetering.
\newblock {ZX}-calculus for the working quantum computer scientist, 2020.
\newblock URL: \url{https://arxiv.org/pdf/2012.13966.pdf}.

\bibitem{van_den_nest_graphical_2004}
Maarten Van~den Nest, Jeroen Dehaene, and Bart De~Moor.
\newblock Graphical description of the action of local {C}lifford
  transformations on graph states.
\newblock 69(2):022316.
\newblock Publisher: American Physical Society.
\newblock URL: \url{https://arxiv.org/pdf/quant-ph/0308151.pdf}, \href
  {https://doi.org/10.1103/PhysRevA.69.022316}
  {\path{doi:10.1103/PhysRevA.69.022316}}.

\bibitem{Von_Neumann2018-nl}
John Von~Neumann.
\newblock {\em Mathematical foundations of quantum mechanics}.
\newblock Princeton, NJ, February 2018.

\bibitem{qutrit}
Quanlong Wang.
\newblock Qutrit {ZX}-calculus is complete for stabilizer quantum mechanics.
\newblock In Bob Coecke and Aleks Kissinger, editors, {\em {\rm Proceedings
  14th International Conference on} Quantum Physics and Logic, {\rm Nijmegen,
  The Netherlands, 3-7 July 2017}}, volume 266 of {\em Electronic Proceedings
  in Theoretical Computer Science}, pages 58--70. Open Publishing Association,
  2018.
\newblock \href {https://doi.org/10.4204/EPTCS.266.3}
  {\path{doi:10.4204/EPTCS.266.3}}.

\bibitem{weedbrook_gaussian_2012}
Christian Weedbrook, Stefano Pirandola, Raul Garcia-Patron, Nicolas~J. Cerf,
  Timothy~C. Ralph, Jeffrey~H. Shapiro, and Seth Lloyd.
\newblock {G}aussian quantum information.
\newblock 84(2):621--669.
\newblock URL: \url{http://arxiv.org/abs/1110.3234}, \href
  {https://arxiv.org/abs/1110.3234} {\path{arXiv:1110.3234}}, \href
  {https://doi.org/10.1103/RevModPhys.84.621}
  {\path{doi:10.1103/RevModPhys.84.621}}.

\bibitem{weinstein_symplectic_2009}
Alan Weinstein.
\newblock Symplectic categories.
\newblock \href {https://arxiv.org/abs/0911.4133 [math]} {\path{arXiv:0911.4133
  [math]}}, \href {https://doi.org/10.48550/arXiv.0911.4133}
  {\path{doi:10.48550/arXiv.0911.4133}}.

\bibitem{dold_symplectic_1982}
Alan Weinstein.
\newblock The symplectic ``category''.
\newblock volume 905, pages 45--51. Springer Berlin Heidelberg.
\newblock Book Title: Differential Geometric Methods in Mathematical Physics
  Series Title: Lecture Notes in Mathematics.
\newblock URL: \url{http://link.springer.com/10.1007/BFb0092426}, \href
  {https://doi.org/10.1007/BFb0092426} {\path{doi:10.1007/BFb0092426}}.

\bibitem{Willems}
Jan~C. Willems.
\newblock Open stochastic systems.
\newblock {\em IEEE Transactions on Automatic Control}, 58(2):406--421, 2013.
\newblock \href {https://doi.org/10.1109/TAC.2012.2210836}
  {\path{doi:10.1109/TAC.2012.2210836}}.

\bibitem{yokoyama_ultra-large-scale_2013}
Shota Yokoyama, Ryuji Ukai, Seiji~C. Armstrong, Chanond Sornphiphatphong,
  Toshiyuki Kaji, Shigenari Suzuki, Jun-ichi Yoshikawa, Hidehiro Yonezawa,
  Nicolas~C. Menicucci, and Akira Furusawa.
\newblock Ultra-large-scale continuous-variable cluster states multiplexed in
  the time domain.
\newblock {\em Nature Photonics}, 7(12):982--986, December 2013.
\newblock URL: \url{http://www.nature.com/articles/nphoton.2013.287}, \href
  {https://doi.org/10.1038/nphoton.2013.287}
  {\path{doi:10.1038/nphoton.2013.287}}.

\bibitem{yoshikawa_generation_2016}
Jun-ichi Yoshikawa, Shota Yokoyama, Toshiyuki Kaji, Chanond Sornphiphatphong,
  Yu~Shiozawa, Kenzo Makino, and Akira Furusawa.
\newblock Generation of one-million-mode continuous-variable cluster state by
  unlimited time-domain multiplexing.
\newblock {\em APL Photonics}, 1(6):060801, September 2016.
\newblock arXiv: 1606.06688.
\newblock URL: \url{http://arxiv.org/abs/1606.06688}, \href
  {https://doi.org/10.1063/1.4962732} {\path{doi:10.1063/1.4962732}}.

\bibitem{zhou_quantum_2003}
D.~L. Zhou, B.~Zeng, Z.~Xu, and C.~P. Sun.
\newblock Quantum computation based on d-level cluster states.
\newblock 68(6):062303.
\newblock URL: \url{http://arxiv.org/abs/quant-ph/0304054}, \href
  {https://arxiv.org/abs/quant-ph/0304054} {\path{arXiv:quant-ph/0304054}},
  \href {https://doi.org/10.1103/PhysRevA.68.062303}
  {\path{doi:10.1103/PhysRevA.68.062303}}.

\end{thebibliography}

\appendix

\section{Axiom tables}

\subsection{Graphical affine algebra}
\label{app:gaa}

The following presentation of \(\AR\) is adapted from, Booth et al. \cite{gsa}, which is itself a modified version of the presentation given by Bonchi et al. \cite{gaa}.
\begin{definition}
\label{def:gaa}
  Let \(\GAA\) be the compact prop where the arrows are generated by  {\em grey and white spiders}, as well as {\em scalar multiplication} for every \(a \in \K\):

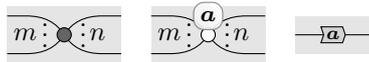
  \[
    \input{figures/AffRel/b-spider.tikz}\quad
    \input{figures/AffRel/w-spider.tikz}\quad
    \input{figures/AffRel/scalar.tikz}\quad
  \]
  Modulo the equations making the grey and white spiders flexsymmetric as well as the equations, for all \(a,b\in \K\) and \(c \in \K^*\):
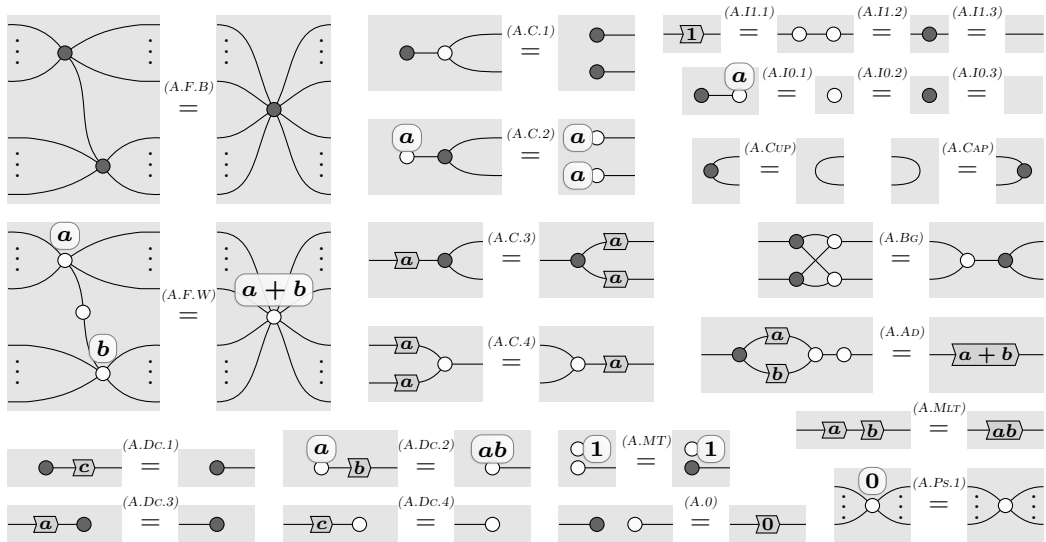
  \[\input{figures/AffRel/axioms.tikz}\]

  With the derived generators, for all \(a\in\K\):
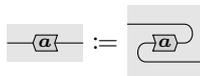
  \[
    \input{figures/AffRel/scalar-daggered.tikz}
    \coloneqq 
    \input{figures/AffRel/scalar-twisted.tikz}
  \]

\end{definition}

\section{Representing symplectomorphisms and rotations}

In this appendix, we review generators for the groups of symplectomorphisms, rotations and symplectic rotations over \(\R\).

\begin{definition}
Denote the group of symplectomorphisms on  \(\left(\K^{2n}, \omega_n \right)\) by  \(\Sp[n][\K]\). 
\end{definition}

\begin{lemma}
Elements of this group have the following matrix representation 
\[
  \Sp[n][\K]=
  \left\{
    \begin{bmatrix}A & B \\ C & D\end{bmatrix}\in \Matrices[2n][2n][\K]
    \ \middle|\
    \begin{matrix}
      A^\trans C = C^\trans A \\
      B^\trans C = C^\trans B \\
      I=A^\trans D - C^\trans B
    \end{matrix} 
  \right\}
\]
\end{lemma}

\begin{lemma}
The prop \(\bigcup_{n\in\N}\Sp[n][\K]\) is generated by the following matrices (with the corresponding matrices in \(\GSA\) drawn beside them) under the direct sum and composition, for all  \(A\in \Gl[n][\K]\) and \(B\in \Sym[n][\K]\):
\[
  \begin{bmatrix}
    A & 0 \\ 0 & A^{-\trans}
  \end{bmatrix} \leftrightarrow \input{figures/AffLagRel/diagsymp.tikz},\
  \begin{bmatrix}
    I & B \\ 0 & I
  \end{bmatrix}\leftrightarrow \input{figures/AffLagRel/xsymp.tikz},\
  \begin{bmatrix}
    I & 0 \\ B & I
  \end{bmatrix} \leftrightarrow \input{figures/AffLagRel/ysymp.tikz},\
  \begin{bmatrix}
    0 & I \\ -I & 0
  \end{bmatrix} \leftrightarrow \input{figures/AffLagRel/hsymp.tikz}
\]
\end{lemma}

\begin{definition}
A  matrix \(Q \in \Gl[n][\K]\) is \textbf{orthogonal} when \(Q^\trans Q = QQ^\trans = 1\). The orthogonal group \(\Orth[n][\K]\) consists of the subgroup of \(\Gl[n][\K]\) given by orthogonal matrices.
A \textbf{rotation} is an orthogonal matrix \(R\in \Orth[n][\K]\) with \(\det(R)=1\).
The \textbf{special orthogonal group} \(\SpOrth[n][\K]\) is the subgroup of \(\Orth[n][\K]\) given by rotations.
\end{definition}

\begin{lemma}
The prop of rotations,  \(\bigcup_{n\in\N} \SpOrth[n][\R]\), is generated rotation  matrices \({\mathcal R}(\theta)\) and the identity matrices on \(0\) and \(1\) under direct sum and composition, for all \(\theta\in [0,2\pi)\):
\[
{\mathcal R}(\theta)
\coloneqq
\begin{bmatrix}
\cos(\theta) & -\sin(\theta) \\
\sin(\theta)  & \cos(\theta) 
\end{bmatrix}
\quad\quad
\text{drawn in \(\GAA\) as}
\quad\quad
\input{figures/SympGauss/rotation_S.tikz}
\]
\end{lemma}

\begin{definition}
\textbf{The group of symplectic rotations} is the intersection \(\Sp[n][\R] \cap  \Orth[2n][\R]\). In other words, a symplectic rotation is a symplectomorphism of the form:

\[R\coloneqq \begin{bmatrix} C & -S \\ S & C \end{bmatrix}\]

The condition that \(\R\) is a symplectomorphism is equivalent to $A+i B$ being unitary; which induces a group isomorphism \(\Sp[n][\R] \cap \Orth[2n][\R] \cong \Unit[n][\C]\).
\end{definition}

\begin{lemma}
  The prop of symplectic rotations, \(\bigcup_{n\in\N} \Sp[n][\R] \cap
  \Orth[2n][\R]\), is generated by the following matrices under direct sum
  and composition, for all \(\theta \in [0,2\pi)\):
  \[
    {\mathcal R}(\theta)
    \coloneqq
    \begin{bmatrix}
    \cos(\theta) & -\sin(\theta) \\
    \sin(\theta)  & \cos(\theta) 
    \end{bmatrix}
    \qand
    {\mathcal S}(\theta)
    \coloneqq
    \begin{bmatrix}
    \cos(\theta) & -\sin(\theta) & 0 & 0 \\
    \sin(\theta) & \cos(\theta) & 0 & 0 \\
    0 & 0 & \cos(\theta) & -\sin(\theta) \\ 
    0 & 0 & \sin(\theta) & \cos(\theta) \\ 
    \end{bmatrix}
  \]
  The symplectic rotations \({\mathcal R}(\pi)\), \({\mathcal R}(\theta)\), and \({\mathcal S}(\vartheta)\) for \(\theta \in (-\pi,\pi)\) and \(\vartheta \in [0,2\pi)\) are respectively given in  \(\GSA[\R]\) by:
  \[
      \input{figures/SympGauss/rotation_antipode.tikz}
      \ ,\quad
    \input{figures/SympGauss/rotation_R.tikz}
    \qand
    \input{figures/SympGauss/rotation_Snew.tikz}
  \]
  Where  \({\mathcal R}(\pi)={\mathcal R}(\pi/2)^2\) so the antipode, which accounts for the asymptote of the tangent function is not needed as a generator.
\end{lemma}

\section{Proofs}
\label{sec:proofs}
\subsection{Proofs of section~\ref{sec:gauss}}

\begin{proof}[Proof of proposition~\ref{prop:positivechar}] 
First, remark that the definition of a Lagrangian subspace being positive is equivalent to the Hermitian form $H\coloneqq\frac{\text{-}i}{2}\Omega$ being positive semidefinite, where we recall that \(\Omega\) is the antisymmetric matrix induced by the symplectic form \(\omega(-,=)\).

We proceed by showing first $(1)\Rightarrow (2)$ then $(2)\Rightarrow (1)$ and finally $(1)\Leftrightarrow (3)$.
	In each case, we have have a Lagrangian subspace $\mathcal{L}\subseteq (\C^{2n}, \omega_n)$ with AP invariant \((E,\phi)\), with $r\coloneqq\rank(\Im(\phi))$.  For convenience, take \(J_r\coloneqq \diag(J_r, 0) \in \Matrices[n][n][\R]\).
	
	\begin{center}
		\fbox{$(1)\Rightarrow (2)$} 
	\end{center}
	
	 Suppose that  \(E\) is real and \(\Im(\varphi) \succeq 0\).  Then,
	$\mathcal{L}$ is real symplectomorphic to:
	
	\[\ker\begin{bmatrix}[cc|cc]
		I & L & \phi & 0 \\ 0 & 0 & -L^\trans & I
	\end{bmatrix} \]
	
	Where $L$ has real coefficients and $\Im(\phi)\succeq 0$. Then:
	
	\[\begin{bmatrix}
		I & -L \\ 0 & I
	\end{bmatrix}\begin{bmatrix}[cc|cc]
		I & L & \phi & 0 \\ 0 & 0 & -L^\trans & I
	\end{bmatrix}\begin{bmatrix}I & 0 & 0 & 0 \\ 0 & I & 0 & 0 \\ 0 & 0 & I & 0 \\ 0 & I & 0 & I
	\end{bmatrix}\begin{bmatrix}I & 0 & -\Re(\phi)-LL^\trans & L \\ 0 & I & L^\trans & -I \\ 0 & 0 & I & 0 \\ 0 & 0 & 0 & I
\end{bmatrix}= \begin{bmatrix}[cc|cc]
		I & 0 & i\Im(\phi) & 0 \\ 0 & I & 0 & 0
	\end{bmatrix}\]
	
	We get that $\mathcal{L}$ is real symplectomorphic to:
	
	\[\ker\begin{bmatrix}[cc|cc]
		I & 0 & i\Im(\phi) & 0 \\ 0 & I & 0 & 0
	\end{bmatrix}\]

	Since $\Im(\phi)\succeq 0$ there is an orthogonal matrix $O$ such that $O^\trans \Im(\phi)O=\Lambda$ where $\Lambda$ is a diagonal matrix with non-negative coefficients. We denote $\sqrt{\Lambda}$ the invertible diagonal matrix whose coefficients are the square root of those of $\Lambda$  if they are non-zero and else $1$. We have $\sqrt{\Lambda}^{-1} \Lambda \sqrt{\Lambda}^{-1} = J_r $. Then:
	
	\[\begin{bmatrix}
		O^\trans & 0 \\ 0 & I
	\end{bmatrix}\begin{bmatrix}[cc|cc]
		I & 0 & i\Im(\phi) & 0 \\ 0 & I & 0 & 0
	\end{bmatrix}\begin{bmatrix}O & 0 & 0 & 0\\ 0 & I & 0 & 0 \\ 0 & 0 & O & 0 \\ 0 & 0 & 0 & I 
	\end{bmatrix}= \begin{bmatrix}[cc|cc]
		I & 0 & i\Lambda & 0 \\ 0 & I & 0 & 0
	\end{bmatrix}\]
	
	and:
	
	\[\begin{bmatrix}
		\sqrt{\Lambda}^{-1} & 0 \\ 0 & I
	\end{bmatrix} \begin{bmatrix}[cc|cc]
		I & 0 & i\Lambda & 0 \\ 0 & I & 0 & 0
	\end{bmatrix}\begin{bmatrix} \sqrt{\Lambda} & 0 & 0 & 0 \\ 0 & I & 0 & 0 \\0 & 0 & \sqrt{\Lambda}^{-1} & 0 \\ 0 & 0 & 0 & 0 
	\end{bmatrix}=\begin{bmatrix}[cc|cc]
		I & 0 & iJ_r & 0 \\ 0 & I & 0 & 0
	\end{bmatrix}\]
	
	Finally, $\mathcal{L}$ is real symplectomorphic to $\ker\begin{bmatrix}[c|c]
		I & iJ_r
	\end{bmatrix}$.\\
	
	\begin{center}
		\fbox{$(2)\Rightarrow (1)$} 
	\end{center}

	Suppose that $\mathcal{L}$ is real symplectomorphic to \(\ker \begin{bmatrix}[c|c] I & iJ_r \end{bmatrix}\). By reducing to the AP-form we already know that $\mathcal{L}$ is real symplectomorphic to 
	
	\[\ker\begin{bmatrix}[cc|cc]
		I & L & \phi & 0 \\ 0 & 0 & -L^\trans & I
	\end{bmatrix}\]
	
	Furthermore we have:
	
	\[\begin{bmatrix}
		I & -L \\ 0 & I
	\end{bmatrix}\begin{bmatrix}[cc|cc]
		I & L & \phi & 0 \\ 0 & 0 & -L^\trans & I
	\end{bmatrix}\begin{bmatrix}I & 0 & 0 & 0 \\ 0 & I & 0 & 0 \\ 0 & 0 & I & 0 \\ 0 & I & 0 & I
	\end{bmatrix}\begin{bmatrix}I & 0 & \Re(-\phi-LL^\trans ) & \Re(L) \\ 0 & I & \Re (L^\trans )& -I \\ 0 & 0 & I & 0 \\ 0 & 0 & 0 & I
	\end{bmatrix}= \begin{bmatrix}[cc|cc]
	I & 0 & i\Im(\phi+LL^\trans) & -i\Im(L) \\ 0 & I & -i\Im(L^\trans) & 0
\end{bmatrix}\]
	
	Thus $\mathcal{L}$ is real symplectomorphic to:
	
	\[\ker\begin{bmatrix}[cc|cc]
		I & 0 & i\Im(\phi+LL^\trans) & -i\Im(L) \\ 0 & I & -i\Im(L^\trans) & 0
	\end{bmatrix}\]
	
	Now setting $\Phi \coloneqq \begin{bmatrix}
		\Im(\phi+LL^\trans ) & -\Im(L^\trans) \\ -\Im(L) & 0
	\end{bmatrix}$, we look at the conditions that $L$ and $\phi$ must satisfy for $\mathcal{L}=\ker\begin{bmatrix}[c|c] I & i\Phi \end{bmatrix}$ to be real symplectomorphic to  $\ker\begin{bmatrix}[c|c]  I & iJ_r
	\end{bmatrix}$. In other words, such that there exists real matrices $A,B,C,D,X$ and $Y$, with $X+iY$ invertible and $\begin{bmatrix}
		A & B \\ C & D
	\end{bmatrix}$ a symplectomorphism, such that:
	
	\[\begin{bmatrix}[c|c] I & iJ_r \end{bmatrix}\begin{bmatrix}
		A & B \\ C & D
	\end{bmatrix}=\left(X+iY\right) \begin{bmatrix}[c|c] I & i\Phi \end{bmatrix}\]
	
	This gives: 
	
	\[\begin{bmatrix}[c|c]
		A+ iJ_r C &   B+ iJ_r D
	\end{bmatrix} = \begin{bmatrix}[c|c]
		X +iY & -Y\Phi + i X\Phi \end{bmatrix}  \]
	
	and then by identifying the real and imaginary parts:
	
	\[\begin{cases}
		A = X \\
		J_r C = Y \\
		B = -Y\Phi \\
		J_r D = X\Phi
	\end{cases}\]
	
	we can elmininate $X$ and $Y$:
	
	\[\begin{cases}
		B = -J_r C \Phi \\
		J_r D  = A\Phi
	\end{cases}\]
	
	To this we can furthermore add the symplectic condition giving:
	
	\[\begin{cases}
		B = -J_r C \Phi \\
		J_r D  = A\Phi\\
		B^\trans D= D^\trans B \\
		A^\trans C = C^\trans A \\
		I = A^\trans D - C^\trans B
	\end{cases}\]
	
	Notice that the third equation is redundant as it follows from the others:
	
	\[B^\trans D = - \Phi C^\trans J_r D = - \Phi C^\trans A \Phi = -\Phi A^\trans C  \Phi= - D^\trans J_r C \Phi = - D^\trans B. \]
	
	We can now eliminate $B$ and get:
	
	\[\begin{cases}
		J_r D  = A\Phi\\
		A^\trans C = C^\trans A \\
		I = A^\trans D + C^\trans J_r C \Phi
	\end{cases}\]
	
	Multiplying the last equation by $\Phi $ on the left we get:
	
	\[ \Phi =  \Phi A^\trans D + \Phi C^\trans J_r C \Phi\]
	
	Using $ J_r D  = A\Phi $ gives: 
	
	\[ \Phi =   D^\trans J_r D + \Phi C^\trans J_r C \Phi \]
	
	We see that $\Phi$ is a sum of positive matrices, this implies that $\Phi \succeq 0 $. But as $\Phi= \begin{bmatrix}
		\Im(\phi+LL^\trans ) & -\Im(L^\trans) \\ -\Im(L) & 0
	\end{bmatrix}$, this implies that $\Im(L)=0$ and $\Im(\phi)\succeq 0 $. In other words $E$ has real coefficients and $\Im(L)=0$.\\
	
	\begin{center}
		\fbox{$(3)\Leftrightarrow (1)$} 
	\end{center}
	
	As $\mathcal{L}$ is Lagrangian, using $\mathcal{L}^{\omega}=\mathcal{L}$ and $\ker(B)^{\omega}=\Im(\Omega B^\trans )$ we get:
	
	\[\mathcal{L}=\ker \begin{bmatrix}[c|c]
		Z & X
	\end{bmatrix} = \Im \left(\Omega_n \begin{bmatrix}
		Z^\trans \\ X^\trans
	\end{bmatrix}\right)= \Im\begin{bmatrix}
		X^\trans \\ -Z^\trans
	\end{bmatrix} \]
	
	Thus the matrix representing the restriction $H_{|\mathcal{L}}$ of the Hermitian form $H$ to $\mathcal{L}$ is:
	
	\[\begin{bmatrix}
		\overline{X} & -\overline{Z}
	\end{bmatrix} \frac{\text{-}i}{2}\Omega_n \begin{bmatrix}
		X^\trans \\ -Z^\trans
	\end{bmatrix}=\frac{\text{-}i}{2}\left(\overline{Z}X^\trans - \overline{X}Z^\trans \right) \]
	
	We now take $\begin{bmatrix}[c|c]
		Z & X
	\end{bmatrix}=\begin{bmatrix}[cc|cc]
		I & L & \phi & 0 \\ 0 & 0 & -L^\trans & I
	\end{bmatrix}\begin{bmatrix}
		P & 0 \\ 0 & P
	\end{bmatrix}$. We notice that for any matrix, setting $B=A\begin{bmatrix}
		P & 0 \\ 0 & P
	\end{bmatrix}$, we have:
	
	\[ (\Omega_n B^\trans )^{\dagger} \frac{\text{-}i}{2}\Omega_n (\Omega_n B^\trans ) =\frac{\text{-}i}{2}\overline{B}\Omega_n B^\trans = \frac{\text{-}i}{2}\overline{A}\begin{bmatrix}
		P & 0 \\ 0 & P
	\end{bmatrix}\Omega_n \begin{bmatrix}
	P^\trans & 0 \\ 0 & P^\trans
\end{bmatrix} A^\trans = \frac{\text{-}i}{2}\overline{A}\Omega_n  A^\trans \]
	
	This gives the following matrix for the restricted Hermittian form:
	
	\begin{align*}
		H_{|\mathcal{L}}&= \frac{\text{-}i}{2} \left(\begin{bmatrix}
			I & \overline{L} \\ 0 & 0
		\end{bmatrix}\begin{bmatrix}
			\phi & -L \\ 0 & I
		\end{bmatrix}-\begin{bmatrix}
			\overline{\phi} & 0 \\ -\overline{L}^\trans & I
		\end{bmatrix}\begin{bmatrix}
			I & 0 \\ L^\trans & 0
		\end{bmatrix}\right) \\
		&=\begin{bmatrix}
			\Im(\phi) & -\Im(L)\\ -\Im(L^\trans) & 0
		\end{bmatrix}
	\end{align*}
	
	This matrix is positive if and only if $\Im(\phi) \succeq 0$ and $\Im(L)=0$ in other words if and only if $\Im(\phi) \succeq 0$ and $E$ has real coefficients.
\end{proof}

\begin{proof}[Proof of proposition~\ref{prop:quasirealchar}] 
We proceed in very similar fashion to the proof of propsition~\ref{prop:positivechar}; indeed we reuse many facts established therein.
As a matter of notation, let \(K\) be the symmetric matrix corresponding to the Hermitian form \(\chi(-,=)\); ie:

$$K\coloneqq \begin{bmatrix}
		0 & I \\ I & 0
	\end{bmatrix}$$
	
Moreover, take  \(J_r\coloneqq \diag(J_r, \vec 0)\) to be the diagonal matrix with \(r\) \(1\)'s followed by the appropriate number of \(0\)'s.

 Again we first show $(1)\Rightarrow (2)$ then $(2)\Rightarrow (1)$ and finally $(1)\Leftrightarrow (3)$.
 In each case, we have have a positive Lagrangian subspace $\mathcal{L}\subseteq (\C^{2n}, \omega_n)$ with AP invariant \((E,\phi)\), with $r\coloneqq\rank(\Im(\phi))$.  For convenience, take \(J_r\coloneqq \diag(J_r, 0) \in \Matrices[n][n][\R]\).
	
	\begin{center}
		\fbox{$(1)\Rightarrow (2)$} 
	\end{center}
	
	$\mathcal{L}$ is real diagonal symplectomorphic to:
	
	\[\ker\begin{bmatrix}[cc|cc]
		I & L & \phi & 0 \\ 0 & 0 & -L^\trans & I
	\end{bmatrix} \]
	
	Where $L$ has real coefficients, $\Im(\phi)\succeq 0$ and $\Re(\phi)=0$. Then:
	
	\[\begin{bmatrix}[cc|cc]
		I & L & \phi & 0 \\ 0 & 0 & -L^\trans & I
	\end{bmatrix}\begin{bmatrix}I & -L & 0 & 0\\ 0 & I & 0 & 0 \\ 0 & 0 & I & 0\\ 0 & 0 & L^\trans & I
	\end{bmatrix}= \begin{bmatrix}[cc|cc]
		I & 0 & i\Im(\phi) & 0 \\ 0 & 0 & 0 & I
	\end{bmatrix}\]	
	
	We reuse the same trick as for the positive case. Since $\Im(\phi)\succeq 0$ there is an orthogonal matrix $O$ such that $O^\trans \Im(\phi)O=\Lambda$ where $\Lambda$ is a diagonal matrix with non-negative coefficients. Then:
	
	\[\begin{bmatrix}
		O^\trans & 0 \\ 0 & I
	\end{bmatrix}\begin{bmatrix}[cc|cc]
		I & 0 & i\Im(\phi) & 0 \\ 0 & 0 & 0 & I
	\end{bmatrix}\begin{bmatrix}O & 0 & 0 & 0\\ 0 & I & 0 & 0 \\ 0 & 0 & O & 0 \\ 0 & 0 & 0 & I 
	\end{bmatrix}= \begin{bmatrix}[cc|cc]
		I & 0 & i\Lambda & 0 \\ 0 & 0 & 0 & I
	\end{bmatrix}\]
	
	We denote $\sqrt{\Lambda}$ the invertible diagonal matrix whose coefficients are the square root of those of $\Lambda$  if they are non-zero and else $1$. We have $\sqrt{\Lambda}^{-1} \Lambda \sqrt{\Lambda}^{-1} = J_r $. So:
	
	\[\begin{bmatrix}
		\sqrt{\Lambda}^{-1} & 0 \\ 0 & I
	\end{bmatrix} \begin{bmatrix}[cc|cc]
		I & 0 & i\Lambda & 0 \\ 0 & 0 & 0 & I
	\end{bmatrix}\begin{bmatrix} \sqrt{\Lambda} & 0 & 0 & 0\\ 0 & I & 0 & 0 \\ 0 & 0 & \sqrt{\Lambda}^{-1} & 0 \\ 0 & 0 & 0 & I
	\end{bmatrix}=\begin{bmatrix}[cc|cc]
		I & 0 & iJ_r & 0 \\ 0 & 0 & 0 & I
	\end{bmatrix}\]
	
	Finally, $\mathcal{L}$ is real diagonal symplectomorphic to $\ker\begin{bmatrix}[cc|cc]
		I & 0 & iJ_r & 0 \\ 0 & 0 & 0 & I
	\end{bmatrix}$.\\
	
	\begin{center}
		\fbox{$(2)\Rightarrow (1)$} 
	\end{center}
	
	First, $\mathcal{L}$ is real diagonal symplectomorphic to 
	
	\[\ker\begin{bmatrix}[cc|cc]
		I & L & \phi & 0 \\ 0 & 0 & -L^\trans & I
	\end{bmatrix}\]
	
	Using the same notation as before, denoting $O$ the orthogonal matrix diagonalizing $\Im(\phi)$ into a diagonal matrix $\Lambda$, we have:
	
	\[\begin{bmatrix}
		O^\trans & 0 \\ 0 & I
	\end{bmatrix}\begin{bmatrix}[cc|cc]
		I & L & \phi & 0 \\ 0 & 0 & -L^\trans & I
	\end{bmatrix}\begin{bmatrix}
		O & -L & 0 & 0 \\ 0 & I & 0 & 0 \\ 0 & 0 & O & 0 \\ 0 & 0 & L^{\trans}O & I
	\end{bmatrix}=\begin{bmatrix}[cc|cc]
		I & 0 & O^\trans \Re(\phi)O +i\Lambda & 0 \\ 0 & 0 & 0 & I
	\end{bmatrix}\]
	
	Denoting again $\sqrt{\Lambda}$ the invertible diagonal matrix whose coefficients are the square root of those of $\Lambda$  if they are non-zero and else $1$. We have $\sqrt{\Lambda}^{-1} \Lambda \sqrt{\Lambda}^{-1} = J_r $ and then:
	
	\[\begin{bmatrix}
		\sqrt{\Lambda}^{-1} & 0 \\ 0 & I
	\end{bmatrix}\begin{bmatrix}[cc|cc]
		I & 0 & O^\trans \Re(\phi)O +i\Lambda & 0 \\ 0 & 0 & 0 & I
	\end{bmatrix}\mathcal{D}\begin{bmatrix}
		\sqrt{\Lambda} & 0 & 0 & 0\\ 0 & I & 0 & 0 \\ 0 & 0 & \sqrt{\Lambda}^{-1} & 0 \\ 0 & 0 & 0 & I
	\end{bmatrix}=\begin{bmatrix}[cc|cc]
		I & 0 & \alpha +iJ_r & 0 \\ 0 & 0 & 0 & I
	\end{bmatrix}\]
	
	with $\alpha\coloneqq \sqrt{\Lambda}^{-1}O^\trans \Re(\phi)O\sqrt{\Lambda}^{-1}$. We can now look for the necessary conditions on $\Re(\phi)$ for $\ker \begin{bmatrix}[cc|cc]
		I & 0 & \alpha + iJ_r & 0 \\ 0 & 0 & 0 & I
	\end{bmatrix}$ to be real diagonal symplectomorphic to $\ker \begin{bmatrix}[cc|cc]
		I & 0 & iJ_r & 0 \\ 0 & 0 & 0 & I
	\end{bmatrix}$. In other words, such that there exists real matrices $A_1$, $A_2$, $A_3$, $A_4$, $D_1$, $D_2$, $D_3$, $D_4$ and complex matrices $W_1$, $W_2$, $W_3$, $W_4$ with $\begin{bmatrix}
		W_1 & W_2 \\ W_3 & W_4 
	\end{bmatrix}$ invertible and $\begin{bmatrix}
		D_1 & D_2 \\ D_3 & D_4 
	\end{bmatrix}$ being the transpose of the inverse of $\begin{bmatrix}
		A_1 & A_2 \\ A_3 & A_4 
	\end{bmatrix}$, satisfying:
	
	\[\begin{bmatrix}[cc|cc]
		I & 0 & \alpha + iJ_r & 0 \\ 0 & 0 & 0 & I
	\end{bmatrix}\begin{bmatrix}
		A_1 & A_2 & 0 & 0 \\ A_3 & A_4 & 0 & 0 \\ 0 & 0 & D_1 & D_2 \\ 0 & 0 & D_3 & D_4 
	\end{bmatrix}=\begin{bmatrix}
		W_1 & W_2 \\ W_3 & W_4 
	\end{bmatrix} \begin{bmatrix}[cc|cc] I & 0 & iJ_r & 0 \\ 0 & 0 & 0 & I \end{bmatrix}\]
	
	This gives: 
	
	\[\begin{bmatrix}[cc|cc]
		A_1 & A_2 & \alpha D_1 + iJ_r D_1 & \alpha D_2 + iJ_r D_2 \\ 0 & 0 & D_3 & D_4
	\end{bmatrix} = \begin{bmatrix}[cc|cc]
		W_1 & 0 & iW_1 J_r & W_2 \\ W_3 & 0 & iW_3 J_r & W_4
	\end{bmatrix}\]
	
	We directly get $W_1 = A_1 $, $W_2 = \alpha D_2 + iJ_r D_2 $, $W_3 = 0$ and $W_4 = D_4 $, and then $A_2 = 0$ and $D_3 = 0$. From the constraints it follows that $D_1 = A_{1}^{-\trans} $, $D_2 = -A_{1}^{-\trans}A_{3}^\trans A_{4}^{-\trans}$ and $D_4 = A_{4}^{-\trans}$.
	
	Then by identifying the real and imaginary parts gives:
	
	\[\begin{cases}
		\alpha A_{1}^{-\trans} = 0 \\
		J_r A_{1}^{-\trans} = A_{1} J_r
	\end{cases}\]
	
	It implies that $\alpha= 0$ and then $\Re(\phi)=0 $, as expected.
	
	\begin{center}
		\fbox{$(3)\Leftrightarrow (1)$} 
	\end{center}
	
	As $\mathcal{L}$ is Lagrangian, using $\mathcal{L}^{\omega}=\mathcal{L}$ and $\ker(B)^{\omega}=\Im(\Omega B^\trans )$ we get:
	
	\[\mathcal{L}=\ker \begin{bmatrix}[c|c]
		Z & X
	\end{bmatrix} = \Im \left(\Omega_n \begin{bmatrix}
		Z^\trans \\ X^\trans
	\end{bmatrix}\right)= \Im\begin{bmatrix}
		X^\trans \\ -Z^\trans
	\end{bmatrix} \]
	
	Thus the matrix representing the restriction $K_{|\mathcal{L}}$ of the Hermitian form $K$ to $\mathcal{L}$ is:
	
	\[\begin{bmatrix}
		\overline{X} & -\overline{Z}
	\end{bmatrix} K \begin{bmatrix}
		X^\trans \\ -Z^\trans
	\end{bmatrix}=-\overline{Z}X^\trans - \overline{X}Z^\trans \]
	
	We now take $\begin{bmatrix}
		Z & X
	\end{bmatrix}=\begin{bmatrix}[cc|cc]
		I & L & \phi & 0 \\ 0 & 0 & -L^\trans & I
	\end{bmatrix}\begin{bmatrix}
		P & 0 \\ 0 & P
	\end{bmatrix}$. We notice that for any matrix, setting $B=A\begin{bmatrix}
	P & 0 \\ 0 & P
\end{bmatrix}$, and using $K \Omega_n = -\Omega_n K =\begin{bmatrix}
		-I & 0 \\ 0 & I \end{bmatrix} $ we have:
	
	\[(\Omega_n B^\trans )^{\dagger} K (\Omega_n B^\trans )=\overline{A}\begin{bmatrix}
		P & 0 \\ 0 & P
	\end{bmatrix}\Omega_n K \Omega_n \begin{bmatrix}
	P^\trans & 0 \\ 0 & P^\trans
\end{bmatrix} A^\trans = \overline{A}K A^\trans \]
	
	This gives the following matrix for the restricted Hermittian form:
	
	\begin{align*}
		K_{|\mathcal{L}}&=- \left(\begin{bmatrix}
			I & \overline{L} \\ 0 & 0
		\end{bmatrix}\begin{bmatrix}
			\phi & -L \\ 0 & I
		\end{bmatrix}+\begin{bmatrix}
			\overline{\phi} & 0 \\ -\overline{L}^\trans & I
		\end{bmatrix}\begin{bmatrix}
			I & 0 \\ L^\trans & 0
		\end{bmatrix}\right) \\
		&=\begin{bmatrix}
			2\Re(\phi) & +2i\Im(L)\\ -2i\Im(L^\trans) & 0
		\end{bmatrix}
	\end{align*}
	
	As we already know that $\Im(L)=0$, this matrix is $0$ if and only if $\Re(\phi)= 0$.
\end{proof}

\begin{proof}[Proof of lemma~\ref{lem:closure}]
  First, we prove that positive affine Lagrangian relations are closed under composition.  Take two such relations \((S+\vec s):n\to m\) and \((R+\vec r):m\to k\). If the composite is empty, then the claim follows vacuously. Otherwise, take some \([\vec v_I^\trans\ \vec v_O^\trans]^\trans \in (R+\vec r)\circ (S+\vec s)\).
  Then there exists some \(\vec w \in \C^{2m}\) such that \( [\vec v_I^\trans\ \vec w^\trans]^\trans \in (S+\vec s)\) and \([\vec w^\trans\ \vec v_O^\trans]\in (R+\vec r)\).  Therefore:
  \begin{align}
    i\omega_{n,k}\left(\overline{\begin{bmatrix} \vec v_I \\ \vec v_O \end{bmatrix}},  \begin{bmatrix} \vec v_I \\ \vec v_O \end{bmatrix}\right)
      &=  i\omega_{n,k}\left(\begin{bmatrix} \overline{\vec v_I} \\ \overline{\vec v_O} \end{bmatrix},  \begin{bmatrix} \vec v_I \\ \vec v_O \end{bmatrix}\right)\\
      &=  i\omega_{k}\left( \overline{\vec v_O},\vec v_O \right)-   i\omega_{n}\left( \overline{\vec v_I},\vec v_I \right)\\
      &=  i\omega_{k}\left( \overline{\vec v_O},\vec v_O \right)-   i\omega_{k}\left( \overline{\vec w},\vec w \right) +  i\omega_{k}\left( \overline{\vec w},w \right)- i\omega_{n}\left( \overline{\vec v_I},\vec v_I \right)\\
      &=  i\omega_{n,m}\left(\begin{bmatrix} \overline{\vec v_I} \\ \overline{\vec w} \end{bmatrix},  \begin{bmatrix} \vec v_I \\ \vec w \end{bmatrix}\right)
         + i\omega_{m,k}\left(\begin{bmatrix} \overline{\vec w} \\ \overline{\vec v_O} \end{bmatrix},  \begin{bmatrix} \vec w \\ \vec v_O \end{bmatrix}\right) \\
      &=  i\omega_{n,m}\left(\overline{\begin{bmatrix} \vec v_I \\ \vec w \end{bmatrix}},  \begin{bmatrix} \vec v_I \\ \vec w \end{bmatrix}\right)
         + i\omega_{m,k}\left(\overline{\begin{bmatrix} \vec w \\ \vec v_O \end{bmatrix}},  \begin{bmatrix} \vec w \\ \vec v_O \end{bmatrix}\right) \geq 0
  \end{align}

  Suppose moreover, that \((S+\vec s):n\to m\) and \((R+\vec r):m\to k\) are quasi-real.  Again if the composite is empty, it is vacuously quasi-real.  Otherwise, take some
  \[[\vec u_I^\trans\ \vec u_O^\trans]^\trans,
    [\vec v_I^\trans\ \vec v_O^\trans]^\trans \in (R+\vec r)\circ (S+\vec s)
  \]
  Then there exist 
  \(\vec a, \vec b\in \C^{2m}\)
  such that  
  \( [\vec v_I^\trans\ \vec a^\trans]^\trans,  [\vec u_I^\trans\ \vec b^\trans]^\trans \in (S+\vec s),
  [\vec a^\trans \ \vec v_O^\trans]^\trans, [\vec b^\trans \ \vec u_O^\trans]^\trans  \in (R+\vec r)\).
%
  Therefore, 
  \begin{align}
    \chi_{k}&(\vec u_O, \vec v_O)-\chi_{n}(\vec u_I, \vec v_I)\\
    =&\chi_{k}(\vec u_O, \vec v_O)-\chi_{n}(\vec u_I, \vec v_I)+0+0\\
    =&\chi_{k}(\vec u_O, \vec v_O)-\chi_{n}(\vec u_I, \vec v_I)
    -\chi_{m}(\vec b, \vec a)+\chi_{n}(\vec u_I, \vec v_I)
    -\chi_{k}(\vec u_O, \vec v_O)+\chi_{m}(\vec b, \vec a)\\
    =&\chi_{k}(\vec u_O, \vec v_O)-\chi_{k}(\vec u_O, \vec v_O)+\chi_{n}(\vec u_I, \vec v_I)-\chi_{n}(\vec u_I, \vec v_I)+\chi_{m}(\vec b, \vec a)-\chi_{m}(\vec b, \vec a)\\
    =&0
  \end{align}

  Now we prove that the cup quasi-real. Take any \((\bullet,[z\ z\ x\ -x])\in \eta_1\).  Then
  \begin{align}
    \omega_{0,1}((\bullet,[z\ z\ x\ -x]^\trans), (\bullet,[z\ z\ x\ -x]^\trans))
    &=i\omega_1(\overline{[z\ z\ x\ -x]}^\trans, [z\ z\ x\ -x]^\trans)\\
    &=i\omega_1([\overline z\ \overline z\ \overline x\ -\overline x]^\trans, [z\ z\ x\ -x]^\trans)\\
    &=i( \overline z (x) +\overline z (-x) + \overline x (z) - \overline x (z))\\
    &=0 \\
  \end{align}
  So it is positive.  Take another \((\bullet,[z'\ z'\ x'\ -x']\in \eta_1)\), then
  \begin{align}
    \chi_1([z\ z\ x\ -x]^\trans,[z'\ z'\ x'\ -x']^\trans)
    =\overline z(x') +\overline z(-x')  +\overline x(z')-\overline x(z')
    =0
  \end{align}
  so it is quasi real.  The cap and identity are also quasi-real by the same argument.  The  symmetry are also clearly quasi-real, and the tensor product clearly preserves both properties.
\end{proof}

\begin{lemma}\label{lm:lindec}
  Any diagram in \(\GAA[\R]\) can be rewritten as: \[\input{Gauss/lindec.tikz}\]
  
  where \(S\) and \(T\) are invertible.  
\end{lemma}

\begin{proof}[Proof of lemma~\ref{lm:lindec}]
  We will work with bipartite states. First any affine subspace can be written as the image of a matrix \(\begin{bmatrix}
    A\\ B
  \end{bmatrix}\) translated by a vector \(\begin{bmatrix}
  x \\ y
\end{bmatrix}\) giving a diagram:
  
  \[\input{Gauss/lindec1.tikz}\]
  
  Then we know that  we can act on \(A\) with an invertible matrix \(P\) and on \(B\) with an invertible matrix \(Q\) such that
  
   \[
     A
     =
     P
     \begin{bmatrix}
      0 \\ X
    \end{bmatrix}
    \quad\text{and}
    \quad B= Q\begin{bmatrix}
      Y \\ 0
    \end{bmatrix}
  \]
  
  Giving a diagram:
  
  \[\input{Gauss/lindec2.tikz}\]
  
  Now we decompose the domain  as \(\ker(X)\cap \ker(Y) \oplus E\)  and then as \(\ker(X)\cap \ker(Y) \oplus E\cap \ker(Y) \oplus S \oplus E\cap \ker(Y) \).  Denoting \(K\) an invertible matrix doing a change of basis adapted to this decomposition we get:
  
  \[
    \begin{bmatrix}[c|c]
      X \\ Y  
    \end{bmatrix}
    K
    =
    \begin{bmatrix}[cc|cc]
      0 & X_0 & X_1 & 0 \\ 0 & 0 & Y_0 & Y_1
    \end{bmatrix}
  \]
  
  Diagramatically:
  
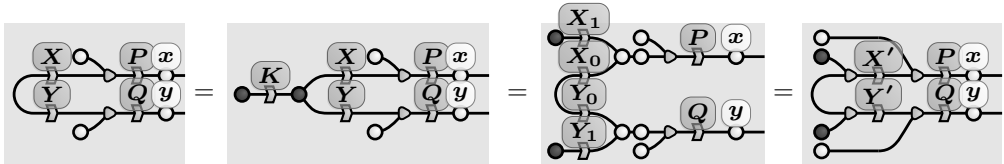
  \[\input{Gauss/lindec3.tikz}\]
  
  Then one can notice that both  \(X' = \begin{bmatrix} X_0 & X_1 \end{bmatrix}\) and \(Y' =\begin{bmatrix} Y_0 & Y_1 \end{bmatrix}\) are invertible.
\end{proof}

\begin{lemma}\label{lm:gauss_to_NF}
	Any diagram in \(\GGA\) reduces to the pseudo normal form:
	\input{Gauss/NForm.tikz}
\end{lemma}
\begin{proof}
  
  We start by moving all vacuums  to the left giving a diagram in the form:
  
  \[\input{Gauss/NForm1.tikz}\]
  
  where \(L\) is a diagram in \(\GAA[\R]\). Applying lemma~\ref{lm:lindec} we get:
  
  \[\input{Gauss/NForm2.tikz}\]
  
  Thus without loss of generality we restrict to reducing in pseudo normal form diagrams of the form:
  
  \[\input{Gauss/NForm3.tikz}\]
  
  We first apply singular value decomposition to \(X\) giving:
  
  \[\input{Gauss/NForm4.tikz}\]
  
  Then applying singular value decomposition to \(Z\) gives:
  
  \[\input{Gauss/NForm5.tikz}\]
\end{proof}

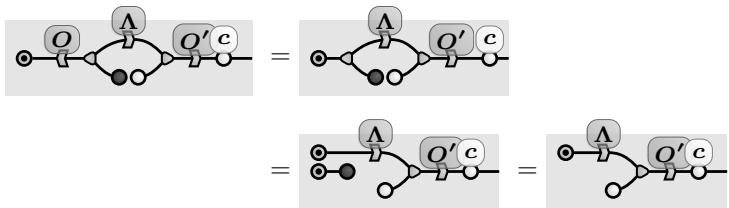
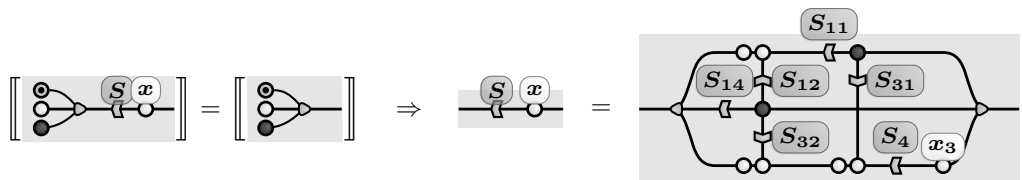
\begin{lemma}\label{lm:gauss_uniq}
	\[\interp{\input{Gauss/uniqueness01.tikz}}=\interp{\input{Gauss/uniqueness02.tikz}} \quad \Rightarrow \quad \input{Gauss/uniqueness.tikz}\]
	Where $S_{11}$ is orthogonal, and $S_{14}$ and $S_4 $ are invertible, furthermore the number of vacuum state generators must be the same on both sides.
\end{lemma}
\begin{proof}
  In the semantics, the equality:
  
  \[\interp{\input{Gauss/uniqueness01.tikz}}=\interp{\input{Gauss/uniqueness02.tikz}}\]
  
   translates to:
  
  \[\begin{bmatrix}
    x-a \\ -b
  \end{bmatrix} \in \ker \begin{bmatrix}[cc|cc]
    I & 0 & iJ_r & 0 \\ 0 & 0 & 0 & I
  \end{bmatrix}\begin{bmatrix}
    S & 0 \\ 0 & S^{-\trans}
  \end{bmatrix} \Leftrightarrow \begin{bmatrix}
    a \\ b
  \end{bmatrix} \in\ker \begin{bmatrix}[cc|cc]
    I & 0 & iI_k & 0 \\ 0 & 0 & 0 & I
  \end{bmatrix}\]
  
  This implies that that affine subspaces on the left contains $0$, thus the affine shift on the left must be in the linear space on the right: 
  \[
    \begin{bmatrix}
      x \\ 0
    \end{bmatrix}
    \in \ker
    \begin{bmatrix}[cc|cc]
      I & 0 & iJ_k & 0 \\ 0 & 0 & 0 & I
    \end{bmatrix}
    \quad\text{so that}\quad
    x=
    \begin{bmatrix}
      0 \\ 0 \\  x_3
    \end{bmatrix}
  \]
  
  Now that the affine shift has been taken care of we have:
  
  \[\ker \begin{bmatrix}[cc|cc]
    I & 0 & iJ_r & 0 \\ 0 & 0 & 0 & I
  \end{bmatrix}\begin{bmatrix}
    S & 0 \\ 0 & S^{-\trans}
  \end{bmatrix} = \ker \begin{bmatrix}[cc|cc]
    I & 0 & iI_k & 0 \\ 0 & 0 & 0 & I
  \end{bmatrix}\]
  
  In other words, there is an invertible complex matrix $W$ such that:
  
  \[\begin{bmatrix}[cc|cc]
    I & 0 & iJ_r & 0 \\ 0 & 0 & 0 & I
  \end{bmatrix}\begin{bmatrix}
    S & 0 \\ 0 & S^{-\trans}
  \end{bmatrix} = W \begin{bmatrix}[cc|cc]
    I & 0 & iJ_k & 0 \\ 0 & 0 & 0 & I
  \end{bmatrix}\] 
  
  Setting:
  \[W=\begin{bmatrix}
    W_1 & W_2 \\ W_3 & W_4 
  \end{bmatrix}$, $S=\begin{bmatrix}
    S_1 & S_2 \\ S_3 & S_4 
  \end{bmatrix}$ and $S^{-\trans}=\begin{bmatrix}
    (S^{-\trans})_1 & (S^{-\trans})_2 \\ (S^{-\trans})_3 & (S^{-\trans})_4 
  \end{bmatrix}\]
  
  we get:
  
  \[\begin{bmatrix}[cc|cc]
    I & 0 & iJ_r & 0 \\ 0 & 0 & 0 & I
  \end{bmatrix}\begin{bmatrix}
    S_1 & S_2 & 0 & 0 \\ S_3 & S_4 & 0 & 0 \\ 0 & 0 & (S^{-\trans})_1 & (S^{-\trans})_2 \\ 0 & 0 & (S^{-\trans})_3 & (S^{-\trans})_4 
  \end{bmatrix}=\begin{bmatrix}
    W_1 & W_2 \\ W_3 & W_4 
  \end{bmatrix} \begin{bmatrix}[cc|cc] I & 0 & iJ_k & 0 \\ 0 & 0 & 0 & I \end{bmatrix}\]
  
  This gives: 
  
  \[\begin{bmatrix}[cc|cc]
    S_1 & S_2 & iJ_r (S^{-\trans})_1 & iJ_r (S^{-\trans})_2 \\ 0 & 0 & (S^{-\trans})_3 & (S^{-\trans})_4
  \end{bmatrix} = \begin{bmatrix}[cc|cc]
    W_1 & 0 & iW_1 J_k & W_2 \\ W_3 & 0 & iW_3 J_k & W_4
  \end{bmatrix}\]
  
  We directly get \(W_1 = S_1\), \(W_2 = iJ_k (S^{-\trans})_2\), \(W_3 = 0\) and \(W_4 = (S^{-\trans})_4\). We also have \(S_2 = 0\), \(J_r (S^{-\trans})_1 = S_1 J_k\) and \((S^{-\trans})_3 = 0\). From the constraints it follows that \((S^{-\trans})_1 = (S_1 )^{-\trans}\), \((S^{-\trans})_4 = (S_4 )^{-\trans}\) and \((S^{-\trans})_2 = -S_{1}^{-\trans}S_{3}^\trans S_{4}^{-\trans}\). Therefore, \(r=k\) and:
  
  \[S=\begin{bmatrix}
    S_1 & 0 \\ S_3 & S_4
  \end{bmatrix}\]
   with only constraint that \(J_r = S_1 J_r (S_1 )^\trans\).
  
  We write: 
  \[S_1 = \begin{bmatrix}
    S_{11} & S_{12} \\ S_{13} & S_{14}
  \end{bmatrix}\]
  The constraint then becomes:
  
  \[\begin{bmatrix}
    I & 0 \\ 0 & 0
  \end{bmatrix}= \begin{bmatrix}
  S_{11} & S_{12} \\ S_{13} & S_{14}
  \end{bmatrix}\begin{bmatrix}
  I & 0 \\ 0 & 0
  \end{bmatrix}\begin{bmatrix}
  (S_{11})^\trans & (S_{13})^\trans \\ (S_{12})^\trans & (S_{14})^\trans
  \end{bmatrix}=\begin{bmatrix}
  S_{11}(S_{11})^\trans & S_{11}(S_{13})^\trans \\ S_{13}(S_{11})^\trans & S_{13} (S_{13})^\trans
  \end{bmatrix}\]
  
  So that \(S_{13} = 0\).  It follows that:
  
  \[S=\begin{bmatrix}
    S_{11} & S_{12} & 0 \\ 0 & S_{14} & 0 \\ S_{31} & S_{32} & S_4
  \end{bmatrix}\]
  
  with \(S_{11}\) orthogonal, and \(S_{14}\) and \(S_4\) invertible.
  
  Then in \(\GAA[\R]\) we have:
  
  \input{Gauss/uniqueness.tikz}

\end{proof}

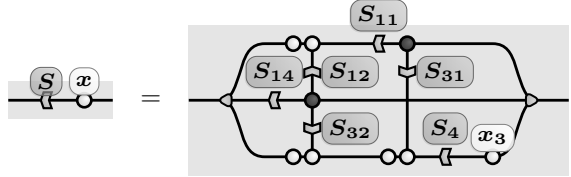
\begin{proof}[Proof of theorem~\ref{thm:gauss_completeness}]
  
  We define an interpretation prop morphism $\interp{\_}:\GGA \to \GaussRel$ by sending all generator from $\GGA$ to the corresponding Lagrangian subspaces and vacuum generators to the subspace corresponding to a white node with phase $(0,i)$.
  
  By lemma~\ref{lm:gauss_to_NF} we know that any diagram in $\GGA$ can be put in pseudo normal form, but by proposition~\ref{prop:quasirealchar}, we see that replacing the vacuum generator by a white node with phase $(0,i)$, we can represent all quasi-real subspaces. In other words, the interpretation of $\GGA $ into quasi-real spaces is surjective. It only remains to prove that this interpretation is also injective.

  Let \(X\) and \(Y\) be two diagrams in $\Gauss $ such that $\interp{X }=\interp{Y }$. Using lemma~\ref{lm:gauss_to_NF} we can assume that both diagrams are in pseudo normal form. We have:
  
  \[\interp{\input{Gauss/comp1.tikz}}=\interp{\input{Gauss/comp2.tikz}}\]
  
  Composing by the same the two diagrams on both sides we get:
  
  \[\interp{\input{Gauss/comp3.tikz}}=\interp{\input{Gauss/comp4.tikz}}\]
  
  Then by lemma~\ref{lm:gauss_uniq}: 
  
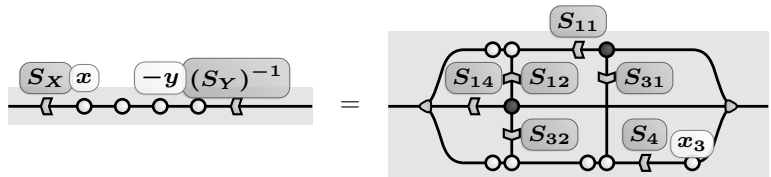
  \[\input{Gauss/uniqueness04.tikz}\]

  with $S_{11}$ orthogonal, and $S_{14}$ and $S_4 $ invertible.
  Therefore:
  
  \[\input{Gauss/comp5.tikz}\]
\end{proof}

\subsection{Proofs of section~\ref{sec:sguass}}

\begin{proof}[Proof of proposition~\ref{prop:gauss_wigner}]
  
  We start  by computing the Wigner function from the wave-function:
  
  \begin{align}
    W(q,p)=& \frac{1}{\pi^n }\int_{\mathbb{R}^n } \overline{\phi}(q+\xi)\phi(q -\xi )\exp(2i p^\trans \xi)\, d\xi \\
    =&\frac{1}{\pi^n }\sqrt{\frac{\det(\Im(\Phi))}{\pi^{n} }}\int_{\R^n }\\& \exp(-\frac{i}{2} (q+\xi)^\trans \overline{\Phi} (q+\xi))\exp(\frac{i}{2} (q-\xi)^\trans \Phi (q-\xi))\exp(2i p^\trans \xi)\, d\xi 
  \end{align}
  
  We first compute the exponent: 
  
  \begin{align}
    &-\frac{i}{2} (q+\xi)^\trans \overline{\Phi} (q+\xi) 
    + \frac{i}{2} (q-\xi)^\trans \Phi (q-\xi)+ 2i\pi p^\trans \xi \\
    &= \frac{i}{2} \left( -q^\trans \overline{\phi} q - 2q^\trans \overline{\Phi} \xi - \xi^\trans \overline{\Phi} \xi +q^\trans \Phi q - 2q^\trans \Phi \xi + \xi^\trans \Phi \xi + 4 p^\trans \xi \right)\\
    &= \frac{i}{2} \left( 2iq^\trans \Im(\Phi) q - 4q^\trans \Re(\Phi) \xi +2i\xi^\trans \Im(\Phi) \xi + 4 p^\trans \xi \right)\\
    &= - q^\trans \Im (\Phi) q -2 i q^\trans \Re(\Phi) \xi - \xi^\trans \Im (\Phi) \xi + 2i p^\trans \xi \\
    &= - q^\trans \Im (\Phi) q - \xi^\trans \Im (\Phi) \xi +2i\left(p^\trans-q^\trans \Re(\Phi)\right) \xi
  \end{align}
  
  Back to the integral:
  
  \begin{align}
    W(q,p)=& \frac{1}{\pi^n }\sqrt{\frac{\det(\Im(\Phi))}{\pi^{n} }}\exp(- q^\trans \Im (\Phi) q)\\
    & \int_{\mathbb{R}^n } \exp(- \xi^\trans \Im (\Phi) \xi +2i\left(p^\trans-q^\trans \Re(\Phi)\right) \xi)\, d\xi \\
    =& \frac{1}{\pi^n } \exp(- q^\trans \Im (\Phi) q- \left(p^\trans-q^\trans \Re(\Phi)\right) \Im(\Phi)^{-1} \left(p- \Re(\Phi)q \right))
  \end{align}
  
  Here we used the identity:

  \[\int_{\mathbb{R}^n } \exp(- \xi^\trans A \xi +iJ^\trans \xi)\, d\xi = \sqrt{\frac{\pi^n }{\det(A)}} \exp(-\frac{1}{4} J^\trans A^{-1} J)\]
  We can rewrite the exponent  as: 
  
  \begin{align}
    &=  - q^\trans \Im (\Phi) q - \left(p^\trans - q^\trans \Re(\Phi) \right) \Im(\Phi)^{-1} \left(p-\Re(\Phi)q\right)\\
    &= - q^\trans \left(\Im (\Phi)+\Re(\Phi)\Im(\Phi)^{-1} \Re(\Phi)\right)q + 2 q^\trans \Re(\Phi) \Im(\Phi)^{-1} p  -  p^\trans \Im(\Phi)^{-1} p \\
    &= -\begin{bmatrix}
      q^\trans & p^\trans 
    \end{bmatrix}  \begin{bmatrix}
      \Im (\Phi) + \Re(\Phi)\Im(\Phi)^{-1} \Re(\Phi)  & -\Re(\Phi)\Im(\Phi)^{-1} \\ -\Im(\Phi)^{-1} \Re(\Phi) & \Im(\Phi)^{-1}
    \end{bmatrix}\begin{bmatrix}
      q \\ p 
    \end{bmatrix} 
  \end{align}
  
  Setting
  \[\Sigma\coloneqq \begin{bmatrix}
    \Im (\Phi) + \Re(\Phi)\Im(\Phi)^{-1} \Re(\Phi)  & -\Re(\Phi)\Im(\Phi)^{-1} \\ -\Im(\Phi)^{-1} \Re(\Phi) & \Im(\Phi)^{-1}
  \end{bmatrix}\]
  we get:
  
  \[W(q,p)= \frac{1}{\pi^n } \exp\left(-\begin{bmatrix}
      q^\trans & p^\trans
    \end{bmatrix}  \Sigma \begin{bmatrix}
    q \\ p
  \end{bmatrix}\right)\]

We still have to prove that $\Sigma$ is a covariance matrix. Notice that by definition $\Sigma^\trans = \Sigma $ and that the Shur complement formula directly gives $\det(\Sigma)=1$. For positivity, setting:

\[S= \begin{bmatrix}
  \Im(\phi)^{\frac{1}{2}} & -\Re(\phi)\Im(\phi)^{\frac{-1}{2}} \\ 0 & \Im(\phi)^{\frac{-1}{2}} 
\end{bmatrix}\]

We get: $S S^\trans = \Sigma \succeq 0 $. Finally, setting

\[T= \begin{bmatrix}
  0 & \Im(\Phi)^{\frac{1}{2}} + i\Re(\Phi)\Im(\Phi)^{\frac{-1}{2}} \\ 0 & -i\Im(\Phi)^{\frac{-1}{2}}
\end{bmatrix}\]

We get: $TT^\dagger = \Sigma + i \Omega \succeq 0 $.

Conversely, given a covariance matrix $\Sigma$ such that $\Sigma^\trans = \Sigma $, $\det(\Sigma)=1$, $\Sigma \geq 0 $ and $\Sigma + i \Omega \geq 0$, Williamson theorem gives us a symplectic matrix $S$  and a positive diagonal $\Lambda$ such that:

\[S \begin{bmatrix}
  \Lambda & 0 \\ 0 & \Lambda
\end{bmatrix}S^\trans = \Sigma\]
So:

\[\Sigma + i \Omega = S \begin{bmatrix}
  \Lambda & 0 \\ 0 & \Lambda
\end{bmatrix}S^\trans + iS \Omega S^\trans = S \begin{bmatrix}
  \Lambda & iI \\ -iI & \Lambda
\end{bmatrix}S^\trans \]

This matrix being positive semi-definite  implies that $\lambda \geq 1 $   (take any principal minor of size two) for all symplectic eigenvalues   $\lambda$. But $\det(\Sigma)=1$ implies that $\lambda = 1 $ such that $\Lambda= I $. Thus $\Sigma= SS^\trans $. Any symplectic matrix $S$  admits a unique decomposition of the following form (see \cite{de2006symplectic}):

\[S= \begin{bmatrix}
  I & B \\ 0 & I
\end{bmatrix}\begin{bmatrix}
  A^{\frac{1}{2}} & 0 \\ 0 & A^{\frac{-1}{2}}
\end{bmatrix}O\]

where $B$ is symmetric and $A$ is symmetric positive and $O$ is symplectic orthogonal. Thus $\Sigma$ is of the form:

\[\Sigma = \begin{bmatrix}
  A + BA^{-1} B & BA^{-1} \\ A^{-1}B & A^{-1}
\end{bmatrix}\]

Setting $\Phi= -B + iA $we then see that there is a bijection between covariance matrices $\Sigma$  and matrices $\Phi$ such that $\Im(\Phi)\succ 0$. The bijection is given by:

\[\Sigma= \begin{bmatrix}
  \Im (\Phi) + \Re(\Phi)\Im(\Phi)^{-1} \Re(\Phi)  & -\Re(\Phi)\Im(\Phi)^{-1} \\ -\Im(\Phi)^{-1} \Re(\Phi) & \Im(\Phi)^{-1}
\end{bmatrix}\]

And if 

 \[\Sigma= \begin{bmatrix}
  A & B \\ B^\trans & C
\end{bmatrix}\]

then \(\Phi = -BC^{-1} + i C^{-1}\).

\end{proof}
%
%
%



\begin{lemma}\label{lm:color_void}
  Given any symmetric matrix $X$ we have: \[\input{SympGauss/colorvoid.tikz}\]
\end{lemma}
\begin{proof}
  First remark that the composition of those three symplectomorphisms is a symplectic rotation:
  
  \[\begin{bmatrix}
    I & 0 \\ -X & I 
  \end{bmatrix}\begin{bmatrix}
I & X (I+X^2)^{-1} \\ 0 & I 
\end{bmatrix}\begin{bmatrix}
(I+X^2)^{-\frac{1}{2}} & 0 \\ 0 & (I+X^2)^{\frac{1}{2}}
\end{bmatrix}=\begin{bmatrix}
(I+X^2)^{\frac{1}{2}} & X(I+X^2)^{\frac{1}{2}} \\-X(I+X^2)^{\frac{1}{2}} & (I+X^2)^{\frac{1}{2}}
\end{bmatrix}\]
  
Thus we have: \[\input{SympGauss/colorvoid1.tikz}\]
  
And then: \[\input{SympGauss/colorvoid.tikz}\]
  
\end{proof}

\begin{lemma}\label{lm:sympgauss_to_NF}
	Any diagram in  \(\GQGA\)  can be reduced to the pseudo normal form:
	
	\input{SympGauss/Nform.tikz}
	
	where \(S\) is a real affine symplectomorphism.
	
\end{lemma}
\begin{proof}

  We start by moving all vacuums to the left which yields a diagram in the form:
  
  \[\input{SympGauss/NForm1.tikz}\]
  
  where $L$ is a diagram in $\GSA[\R]$, putting $L$ in graph-like form as in \cite{gsa}, and then applying Lemma \ref{lm:color_void}, we get:
  
  \[\input{SympGauss/NForm2.tikz}\]
  
  Where $Y$ is invertible. Thus without loss of generality we can restrict  to diagrams of the following form after applying singular value decomposition on $C'$:
  
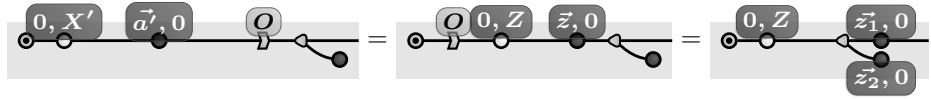
  \[\input{SympGauss/NForm3.tikz}\]
  
  Where $O$ is a rotation. Without loss of generality we now have to reduce a diagram of the form:
  
  \[\input{SympGauss/NForm4.tikz}\]
  
  Finally, applying singular value decomposition on $T$ gives:
  
  \[\input{SympGauss/NForm6.tikz}\]
  
  Which is either zero or in the desired form.
\end{proof}

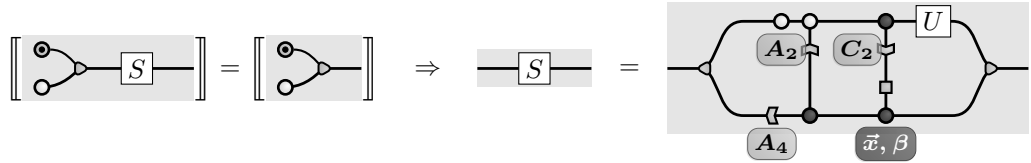
\begin{lemma}\label{lm:sympgauss_uniq}
	\[\interp{\input{SympGauss/uniqueness01.tikz}}=\interp{\input{SympGauss/uniqueness02.tikz}} \quad \Rightarrow \quad \input{SympGauss/uniqueness.tikz}\]
	Where \(\beta\) is symmetric, \(U\) is a symplectic rotation, and \(A_{4}\) is invertible, furthermore the number of vacuum state generators must be the same on both sides.
\end{lemma}
\begin{proof}
  In the semantics, the equality:
  
  \[\interp{\input{SympGauss/uniqueness01.tikz}}=\interp{\input{SympGauss/uniqueness02.tikz}}\]

  implies that that affine subspaces on the left contains $0$, thus the affine shift on the left must be in the linear space on the right:
  \[\begin{bmatrix}
    x
  \end{bmatrix} \in \ker \begin{bmatrix}[c|c]
    I & iJ_k 
  \end{bmatrix}\]
  so that
  \[x= \begin{bmatrix}
    0 \\ 0 \\ 0 \\ x_4
  \end{bmatrix}\]
  
  Now that the affine shift has been taken care of  we end up with the equality:
  
  \[\ker \begin{bmatrix}[cc|cc]
    I & 0 & iJ_r & 0 \\ 0 & 0 & 0 & I
  \end{bmatrix}S = \ker \begin{bmatrix}[cc|cc]
    I & 0 & iI_k & 0 \\ 0 & 0 & 0 & I
  \end{bmatrix}\]
  
  For a real linear symplectomorphism $S$. In other words, such that there exists real matrices $A,B,C,D,X$ and $Y$, with $X+iY$ invertible and 
  \[\begin{bmatrix}
    A & B \\ C & D
  \end{bmatrix}\]
  a symplectomorphism, such that:
  
  \[\begin{bmatrix}[c|c] I & iJ_r \end{bmatrix}\begin{bmatrix}
    A & B \\ C & D
  \end{bmatrix}=\left(X+iY\right) \begin{bmatrix}[c|c] I & iJ_k \end{bmatrix}\]
  
  This gives: 
  
  \[\begin{bmatrix}[c|c]
    A+ iJ_r C &   B+ iJ_r D
  \end{bmatrix} = \begin{bmatrix}[c|c]
    X +iY & -Y\Phi + i X\Phi \end{bmatrix}  \]
  
  and then by identifying the real and imaginary parts:

  \[\begin{cases}
    A = X \\
    J_r C = Y \\
    B = -YI_k \\
    J_r D = XI_k
  \end{cases}\]
  
  we can eliminate $X$ and $Y$:
  
  \[\begin{cases}
    B = -J_r C I_k \\
    J_r D  = AI_k
  \end{cases}\]
  
  To this we can furthermore add the symplectic condition giving:
  
  \[\begin{cases}
    B = -J_r C J_k \\
    J_r D  = AJ_k\\
    B^\trans D= D^\trans B \\
    A^\trans C = C^\trans A \\
    I = A^\trans D - C^\trans B
  \end{cases}\]
  
  Notice that the third equation follows from the others:
  
  \[B^\trans D = - J_k C^\trans J_r D = - J_k C^\trans A J_k = -J_k A^\trans C  J_k = - D^\trans J_r C J_k = - D^\trans B. \]
  
  We can now eliminate $B$ and get:
  
  \[\begin{cases}
    J_r D  = AJ_k\\
    A^\trans C = C^\trans A \\
    I = A^\trans D + C^\trans J_r C J_k
  \end{cases}\]
  
  We can now divide our matrices in blocks:
  
  \[A\coloneqq \begin{bmatrix}
    A_1 & A_2 \\ A_3 & A_4
  \end{bmatrix} \qquad C\coloneqq \begin{bmatrix}
    C_1 & C_2 \\ C_3 & C_4
  \end{bmatrix}\qquad D\coloneqq \begin{bmatrix}
    D_1 & D_2 \\ D_3 & D_4
  \end{bmatrix}\]
  
  The first  constraint then gives:
  
  \[\begin{bmatrix}
    I & 0 \\ 0 & 0
  \end{bmatrix}\begin{bmatrix}
    D_1 & D_2 \\ D_3 & D_4
  \end{bmatrix}=\begin{bmatrix}
    A_1 & A_2 \\ A_3 & A_4
  \end{bmatrix}\begin{bmatrix}
    I & 0 \\ 0 & 0
  \end{bmatrix}\]
  
  Implying that $D_2 = 0 $, $A_3 = 0$, and $A_1 = D_1 $. The third  constraint gives:
  
  \[\begin{bmatrix}
    I & 0 \\ 0 & I
  \end{bmatrix}=\begin{bmatrix}
    A_{1}^\trans & 0 \\ A_{2}^\trans & A_{4}^\trans
  \end{bmatrix}\begin{bmatrix}
    A_1 & 0 \\ D_3 & D_4
  \end{bmatrix}+\begin{bmatrix}
    C_{1}^\trans & C_{3}^\trans \\ C_{2}^\trans & C_{4}^\trans
  \end{bmatrix}\begin{bmatrix}
    I & 0 \\ 0 & 0
  \end{bmatrix}\begin{bmatrix}
    C_1 & C_2 \\ C_3 & C_4
  \end{bmatrix}\begin{bmatrix}
    I & 0 \\ 0 & 0
  \end{bmatrix}\]
  
  \[\begin{bmatrix}
    I & 0 \\ 0 & I
  \end{bmatrix}=\begin{bmatrix}
    A_{1}^\trans A_1 & 0 \\ A_{2}^\trans A_1 + A_{4}^\trans D_3 & A_{4}^\trans D_{4}
  \end{bmatrix}+\begin{bmatrix}
    C_{1}^\trans C_1 & 0 \\ C_{2}^\trans C_1 & 0
  \end{bmatrix}\]
  
  This implies that $A_{4}^\trans D_4 = I$, so the  $(n-r) \times (n-k)$ matrix $D_4 $  is injective. Therefore, $(n-k) \leq (n-r)$, and thus $r\leq k$. Since the role of $k$ and $r$ is symmetric up to replacing $S$ by its inverse, we finally get that $r=k$, and that $A_{4}^\trans $ is invertible. We can now eliminate $D_4 = A_{4}^{-\trans}$  and reformulate the third  constraint as:
  
  \[\begin{cases}
    I=A_{1}^\trans A_{1}+ C_{1}^\trans C_{1} \\
    D_3 = -A_{4}^{-\trans}A_{2}^\trans A_1 - A_{4}^{-\trans}C_{2}^\trans C_1 \\
    
  \end{cases}\]
  
  Allowing to eliminate $D_3$. The second  constraint gives us:
  
  \[\begin{bmatrix}
    A_{1}^{\trans} & 0 \\ A_{2}^{\trans} & A_{4}^{\trans}
  \end{bmatrix}\begin{bmatrix}
    C_1 & C_2 \\ C_3 & C_4
  \end{bmatrix}=\begin{bmatrix}
    C_{1}^\trans & C_{3}^\trans \\ C_{2}^\trans & C_{4}^\trans
  \end{bmatrix}\begin{bmatrix}
    A_{1} & A_{2} \\ 0 & A_{4}
  \end{bmatrix}\]
  
  In other words: 
  
  \[\begin{cases}
    A_{1}^\trans C_{1}= C_{1}^\trans A_{1}   \\
    A_{2}^\trans C_{1} + A_{4}^\trans C_{3} =  C_{2}^\trans A_{1}\\ 
    A_{2}^\trans C_{2} + A_{4}^\trans C_{4} =  C_{2}^\trans A_{2} + C_{4}^\trans A_{4}\\
  \end{cases}\]
  
  We can thus eliminate $C_{3} = A_{4}^{-\trans} C_{2}^\trans A_{1} - A_{4}^{-\trans} A_{2}^\trans C_{1}$.
  
  If we gather the constraints we have so far, we have only $A_1$, $A_2$, $A_4$, $C_1$, $C_2$ and $C_4$ as degree of freedom satisfying:
  
  \[\begin{cases}
    I=A_{1}^\trans A_{1}+ C_{1}^\trans C_{1} \\
    A_{1}^\trans C_{1}= C_{1}^\trans A_{1}   \\
    A_{2}^\trans C_{2} + A_{4}^\trans C_{4} =  C_{2}^\trans A_{2} + C_{4}^\trans A_{4}
  \end{cases}\]
  
  We get the following symplectomorphism:
  
  \[\begin{bmatrix}
    A_{1} & A_{2} & -C_1 & 0\\
    0 & A_{4} & 0 & 0\\
    C_1 & C_2 & A_{1} & 0\\
    A_{4}^{-\trans} C_{2}^\trans A_{1} - A_{4}^{-\trans} A_{2}^\trans C_{1} & C_4 & -A_{4}^{-\trans}A_{2}^\trans A_1 - A_{4}^{-\trans}C_{2}^\trans C_1 & A_{4}^{-\trans}
  \end{bmatrix}\]
  
  We can now decompose it. First $A_4 \in \Gl[n-r]$ only acts on the lower part:
  
  \[\begin{bmatrix}
    I & 0 & 0 & 0\\
    0 & A_{4} & 0 & 0\\
    0&0& I & 0\\
    0& 0 & 0 & A_{4}^{-\trans}
  \end{bmatrix}\begin{bmatrix}
    A_{1} & A_{2} & -C_1 & 0\\
    0 & I & 0 & 0\\
    C_1 & C_2 & A_{1} & 0\\
    C_{2}^\trans A_{1} -  A_{2}^\trans C_{1} & A_{4}^{\trans} C_4 & -A_{2}^\trans A_1 - C_{2}^\trans C_1 & I
  \end{bmatrix}\]
  
  Then, we can see the symplectic rotation
  \[\begin{bmatrix}
    A_{1} & -C_1 \\ C_1 & A_{1}
  \end{bmatrix}\]
   acting on the upper part:
  
  \[\begin{bmatrix}
    I & 0 & 0 & 0\\
    0 & A_{4} & 0 & 0\\
    0&0& I & 0\\
    0& 0 & 0 & A_{4}^{-\trans}
  \end{bmatrix}\begin{bmatrix}
    I & A_{2} & 0 & 0\\
    0 & I & 0 & 0\\
    0 & C_2 & I & 0\\
    C_{2}^\trans  & A_{4}^{\trans} C_4 & -A_{2}^\trans  & I
  \end{bmatrix}\begin{bmatrix}
    A_{1} & 0 & -C_1  & 0\\
    0 & I & 0 & 0\\
    C_1 &0& A_{1} & 0\\
    0& 0 & 0 & I
  \end{bmatrix}\]
  
  Finally, denoting $\beta\coloneqq A_{4}^{\trans} C_4 + A_{2}^{\trans} C_2  $, we know that $\beta$ is symmetric, thus we end up with a symplectomorphism of the form:
  
  \[\begin{bmatrix}
    I & 0 & 0 & 0\\
    0 & A_{4} & 0 & 0\\
    0&0& I & 0\\
    0& 0 & 0 & A_{4}^{-\trans}
  \end{bmatrix}\begin{bmatrix}
    I & A_{2} & 0 & 0\\
    0 & I & 0 & 0\\
    0 & 0 & I & 0\\
    0  & 0 & -A_{2}^\trans  & I
  \end{bmatrix}\begin{bmatrix}
    I & 0  & 0 & 0 \\
    0 & I & 0 & 0\\
    0&C_2 &I & 0\\
    C_{2}^\trans &\beta & 0 & I
  \end{bmatrix}\begin{bmatrix}
    A_{1} & 0 & -C_1  & 0\\
    0 & I & 0 & 0\\
    C_1 &0& A_{1} & 0\\
    0& 0 & 0 & I
  \end{bmatrix}\]
  
  Then in $\GSA[\R]$: 
  
  \input{SympGauss/uniqueness.tikz}
  
  Where $U$ is a symplectic rotation.
\end{proof}

\begin{proof}[Proof of theorem~\ref{thm:sympgauss_completeness}]
  
  We define an interpretation prop morphism $\interp{\_}:\GSA[\R]^+ \to \ALR[\C]^+$ by sending all generator from $\GSA[\R]$ to the corresponding Lagrangian subspaces and vacuum generators to the subspace corresponding to a white node with phase $(0,i)$.
  
  By lemma~\ref{lm:sympgauss_to_NF} we know that any diagram in $\GSA[\R]^+$ can be put in pseudo normal form, but by proposition~\ref{prop:positivechar}, we see that replacing the vacuum generator by a white node with phase $(0,i)$, we can represent all positive subspaces. In other words the interpretation of $\GSA[\R]^+ $ into positive subspaces is surjective. It only remains to prove that this interpretation is also injective.
  
  Let $X $ and $Y $ be two diagrams in $\GSA[\R]^+ $ such that $\interp{X }=\interp{Y }$. Using lemma~\ref{lm:sympgauss_to_NF} we can assume that both diagrams are in pseudo normal form. We have:
  
  \[\interp{\input{SympGauss/comp1.tikz}}=\interp{\input{SympGauss/comp2.tikz}}\]
  
  And then, by lemma~\ref{lm:sympgauss_uniq} it follows that $S= S_x (S_y )^{-1}$ is of the form:
  
  \[\input{SympGauss/uniqueness.tikz}\]

  with $\beta$ symmetric, $U$ symplectic rotation and $A_4 $ invertible. Therefore:
  
  \[\input{SympGauss/comp5.tikz}\] 
\end{proof}

\end{document}